\documentclass[journal]{IEEEtran} 

\usepackage{psfrag}
\usepackage{amsmath,amssymb,graphicx}
\usepackage{cite, amsbsy, amsfonts}

\begin{document}

\newcommand{\girth}{g}
\newcommand{\weight}{\textbf{w}}
\newcommand{\parChkMat}{\textbf{H}}
\newcommand{\baseMat}{\textbf{B}}
\newcommand{\shiftMat}[1]{\textbf{I}_{#1}}
\newcommand{\circThree}[3]{\begin{array}{ccc} #1 & #2 & #3
    \\ #3 & #1 & #2\\ #2 & #3 & #1\end{array}}
\newcommand{\x}[1]{x_{[#1]}}
\newcommand{\p}[1]{p_{[#1]}}
\newcommand{\Jrow}[1]{J_{[#1]}}
\newcommand{\Lcol}[1]{L_{[#1]}}
\newcommand{\baseSym}{c}
\newcommand{\sym}[3]{\baseSym_{#1}[#2,#3]}
\newcommand{\symThree}[4]{\baseSym_{#1}[#2,#3,#4]}
\newcommand{\symTwo}[2]{\baseSym[#1,#2]}
\newcommand{\symTil}[3]{\tilde{\baseSym}_{#1}[#2,#3]}

\newcommand{\indSym}{s}
\newcommand{\indSymBf}{{\bf \indSym}}
\newcommand{\ind}[1]{\indSym_{#1}}
\newcommand{\indTil}[1]{\tilde{\indSym}_{#1}}
\newcommand{\bfInd}[2]{{\bf \indSym}[#1,#2]}
\newcommand{\currLevel}{\tilde{k}}
\newcommand{\kronecker}{\raisebox{1pt}{\ensuremath{\:\otimes\:}}} 
\newcommand{\conMat}{{\bf T}}
\newcommand{\connectMat}[1]{\conMat_{[#1]}}
\newcommand{\connectMatEl}[2]{\conMat_{[#1]}[#2]}
\newcommand{\len}{\Lambda}
\newcommand{\coeffSet}{{\cal C}}

\newcommand{\path}{{\cal P}}
\newcommand{\ordSet}{{\cal O}}
\newcommand{\pathCoeff}{{\cal S}}
\newcommand{\move}[2]{\Delta_{#1}{[#2]}}

\newcommand{\setTree}{{\bf \cal T}}
\newcommand{\N}[3]{{\cal N}_{#1}(#2, #3)}
\newcommand{\costVecSym}{\Gamma}
\newcommand{\costSym}{\gamma}
\newcommand{\costB}[1]{\costVecSym_{#1}}
\newcommand{\costBneg}[1]{\costVecSym_{#1}^{-}}
\newcommand{\cost}[1]{\costSym_{#1}}
\newcommand{\costTil}[1]{\tilde{\costVecSym}_{#1}}

\newcommand{\costVecSymHier}[1]{\Gamma[#1]}
\newcommand{\costBHier}[1]{\costVecSym_{#1}}
\newcommand{\baseMatHier}[1]{\baseMat[#1]}

\newcommand{\wt}{{\rm wt}}
\newcommand{\nullTree}{\ast}

\newcommand{\gainMat}{\textbf{G}}

\newtheorem{defn}{Definition}
\newtheorem{example}{Example}
\newtheorem{algorithm}{Algorithm}
\newtheorem{thm}{Theorem}
\newtheorem{lemma}{Lemma}
\newtheorem{cor}{Corollary}

\newcommand{\tightArray}[2]{\left[ \! \begin{array}{c} \!\! #1 \!\! \\ \!\!
      #2 \!\! \end{array} \! \right]}

\title{Hierarchical and High-Girth QC LDPC Codes}

\author{ Yige~Wang,~\IEEEmembership{Member,~IEEE,}
  Stark~C.~Draper,~\IEEEmembership{Member,~IEEE,}
  Jonathan~S.~Yedidia,~\IEEEmembership{Senior~Member,~IEEE}
  \thanks{Y.~Wang is with the Mitsubishi Electric Research
    Laboratories, Cambridge, MA 02139 (yigewang@ieee.org).}
  \thanks{S.~C.~Draper is with the Dept.~of Electrical and Computer
    Engineering, University of Wisconsin, Madison, WI 53706
    (sdraper@ece.wisc.edu).}  \thanks{J.~S.~Yedidia is with the
    Mitsubishi Electric Research Laboratories, Cambridge, MA 02139
    (yedidia@merl.com).}}

\maketitle

\begin{abstract}
We present a general approach to designing capacity-approaching
high-girth low-density parity-check (LDPC) codes that are friendly to
hardware implementation.  Our methodology starts by defining a new
class of {\em hierarchical} quasi-cyclic (HQC) LDPC codes that
generalizes the structure of quasi-cyclic (QC) LDPC codes.  Whereas
the parity check matrices of QC LDPC codes are composed of circulant
sub-matrices, those of HQC LDPC codes are composed of a hierarchy of
circulant sub-matrices that are in turn constructed from circulant
sub-matrices, and so on, through some number of levels.  We show how
to map any class of codes defined using a protograph into a family of
HQC LDPC codes.  Next, we present a girth-maximizing algorithm that
optimizes the degrees of freedom within the family of codes to yield a
high-girth HQC LDPC code.  Finally, we discuss how certain
characteristics of a code protograph will lead to inevitable short
cycles, and show that these short cycles can be eliminated using a
``squashing'' procedure that results in a high-girth QC LDPC code,
although not a hierarchical one.  We illustrate our approach with
designed examples of girth-10 QC LDPC codes obtained from protographs
of one-sided spatially-coupled codes.
\end{abstract}

\section{Introduction}

Two broad classes of methods have emerged for the construction of
low-density parity-check (LDPC) codes~\cite{gallager, RyanLin}.  One
set of methods is based on highly random graph constructions, while
the other is based on structured algebraic constructions.  It is now
well-known that random constructions (see, e.g.,
\cite{Mackay,RichardCap,Luby,chung,ModernCodingTheory}) can produce
LDPC codes that closely approach the Shannon capacity.  However,
highly random constructions are not easy to implement in hardware as
the irregular connections between check and variable nodes in the code
graph imply high wiring complexity.  In actual implementations, more
structured constructions have been strongly preferred because they
result in much more practical wiring and more straightforward
parallelism in the decoders.

Quasi-cyclic LDPC (QC LDPC) codes are a particularly practical and
widely-used class of structured LDPC codes. These codes have a parity
check matrix which is broken into sub-matrices that have a circulant
structure. QC LDPC codes are featured in a variety of communications
system standards, such as IEEE 802.16e \cite{80216e},
DVB-S2\cite{dvbs2} and 802.11 \cite{80211}. In view of their
practicality, we focus in this paper on the design of QC LDPC codes that have good
decoding performance.

For nearly any application, it is important to optimize decoding
performance in the ``water-fall'' regime where the signal-to-noise
ratio (SNR) is relatively low. The standard way to do that for
irregular random constructions is to use ``density-evolution'' or
``EXIT chart'' techniques to obtain the degree distribution that
optimizes the code threshold in the asymptotic limit of long block
lengths \cite{ModernCodingTheory}. These
techniques can also be adapted to QC LDPC codes \cite{Liva}.

However, for some applications, optimizing water-fall performance is not sufficient, and
one must also avoid the ``error floors'' that plague many LDPC codes in the
higher SNR regime. An ``error floor'' 
in the performance curve means that the
decoding failure rate does not continue to decrease rapidly as the SNR increases. 
Eliminating or lowering error floors is 
particularly important for applications that have extreme reliability
demands, including magnetic recording and fiber-optic communication
systems.

In the past, QC LDPC codes have been constructed based on a wide
variety of mathematical ideas, including finite geometries, finite
fields, and combinatorial designs \cite{RyanLin, KouLinIT, Fan,
  TannerQuasi, Vasic02, LinQC, Sridhara, Okamura}.  Recently, there
has also been great interest in the class of ``convolutional''
\cite{TannerQC, Lentmaier 2010} or ``spatially-coupled''
\cite{Kudekar} LDPC codes.  These codes have much more structured than
traditional random constructions.  They have also been shown, using
density evolution techniques, to approach Shannon capacity closely, or
even provably to achieve it on the binary erasure channel (BEC)
\cite{Kudekar}.  These codes are significant here, because they can be
implemented using quasi-cyclic constructions, and they should thus be
able to achieve very good performance while retaining the practicality
of other structured QC LDPC codes.  In this paper, we will focus on
the design of QC LDPC codes based on structures that let them perform
near the Shannon limit in the waterfall regime (such as spatially
coupled codes) but we also aim for excellent error floor performance.

Error floor issues for LDPC codes are investigated in
\cite{ErrorFloorRi, Vasic}, which shows that error floors in belief
propagation (BP) LDPC decoders are generally caused by ``trapping
sets.'' A trapping set is a set of a small number of bits that
reinforce each other in their incorrect beliefs.  Trapping sets of
bits are invariably arranged in clustered short cycles in a code's
Tanner graph~\cite{TannerGraph}.  Therefore, one way to try to remove
trapping sets is to design the code's Tanner graph carefully so that
the dangerous clusters of short cycles do not exist.

An alternative, and at least conceptually simpler approach, is to
design codes with larger girths---the ``girth'' of a code is the
length of the shortest cycle in the code graph.  By removing short
cycles, we remove large swaths of potentially dangerous configurations
of cycles and, at one fell swoop, hopefully lower the error floor.
Motivated by this idea, in this paper, we focus on optimizing the
girth of QC LDPC codes that have also been optimized for waterfall
performance. In this way we hope to design a practical code that
simultaneously has good waterfall and error floor performance.

There has been considerable work on optimizing girth in LDPC codes.
In~\cite{HuPEG} a progressive-edge growth (PEG) algorithm is proposed
for random LDPC codes. The PEG technique is generalized to QC LDPC
codes in~\cite{LiKumar}.  Another approach to optimizing the girth of
QC LDPC codes is studied in~\cite{MarcQC}, where high-girth QC LDPC
codes are obtained using a random ``guess-and-test''
algorithm. Shortened array codes with large girth are proposed
in~\cite{Milenkovic}. However, after shortening, the resulting codes
no longer have a quasi-cyclic structure. In~\cite{Bocharova}, another
class of large girth QC LDPC codes is designed, where the methodology
is mainly for regular LDPC codes.

In this paper, we propose a hill-climbing search algorithm for optimizing girth
 that is more efficient than previous techniques. The hill-climbing algorithm
greedily adjusts an initial QC LDPC code (hopefully) to find a code of
short length that meets the specified code and girth parameters.
Since the algorithm is greedy, it can get stuck in local minima.
However, given a set of parameters, the algorithm finds QC LDPC codes
of shorter length and in less time than guess-and-test or PEG.

Codes with good water-fall performance inevitably have some
irregularity in the degree-distribution of the factor or variable
nodes of the code graph. For the case of QC LDPC codes, these
irregular distributions are most easily described in terms of
``protographs'' \cite{Thorpe}.  Protographs are variants of Tanner
graphs where bits and checks of the same type are represented by a
single ``proto-bit'' or ``proto-check.'' In the case of QC LDPC codes,
proto-bits can, for example, represent sets of bits belonging to the
same circulant sub-matrices.

The protographs that arise in codes that have been optimized for
waterfall performance typically have some pairs of proto-bits and
proto-checks that are connected by multiple edges. A straightforward
way to handle this would be to use QC LDPC codes where the circulant
matrices had rows and columns of weight greater than one. However, as
we shall see, this direct approach inevitably introduces short cycles
into the graph.

The tricky problem of creating QC LDPC codes with good girth and that
correspond to protographs optimized for waterfall performance is
solved in this paper by a somewhat complicated procedure.  First, we
need to introduce a new family of generalized QC LDPC codes, which we
call ``hierarchical'' QC LDPC (HQC LDPC) codes.  The parity check
matrices of these hierarchical codes consist of circulant
sub-matrices, which in turn consist of of circulant sub-sub-matrices,
and so on for multiple ``levels.''  We show that we can directly
transform any protograph with multiple edges between proto-checks and
proto-bits into a {\em two-level} HQC LDPC code with circulant
matrices with higher weight at second level.

It turns out that many different hierarchical QC LDPC codes correspond
to a particular protograph, and thus many degrees of freedom exist
following the ``direct'' transformation.  We use our hill-climbing
algorithm to choose from this family to get rid of as many short
cycles as possible.  However, HQC LDPC codes with weights greater than
one at higher levels will also automatically have some short cycles,
just as non-hierarchical QC LDPCs do.  Our hill-climbing algorithm can
do nothing about these ``inevitable'' cycles but it can, hopefully,
eliminate all short non-inevitable cycles.  To get ride of the
inevitable cycles we introduce a ``squashing'' procedure. The squashing
step destroys the hierarchical structure of the code, but the squashed
code nevertheless remains a QC LDPC code that corresponds to the
desired protograph.  The squashing procedure is computationally
trivial.  This makes the two-step procedure -- first HQC LDPC then
squash into a QC LDPC code -- much more computationally efficient than
directly applying our hill-climbing procedure to maximize the girth of
higher-weight QC LDPC codes with the desired protograph structure.
Thus hierarchical QC LDPC codes are a necessary intermediate stage in
the design of practical QC LDPC codes that will simultaneously have
good waterfall and error floor performance.

The rest of the paper will explain in much more detail the ideas
outlined above. We begin in Section~\ref{sec.defQC_LDPC} by reviewing
the standard construction of QC LDPC codes in terms of their parity
check matrices. Then in Section~\ref{sec.protographs}, we review the
standard Tanner graph representation of LDPC codes and the
``protograph'' representation of structured codes. In
Section~\ref{sec.findCycles} we show how short cycles can be
identified from the parity check matrix of a QC LDPC code.  We discuss
why the most direct transformation of the protographs of interesting
LDPC codes, such as spatially coupled codes, into QC LDPC codes will
lead to inevitable short cycles in the Tanner graph of the resulting
codes. This leads us to the heart of our paper, where we introduce
hierarchical QC LDPC codes that can be used to solve the problem of
inevitable short cycles. In Section~\ref{sec.defHierQC_LDPC} we
introduce the most general form of HQC LDPC codes and show that they
can be described both in terms of a multi-variate polynomial parity
check matrix in multiple variables and in terms of a tree structure.
In Section~\ref{sec.condForCycles}, we explain how to find cycles in
the Tanner graphs of HQC LDPC codes.  In Section~\ref{sec.girthMax} we
describe our hill-climbing algorithm for finding high girth QC LDPC
codes and HQC LDPC codes.  In Section~\ref{sec.pipeline}, we discuss
restricted two-level HQC LDPC codes, the direct transformation of
protographs into such codes, and the ``squashing'' procedure that
efficiently eliminates inevitable cycles.  Finally, in
Section~\ref{sec.results}, we exhibit high-girth QC LDPC codes that
simultaneously have good waterfall behavior (because they are
spatially-coupled codes) and have good error-floor behavior resulting
from their high girth (which in turn is a result of the fact that they
are squashed versions of HQC LDPC codes).  Many details and lemmas are
deferred to the appendices.

\section{Quasi-cyclic LDPC codes}
\label{sec.defQC_LDPC}

We begin by reviewing the construction of standard quasi-cyclic
low-density parity-check (QC LDPC) codes as previously described in the
literature \cite{RyanLin}. In section~\ref{sec.defHierQC_LDPC} we will
generalize these codes and introduce a novel {\em hierarchical} family
of QC LDPC codes.

%%%%%%%%%%%%%%%%%%%%%

Before considering the general case of standard QC LDPC codes, it is
helpful to start with an important special case, that we will call
``weight-I $(J,L)$ regular'' QC LDPC codes. The parity check matrix of
these codes consists of $J \cdot L$ sub-matrices, each of which is a
$p \times p$ circulant permutation matrix.

Let $\shiftMat{i, p}$ denote the circulant permutation matrix, or
``cyclic shift matrix,'' obtained by cyclically right-shifting a $p
\times p$ identity matrix by $i$ positions, where $0 \le i \le p-1$;
$\shiftMat{0,p}$ is thus the $p \times p$ identity matrix.  We often
suppress the dependence on $p$, writing $\shiftMat{i}$ instead of
$\shiftMat{i,p}$.  As an example, if $p = 4$, then
\begin{equation*}
\shiftMat{1} = \left[ \begin{array}{cccc} 0 & 1 & 0 & 0\\ 0 & 0 & 1
    & 0\\ 0 & 0 & 0 & 1\\ 1 & 0 & 0 & 0 \end{array} \right].
\end{equation*}

We can write the parity check matrix of a weight-I $(J,L)$ regular QC
LDPC code using $J$ rows and $L$ columns of $p \times p$ cyclic shift
sub-matrices:
\begin{equation}
\parChkMat = \left[
\begin{array}{llll} \shiftMat{i_{1,1}} & \shiftMat{i_{1,2}} & \cdots & \shiftMat{i_{1,L}} \\ 
\shiftMat{i_{2,1}} & \shiftMat{i_{2,2}} & \cdots & \shiftMat{i_{2,
    L}}\\ \vdots &&\ddots& \vdots\\ \shiftMat{i_{J,1}} &
\shiftMat{i_{J,2}} & \cdots & \shiftMat{i_{J,L}} \end{array}
\right]. \label{eq.parChkExAinit}
\end{equation}
The blocklength of such a code is $N=pL$.

Using $\shiftMat{k} = (\shiftMat{1})^k$, we can
re-write~(\ref{eq.parChkExAinit}) as
\begin{equation}
\parChkMat = \left[
\begin{array}{llll} (\shiftMat{1})^{i_{1,1}} & (\shiftMat{1})^{i_{1,2}} & \cdots & (\shiftMat{1})^{i_{1,L}} \\ 
(\shiftMat{1})^{i_{2,1}} & (\shiftMat{1})^{i_{2,2}} & \cdots &
  (\shiftMat{1})^{i_{2, L}}\\ \vdots &&\ddots&
  \vdots\\ (\shiftMat{1})^{i_{J,1}} & (\shiftMat{1})^{i_{J,2}} &
  \cdots & (\shiftMat{1})^{i_{J,L}} \end{array}
\right]. \label{eq.parChkExA}
\end{equation}

We can now abstractly represent $\parChkMat$ as a matrix whose entries
are powers of a dummy variable $x$:
\begin{equation}
\parChkMat(x) = \left[
\begin{array}{cccc} x^{i_{1,1}} & x^{i_{1,2}} & \cdots & x^{i_{1,L}} \\ 
x^{i_{2,1}} & x^{i_{2,2}} & \cdots & x^{i_{2, L}}\\ \vdots &&\ddots&
\vdots\\ x^{i_{J,1}} & x^{i_{J,2}} & \cdots & x^{i_{J,L}} \end{array}
\right]. \label{eq.polyParChkExA}
\end{equation}

The point of all these trivial re-writings will now become clear: we
can generalize such a matrix $\parChkMat(x)$ to a parity check matrix
whose entries are {\em polynomials} in $x$, giving us the {\em
  polynomial} parity check matrix of a standard QC LDPC code:
\begin{equation}
\parChkMat(x) = \left[
\begin{array}{cccc} h_{1,1}(x) & h_{1,2}(x) & \cdots & h_{1,L}(x) \\ 
h_{2,1}(x) & h_{2,2}(x) & \cdots & h_{2,L}(x)\\ \vdots &&\ddots&
\vdots\\ h_{J,1}(x) & h_{J,2}(x) & \cdots & h_{J,L}(x) \end{array}
\right], \label{eq.polyParChkGeneral}
\end{equation}
where 
\begin{equation}
h_{j,l}(x) = \sum_{s = 0}^{p-1} c_s[j,l] x^s
\label{eq.standardCoefficients}
\end{equation}
 for $1 \leq j
\leq J$, $1 \leq l \leq L$.

For {\em binary} QC LDPC codes, which will be our focus for the rest
of this paper, the polynomial coefficients $c_s[j,l]$ must all be 0
or 1.

\begin{example} 
\label{ex.QCLDPC}
 Let ${\cal C}$ be a length-$9$ QC LDPC code described by

\begin{equation}
\parChkMat \! = \!\left[\!\begin{array}{ccc|ccc|ccc}
1&0&0& 1&0&0& 1&0&0 \\
0&1&0& 0&1&0& 0&1&0 \\
0&0&1& 0&0&1& 0&0&1 \\ \hline
0&0&0& 1&0&0& 0&1&1 \\
0&0&0& 0&1&0& 1&0&1 \\
0&0&0& 0&0&1& 1&1&0 \\
\end{array}\!\right]. 
\end{equation}
For this code $J = 2$, $L =3$, and $p = 3$, and $\parChkMat$ can
equivalently be written as
\begin{equation}
\parChkMat = \left[ \begin{array}{ccc} \shiftMat{0} &
  \shiftMat{0} & \shiftMat{0} \\ {\bf 0} & \shiftMat{0} & \shiftMat{1} + \shiftMat{2} 
\end{array} \right].
\end{equation}
The polynomial version of the parity check matrix is
\begin{equation}
\parChkMat(x) = \left[ \begin{array}{ccc} x^0 & x^0 & x^0 \\ 0 & x^0 &
    x^1 + x^2
\end{array} \right] = \left[ \begin{array}{ccc} 1 & 1 & 1 \\ 0 & 1 &
    x^1 + x^2
\end{array} \right]. \label{eq.polyChkMatExA}
\end{equation}
\hfill \QED
\end{example}

In~\cite{vontobelRefs}, Smarandache and Vontobel classified QC LDPC
codes according to the maximum weight among the circulant sub-matrices
in their parity check matrix, or equivalently, according to the
maximum weight of the polynomials in their polynomial parity check
matrix. (The weight of a polynomial as simply the number of non-zero
terms in that polynomial.)  They defined a ``type-$M$'' QC LDPC code
as one for which the maximum weight among all polynomial entries
$h_{j,l}(x)$ in $\parChkMat(x)$ is $M$. We will change their
terminology slightly and call such a code a {\em weight-$M$} QC-LDPC
code

Since ${\rm wt}(h_{2,3}(x)) = 2$ in the code of Example
\ref{ex.QCLDPC}---that is, $h_{2,3}(x) = x^1+x^2$ is a binomial---and
because ${\rm wt}(h_{2,3}(x)) \geq {\rm wt}(h_{j,l}(x))$ for all $1
\leq j \leq J$, $1 \leq l \leq L$, the code in Example~\ref{ex.QCLDPC}
is a weight-II QC LDPC code.

For any QC LDPC code, we define the vector of weight sums
$\sum_{j=1}^{J} {\rm wt}(h_{j,l}(x))$ for $1 \leq l \leq L$, to be the
``column weight sum,'' ${\rm wt}_{\rm col}(\parChkMat(x))$, of
$\parChkMat(x)$.  We define the row weight sum ${\rm wt}_{\rm row}$ of
$\parChkMat(x)$ similarly.  Thus, the code of Example \ref{ex.QCLDPC}
has column and row weight sums
$${\rm wt}_{\rm col}(\parChkMat(x)) = [1 \; 2 \; 3]; \hskip 1.0em {\rm wt}_{\rm
  row}(\parChkMat(x)) = [3 \; 3].
$$

It should now be clear why we previously referred to codes of the form
of equation~(\ref{eq.parChkExAinit}) or (\ref{eq.polyParChkExA}) as
``weight-I'' codes, as all the entries in the polynomial parity check
matrix are monomials. The class of weight-I codes is more general than
that shown in equation~(\ref{eq.parChkExAinit}) though: some of the
cyclic shift sub-matrices could be replaced with all-zeros matrices.

As we often work with weight-I QC LDPC codes, and these codes are
particularly important in practice, we introduce some additional
useful notation for them.  We define the {\em base matrix} of a
weight-I QC LDPC code to be the $J \times L$ matrix of powers
(circulant shifts) that defines the code, i.e.,
$\log_x(\parChkMat(x))$ where logarithms are taken entry-by-entry, and
where we define $\log_x(0)$ to be $-1$, used to indicate an all-zero
sub-matrix.  For example, the base matrix corresponding to the parity
check matrix (\ref{eq.polyParChkExA}) is simply
\begin{eqnarray}\label{DefB}
\baseMat = \left[
  \begin{array}{cccc} i_{1,1} & i_{1,2} & \cdots & i_{1,L} \\ i_{2,1} & i_{2,2} & \cdots & i_{2,L} \\
    \vdots & & \ddots & \vdots \\ i_{J,1} & i_{J,2} & \cdots & i_{J,L}
      \end{array} \right].
\end{eqnarray}

\section{Graphical Representations of QC LDPC Codes}
\label{sec.protographs}

As is very well known, an LDPC code can either be represented by its
parity check matrix $\parChkMat$, or equivalently by its Tanner graph
\cite{TannerGraph}. A Tanner graph for an LDPC code is a bi-partite
graph consisting of ``variable'' nodes representing the codeword bits,
and ``check'' nodes representing the parity checks, where a variable
node is connected to a check node by an edge if and only if the
corresponding entry in $\parChkMat$ is nonzero. The degree of a node
is defined as the number of edges incident to that node.

A ``protograph,'' as introduced by Thorpe in \cite{Thorpe}, is a
template that can be used to derive a class of Tanner graphs.  Each
node in a protograph represents a ``type'' of node in a Tanner graph.
The nodes will all be duplicated $p$ times in the Tanner graph derived
from the protograph.

\begin{figure}
\centering
\includegraphics[width=3.0in]{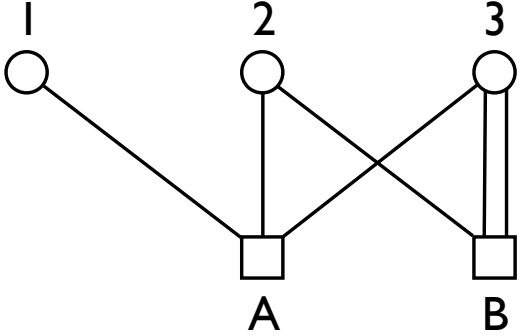}
\caption{A simple protograph with three types of variables and two types of checks.}
\label{fig:proto1}
\end{figure}

As an example, consider Fig.~\ref{fig:proto1}, which shows a simple
example of a protograph that has three types of variable nodes and two
types of check nodes. This protograph tells us that each check of type
A should be connected to one variable of each of the three types, and
each check of type B should be connected to one variable of type 2 and
two variables of type 3.  Similarly, each variable of type 1 should be
connected to one check of type A, and so on.

\begin{figure}
\centering
\includegraphics[width=3.2in]{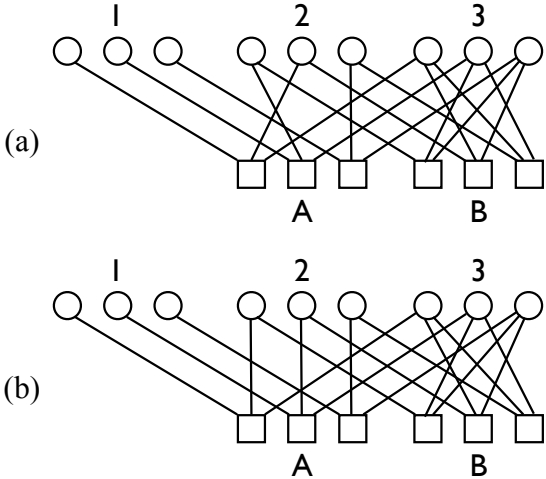}
\caption{Two Tanner graphs corresponding to the protograph shown in Fig.~\ref{fig:proto1}. The Tanner graph in (a) does not have a quasi-cyclic structure; the
one in (b) does, and in fact has the parity check matrix of the QC LDPC code given in 
Example~\ref{ex.QCLDPC}.}
\label{fig:proto2}
\end{figure}

Fig.~\ref{fig:proto2} shows two Tanner graphs derived from the
protograph of Fig. 1, with $p=3$. Note that there are many possible
Tanner graphs that one can construct that correspond to a particular
protograph, and they need not necessarily have a quasi-cyclic
structure. The Tanner graph shown in Fig.~\ref{fig:proto2}~(a) is not
quasi-cyclic. But it is always easy to construct a quasi-cyclic
version of any protograph.

In fact, protographs can equivalently be described by a ``connectivity
matrices.''  A connectivity matrix has a number of rows equal to the
number of types of checks in the protograph and a number of columns
equal to the number of types of variables. Each entry in the
connectivity matrix tells you how many edges there are connecting a
type of check node to a type of variable node in the protograph.  For
example, the connectivity matrix $C$ for the protograph in
Fig.~\ref{fig:proto1} would be
\begin{equation}
C = \left[ \begin{array}{cccc} 
1 & 1 & 1 \\
0 & 1 & 2 \end{array}\right].
\label{eq:Cmatrix}
\end{equation}

To derive a quasi-cyclic parity-check matrix $\parChkMat(x)$ 
from the template specified by a particular protograph, one can
simply replace each entry in the equivalent
connectivity matrix with a polynomial of weight equal
to the entry. We will call this procedure
 a ``direct transformation'' of a protograph into a QC LDPC code.
  
For example, the protograph in Fig.~\ref{fig:proto1} which has the connectivity
matrix $C$ given in (\ref{eq:Cmatrix}),
can be directly transformed into a QC LDPC code with parity check matrix
\begin{equation}
\parChkMat(x) = \left[ \begin{array}{cccc} 
x^a & x^b & x^c \\
0 & x^d & x^e + x^f \end{array}\right],
\end{equation}
where $a$, $b$, $c$, $d$, $e$ and $f$ are integer exponents between $0$ and $p-1$, 
with $e \ne f$.

There are many possible direct transformations of a protograph into a QC LDPC code, 
depending on what exponents one
chooses for the polynomials; 
one particular direct transformation would convert this protograph into the QC LDPC
code with parity check matrix
\begin{equation}
\parChkMat(x) = \left[ \begin{array}{cccc} 
x^0 & x^0 & x^0 \\
0 & x^0 & x^1 + x^2 \end{array}\right].
\end{equation}
which would correspond to the Tanner graph shown in Fig.~\ref{fig:proto2}~(b)
and the code given in Example~\ref{ex.QCLDPC}.

\section{Cycles in QC LDPC codes}
\label{sec.findCycles}

In this section we discuss how to identify cycles in QC LDPC codes
from their parity check matrices.  Each check node in the Tanner graph
of a code corresponds to a row in its parity check matrix, and each
variable node corresponds to a column.  A cycle is a path through
nodes in the Tanner graph, alternating between check and variable
nodes, that starts and ends at the same node.  In terms of the code's
parity check matrix, a cycle can be visualized as a sequence of
alternating vertical horizontal moves through the matrix starting and
ending on the same row of the matrix.  A vertical move (along a
column) corresponds to choosing a second edge connected to the same
variable node that will form the next step in the cycle.  A horizontal
move (along a row) corresponds to choosing two edges connected to the
same check node that form part of the path.

For QC LDPC codes there are efficient ways to describe sets of cycles
in terms of the code's polynomial parity check matrix.  In
Section~\ref{sec.findCyclesIntro} we introduce the basic ideas behind
identifying cycles in weight-I QC LDPC codes. In
Section~\ref{sec.findCyclesNonHier}, we show how to identify cycles in
QC LDPC codes of arbitrary weight. Then, in
Section~\ref{sec.inevitable}, we show that higher-weight QC LDPC codes
with certain characteristics inevitably have short cycles, and point
out that this poses an obstacle to constructing QC LDPC codes with
good girth and good waterfall performance---an obstacle that we will
overcome by introducing hierarchical QC LDPC codes.

%%%%%%%%%%%%%%%%%%%%%%%%%%%%%%%%%%%%%%
\subsection{Finding cycles in weight-I QC LDPC codes}
\label{sec.findCyclesIntro}

To make the logic of the section introduction more concrete, consider
Fig.~\ref{fig.findCycle} which depicts the parity check matrix of a weight-I
QC LDPC code with parameters $J=4$, $L = 9$, and $p = 3$. We focus in
on the four $3 \times 3$ cyclic shift matrices (represented by the
black squares) $\shiftMat{a}$, $\shiftMat{b}$, $\shiftMat{c}$,
and $\shiftMat{d}$. Two choices for the parameters of these four
matrices are shown in the sub-figures: $a = 0$, $b = 1$, $c = 2$,
and $d = 1$ on the left, and $a = 0$, $b=c=d=1$ on the right.

\begin{figure}[htbp]
      \centering
      \includegraphics[scale=0.5,type=eps,ext=.eps,read=.eps]
                      {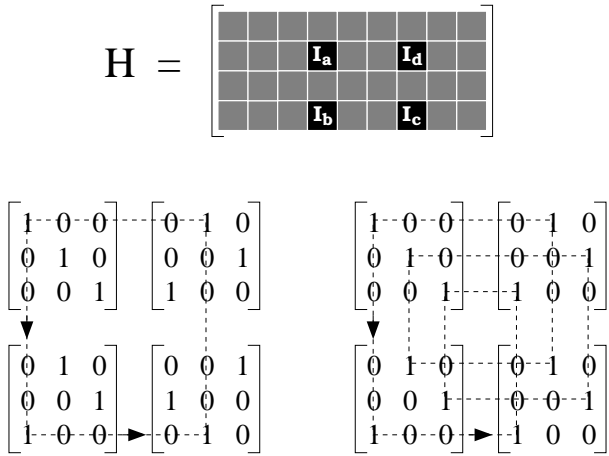}
      \caption{A parity-check matrix and four $3\times 3$ circulant
        permutation matrices ($\shiftMat{a}$, $\shiftMat{b}$,
        $\shiftMat{c}$ and $\shiftMat{d}$) selected from it.  One
        set of parameters (lower left, $a = 0$, $b = 1$, $c =
        2$, $d = 1$ ) results in a cycle of length four.  An
        alternate set (lower right, $a=0$, $b=c=d=1$) results
        in a cycle of length twelve.}
      \label{fig.findCycle}
\end{figure}

Consider any path through the base matrix of the code.  Due to the way
we generate the code's parity check matrix by replacement of each base
matrix entry by a $p \times p$ circulant matrix, a path through the
base matrix corresponds to $p$ paths through the Tanner graph of the
code.  For any of these paths through the Tanner graph to be a cycle,
the path must end at the same variable node from which it started.
For this to happen in a weight-I QC LDPC code, it is necessary for the
path through the base matrix to form a cycle, without passing through
any all-zeros matrices.  But this is not sufficient, since each cyclic
shift matrix corresponds to $p$ parity and $p$ variable nodes.  The
path could end up at a different variable node in the same cyclic
shift matrix and not complete a cycle.

The necessary and sufficient condition for cycles to exist is that
when the path through the base matrix returns to the starting entry,
it returns to the same column of the cyclic shift matrix from which it
started.  In the example of Fig.~\ref{fig.findCycle}, consider the
path through the base matrix starting at the entry labeled $a$, then
progressing through the entries labeled $b$, $c$, and $d$ in turn, and
terminating at the entry labeled $a$.  The corresponding path through
the parity check matrix, with parameter settings $a=0$, $b = 1$, $c =
2$, $d = 1$, is depicted in the left-hand example of
Fig.~\ref{fig.findCycle} and results in a cycle of length four.
However, with the slightly different choice of circulant shifts of the
right-hand example, a return to the same column of the cyclic shift
matrix occurs only after two more passes around the base matrix and an
overall cycle of length $12$.

We now specify the conditions on the $\{a,b,c,d\}$ that result in a
cycle (in fact in a set of $p$ cycles).  Calculate an alternating sum
of the shift indices associated with neighboring permutation matrices
along a given path, where every odd shift index is subtracted rather
than added.  For example, consider the left-hand path of
Fig.~\ref{fig.findCycle}.  The sum is $-a + b - c + d$.  Each
difference between neighboring shift indices in the sum corresponds to
the shift in what column (i.e., what variable node) of the cyclic
permutation matrices the path passes through.  Only if the differences
sum to zero (mod-$p$) at the end of the path will the path return to
the same variable node in the starting permutation matrix, thereby
forming a cycle.  For the example of Fig.~\ref{fig.findCycle}, the
condition for a length-four cycle to exist is:
\begin{equation}\label{p4cycle}
(-a + b - c + d) \; \; \textmd{mod} \; \; p  = 0,
\end{equation}
which is satisfied for $a = 0$, $b = 1$, $c = 2$, $d = 1$, but
is not satisfied by $a = 0$, $b = c = d = 1$.  

\subsection{Finding cycles in higher-weight QC LDPC codes}
\label{sec.findCyclesNonHier}

We now take a step up in complexity from weight-I QC LDPC codes, and consider 
the more involved example of the weight-II code of
Example~\ref{ex.QCLDPC} from Section~\ref{sec.defQC_LDPC}.  
Recall that this code
is defined by the $2 \times 3$ polynomial
parity-check matrix 
\begin{equation}
\parChkMat(x) = \left[ \begin{array}{ccc} x^0 & x^0 & x^0 \\ 0 & x^0 &
    x^1 + x^2
\end{array} \right]. \label{eq.polyChkMatExA2}
\end{equation}

In terms of the coefficients $c_s[j,l]$ defined by 
$h_{j,l}(x) = \sum_{s = 0}^{p-1} c_s[j,l] x^s$ (see (\ref{eq.standardCoefficients})),
we have that all the coefficients $c_s[j,l]$ are equal to
zero except for $c_s[j,l] = 1$ 
when $s = 0$ and $(j,l)$ equals $(1,1)$, $(1,2)$,
$(1,3)$ or $(2,2)$, and for $s=1$ or $s= 2$, when $(j,l) = (2,3)$.

Now, consider the following ordered series:
\begin{equation}
\ordSet = \{(1,2), (2,2), (2,3), (2,3), (2,3), (1,3)\} \label{eq.ordSerEx}
\end{equation}
where each pair $(j,l)$ in $\ordSet$ satisfies $1 \leq j \leq J = 2$
and $1 \leq l \leq L = 3$.  This ordered series specifies a sequence
of rectilinear moves through $\parChkMat(x)$.  These moves are
analogous to those in Fig.~\ref{fig.findCycle} with the important
distinction that if the polynomial in position $(j,l)$ has more than
one term (that is, $c_{s}[j,l]$ is non-zero for more than one value of
$s$), then the next pair in the sequence {\em can} be the same.  For
example, in~(\ref{eq.ordSerEx}) the third, fourth, and fifth pairs are
identical.

To specify a candidate cycle through the Tanner graph, we associate a
coefficient index $s$ with each pair $(j,l)$ in $\ordSet$, such that
$c_s[j,l] \ne 0$.  We denote this series of coefficient indices by
$\pathCoeff$. To ensure that each step in the series corresponds to
traversing a distinct edge in the Tanner graph we require the
following of neighboring pairs $(j^-, l^-)$ and $(j^+, l^+)$ in
$\ordSet$ and the corresponding neighboring coefficient indices $s^-$
and $s^+$ in $\pathCoeff$: if $(j^-, l^-) = (j^+, l^+)$, then the
corresponding indices $s^- \neq s^+$.

The candidate cycle will actually be a cycle if the alternating sum of
coefficient indices in $\pathCoeff$ modulo $p$ equals zero.

In our example, consider the
two following choices for the respective (ordered) sets of
coefficient indices:
\begin{align}
\pathCoeff_a = \{0,0,2,1,2,0\} \label{eq.pathCoeffEx}\\
\pathCoeff_b = \{0,0,1,2,1,0\}. \label{eq.pathCoeffExB}
\end{align}
Each of these choices corresponds to a cycle of length-$6$ through the
Tanner graph of the code, illustrated in Fig.~\ref{fig.weightIIcycles}.
The alternating sums modulo-$3$ can be verified to be equal to zero.
Respectively these sums are:
\begin{eqnarray*}
(- 0 + 0 - 2 + 1 - 2 + 0) \; \mbox{mod} \; 3 & =  (-3) \; \mbox{mod} \; 3 & =  0 \\
(- 0 + 0 - 1 + 2 - 1 + 0) \; \mbox{mod} \; 3 & =   \;\;\;(0) \; \mbox{mod} \; 3 & =  0.
\end{eqnarray*}

\begin{figure}[htbp]
      \centering
      \includegraphics[scale=0.5,type=eps,ext=.eps,read=.eps]
                      {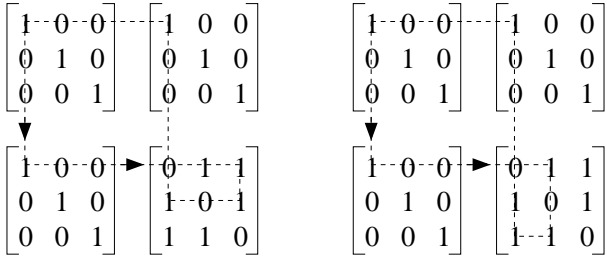}
      \caption{The two length-$6$ cycles through the Tanner graph of
        the weight-II QC LDPC code of Example~\ref{ex.QCLDPC}.}
      \label{fig.weightIIcycles}
\end{figure}

\subsection{Inevitable cycles in higher-weight QC LDPC codes}
\label{sec.inevitable}

Unfortunately, the logic described in the previous section implies that higher-weight
QC LDPC codes will inevitably contain short cycles. Let us begin with a 
straightforward and important theorem, already proven by 
Smarandache and Vontobel \cite{vontobelRefs}, 
that states that any weight-III QC LDPC code will inevitably
contain cycles of length six. To prove this, we note that we can choose
a cycle with an ordered series 
\begin{equation}
\ordSet = \{(j,l),(j,l),(j,l),(j,l),(j,l),(j,l)\} \label{eq.ordSerExB}
\end{equation}
of six identical entries such that each pair $(j,l)$
gives the row $j$  and column $l$ of the
same weight-III polynomial in the parity check matrix $\parChkMat(x)$. 
Suppose, without loss of generality, that the weight-III polynomial has the
form $x^a + x^b + x^c$. Then we can choose for the cycle 
an ordered set of coefficient indices
$\pathCoeff = \{a,b,c,a,b,c\}$ and we will find that 
\begin{equation}
(-a + b - c + a - b +c) \; \; \textmd{mod} \; \; p  = 0, \label{eq.automaticCycleEx}
\end{equation}
automatically for any $p$.

Smarandache and Vontobel also proved (see their Theorem 17) 
that if the parity check matrix
$\parChkMat(x)$ of a  weight-II QC LDPC code contains
two weight-two polynomials in the same row or the same column, that code will
inevitably have eight-cycles. Again, this is easy to verify using our approach.
Suppose for example that the two weight-2 polynomials are in the same row $j$ and two
different columns $l_1$ and $l_2$, and that the polynomial at $(j,l_1)$ is
$x^a + x^b$, while the polynomial at $(j,l_2)$ is $x^c + x^d$. We can find an eight-cycle
that has
the ordered series
\begin{equation}
\ordSet =
\{(j,l_1),(j,l_1),(j,l_2),(j,l_2),(j,l_1),(j,l_1),(j,l_2),(j,l_2)\} \label{eq.ordSerExC}
\end{equation}
and the ordered set of indices
\begin{equation}
\pathCoeff = \{a, b, c, d, b, a, d, c\} \label{eq.pathCoeffExC}
\end{equation}
so that we find
\begin{equation}
(-a + b - c + d - b + a - d + c) \; \; \textmd{mod} \; \; p  = 0,
\end{equation}
regardless of the value of $p$.

These inevitable six-cycles and eight-cycles at first sight appear to
put serious limitations on what protographs can be converted into
quasi-cyclic codes with high girth.  We noted in
Section~\ref{sec.protographs} that a protograph could be equivalently
described using a connectivity matrix, and that a parity check matrix
of a quasi-cyclic code could be derived from the connectivity matrix
by the ``direct transformation'' which replaces the entries of the
connectivity matrix by polynomials with weight equal to the entry. We
now see that if, for example, the protograph has a type of variable
that is connected to a type of check by three edges, a direct
transformation will inevitably lead to six-cycles in the obtained QC
LDPC code.

\begin{figure}[htbp]
      \centering
      \includegraphics[scale=0.3,type=eps,ext=.eps,read=.eps]
                      {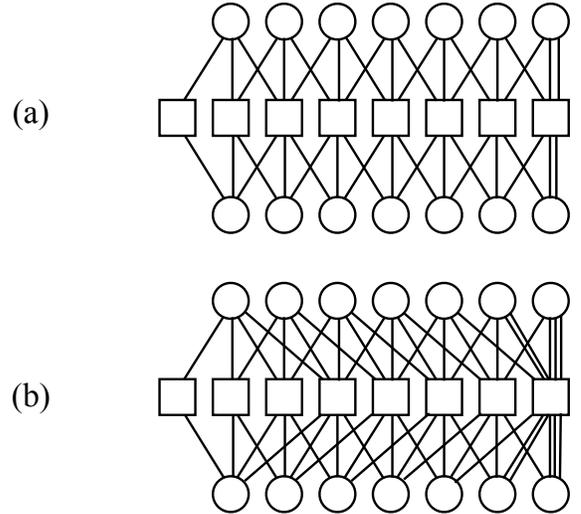}
      \caption{Protographs for ``one-sided'' spatially-coupled codes
        as described in \cite{Kudekar}. The QC LDPC code constructed by
        a direct transformation from the protograph in (a) will inevitably have
        eight-cycles because the check type at the right end is connected
        by two edges to the bit types above and below it. The QC LDPC
        code constructed by a direct transformation from the protograph in (b)
        will inevitably have six-cycles because there exist bits types
        at the right end that are connected by three edges to a check type.}
      \label{fig.onesided}
\end{figure}

Furthermore, protographs with higher edge weights are not particularly
exotic. Consider for example the protographs shown in Fig.~\ref{fig.onesided}, 
which are the protographs for ``one-sided''
spatially coupled codes as described by Kudekar et
al. \cite{Kudekar}. Notice that if we used a direct transformation to
convert these protographs into QC LDPC codes, the QC LDPC codes
corresponding to the protographs in Fig.~\ref{fig.onesided}~(a) would
inevitably have eight-cycles, while those in Fig.~\ref{fig.onesided}~(b) would inevitably have six-cycles.

It turns out that there do exist techniques to construct QC LDPC codes corresponding
to these protographs that have girth of 10 or greater, but to understand these
techniques, we need to make an apparent detour, and introduce {\em hierarchical}
QC LDPC codes.

\begin{figure*}
\begin{align}
\parChkMat(x) & = \left[ \begin{array}{c|c|c|c}
  \circThree{x^2}{0}{x^1+x^7} & \circThree{0}{x^7}{1+x^6} &
  \circThree{0}{0}{0} & \circThree{x^2}{x^5}{1} \\\hline 
  \circThree{0}{x^7}{1+x^6} & \circThree{x^2}{0}{x^1+x^7} &
  \circThree{x^2}{x^5}{1} & \circThree{0}{0}{0}
\end{array} \right] \label{eq.polyParChkB} \\ \nonumber\\
%%% 
\parChkMat(x,y) & = \left[\begin{array}{c|c|c|c} x^2 + (x+x^7)y^2 &
    x^7y+ (1+x^6)y^2 & 0 & x^2 + x^5y + y^2\\ \hline
    x^7y+ (1+x^6)y^2 & x^2 + (x+x^7)y^2 & x^2 + x^5y + y^2 & 0
\end{array} \right] \label{eq.polyParChkTwoVarB} \\ \nonumber \\
%%%
\parChkMat(x,y,z) & = \Big[ x^2 + (x+x^7)y^2 + (x^7y+(1+x^6)y^2)z \; \Big|
  \; (x^2+x^5y+y^2)z\Big] \label{eq.polyParChkThreeVarB}
\end{align}
\end{figure*}

%%%%%%%%%%%%%%%%%%%%%%%%%%%%%%%%%%
\section{Hierarchical QC LDPC codes}
\label{sec.defHierQC_LDPC}

We now introduce {\em hierarchical} QC LDPC codes (HQC LDPC codes),
motivated by the fact that these codes will ultimately enable us to
solve the problem of constructing QC LDPC codes corresponding to
protographs with multiple edges between check and variable types,
without creating inevitable short cycles in the Tanner graph of the
code.  However, because these codes may eventually have other
applications, we present their construction in a form that is actually
more general than we will need for the purpose of eliminating
inevitable short cycles.

A hierarchical QC LDPC code is formed from ``levels'' that each have a
quasi-cyclic structure. The structure can be specified in two
equivalent, complementary forms: one in terms of the polynomial parity
check matrices of these codes, and another in terms of the ``tree
structure'' of these codes.

\subsection{Parity check matrices of hierarchical QC LDPC codes}

Before fully defining HQC LDPC codes formally, it is easier to have a
concrete example in mind.

\vskip .1cm
\begin{example} \label{ex.hierQCLDPC}
Consider the polynomial parity check matrix specified in
equation~(\ref{eq.polyParChkB}) with $p=8$. Because the highest weight
of any of the polynomial entries is $2$, (e.g., $h_{1,3}(x) = x^1 +
x^7$), and because there are $12$ columns in the matrix, this is a
length-96 weight-II QC LDPC code.

But note that this parity check matrix has additional structure which
makes it a {\em hierarchical} QC LDPC code.  In particular, in this
example, each $3 \times 3$ sub-matrix of polynomials
in~(\ref{eq.polyParChkB}) has a circulant structure, as do both the
left-hand and right-hand sets of $2 \times 2$ sub-matrices of $3
\times 3$ sub-matrices.

Just as we use polynomials in the dummy variable $x$ to represent the
underlying circulant sub-matrices in a standard QC LDPC code, we can
use a bi-variate polynomial in the two dummy variables $x$ and $y$ to
represent both the circulant matrices represented by the variable $x$
in~(\ref{eq.polyParChkB}) as well as the circulant arrangements within
each $3 \times 3$ sub-matrix of polynomials in $x$.  The latter
circulant structure we represent using the dummy variable $y$. We can
further represent the $2 \times 2$ circulant structure of $3 \times 3$
circulant sub-matrices using the additional dummy variable $z$.

Thus, in equation~(\ref{eq.polyParChkTwoVarB}) we contract the $6
\times 12$ polynomial parity check matrix $\parChkMat(x)$ of
equation~(\ref{eq.polyParChkB}) into the $2 \times 4$ bi-variate
polynomial parity check matrix $\parChkMat(x,y)$. As we use this
example to illustrate many aspects of the ensuing discussion, please
make sure you think about and understand why, e.g., the upper right $3
\times 3$ sub-matrix in $\parChkMat(x)$ is represented by the
bi-variate polynomial $x^2 + x^5 y + y^2$ in $\parChkMat(x,y)$.

We can repeat the process to contract $\parChkMat(x,y)$ into the $1
\times 2$ tri-variate polynomial parity check matrix
$\parChkMat(x,y,z)$ given in equation~(\ref{eq.polyParChkThreeVarB}).

Each of the three contractions of the parity check matrix of this code
into the polynomial parity check matrices represented
by~(\ref{eq.polyParChkB}), (\ref{eq.polyParChkTwoVarB}),
and~(\ref{eq.polyParChkThreeVarB}), corresponds to a ``level'' in the
hierarchy of this 3-level HQC LDPC code. \hfill \QED
\end{example}
\vskip .1cm

In this example, we started with a polynomial parity check matrix
$\parChkMat(x)$, and contracted it first to $\parChkMat(x,y)$ and then
to $\parChkMat(x,y,z)$. When constructing an HQC LDPC code, it is
often more natural to go in the other direction---expanding a matrix
like $\parChkMat(x,y,z)$ into $\parChkMat(x)$ and then ultimately into
the full parity check matrix whose entries are ones and zeroes.  To
expand a polynomial matrix, we obviously need to know the size of the
circulant matrices at every level.

We now present a formal definition of the family of $K$-level
hierarchical QC LDPC codes which generalizes our example.

\vskip .1cm
\begin{defn} \label{def.hierQC_LDPC}
A hierarchical QC LDPC code with $K$ levels is defined by a 
$\Jrow{K}
\times \Lcol{K}$ multi-variate polynomial parity check matrix
$\parChkMat(\cdot)$ in $K$ variables. 
The entry in the $j$th row and
$l$th column of $\parChkMat(\cdot)$, $1 \leq j \leq \Jrow{K}$, $1 \leq
l \leq \Lcol{K}$ is a $K$-variate polynomial $h_{j,l}(\cdot, \ldots,
\cdot)$ over the $K$ variables, $\x{1}, \ldots, \x{K}$.  The
maximum exponent of any of these polynomials in $\x{k}$, $1 \leq k
\leq K$, is $\p{k}-1$.  The coefficient associated with the term
$\x{1}^{\ind{1}} \cdot \x{2}^{\ind{2}} \cdots \x{K}^{\ind{K}}$ where
$0 \leq \ind{k} \leq \p{k}-1$ for all $k$ is $\sym{\ind{1}, \ldots,
  \ind{K}}{j}{l}$.  With these definitions we defined the code by the
$\Jrow{K} \cdot \Lcol{K}$ polynomials
\begin{align}
h_{j,l}(\x{1}, \ldots, & \x{K}) = \nonumber \\ & \sum_{\ind{K} =
  0}^{\p{K}-1} \ldots \sum_{\ind{1} = 0}^{\p{1}-1} \sym{\ind{1},
  \ldots, \ind{K}}{j}{l} \left( \prod_{k=1}^{K}
\x{k}^{\ind{k}}\right).
\end{align}
The parity check matrix of such a code is obtained by replacing each
of the $\Jrow{K} \cdot \Lcol{K}$ entries of $\parChkMat(\x{1}, \ldots,
\x{K})$ with the sub-matrix
\begin{equation}
\sum_{\ind{K} = 0}^{\p{K}-1} \!\!\! \ldots \sum_{\ind{1} = 0}^{\p{1}-1}
\sym{\ind{1}, \ldots \ind{K}}{j}{l} \left(
\shiftMat{1,\p{K}}^{\ind{K}} \kronecker \ldots \kronecker
\shiftMat{1,\p{1}}^{\ind{1}}\right), \label{eq.constPolyParCheck}
\end{equation}
where $\kronecker$ denotes a Kronecker product.  Defining the
recursive relations $\Jrow{k-1} = \Jrow{k} \cdot \p{k}$ and
$\Lcol{k-1} = \Lcol{k} \cdot \p{k}$, where $0 \leq k \leq K$, the
parity check matrix thus constructed has $\Jrow{0} = \Jrow{K} \cdot
\prod_{k=1}^{K} \p{k}$ rows and $\Lcol{0} = \Lcol{K} \cdot
\prod_{k=1}^{K} \p{k}$ columns. \hfill \QED
\end{defn}
\vskip .1cm

While the definition of HQC LDPC codes holds more generally for codes
defined in fields other than GF(2), in this paper we exclusively
consider {\em binary} QC LDPC codes wherein all $\sym{\ind{1}, \ldots,
  \ind{K}}{j}{l}$ are binary.  We return to our previous example to
illustrate our definitions.

{\em Example~\ref{ex.hierQCLDPC} (continued):} The code of
this example is a three-level HQC LDPC code.  To cast this
example into the language of Definition~\ref{def.hierQC_LDPC} we first
identify $x$ with $\x{1}$, $y$ with $\x{2}$, and $z$ with $\x{3}$.

In this example $\p{1} = 8$, $\p{2} = 3$, $\p{3} = 2$. Therefore,
$\Jrow{3} = 1$, $\Lcol{3} = 2$; $\Jrow{2} = 2$, $\Lcol{2} = 4$;
$\Jrow{1} = 6$, $\Lcol{1} = 12$; and $\Jrow{0} = 48$, $\Lcol{0} = 96$.

We can rewrite, e.g., the term $h_{1,1}(x,y,z)$
of~(\ref{eq.polyParChkThreeVarB}) as
\begin{align*}
& h_{1,1}(\x{1}, \x{2}, \x{3}) \\ & = \x{1}^2 +
  \left(\x{1}+\x{1}^7\right)\x{2}^2 +
  \left(\x{1}^7\x{2}+\left(1+\x{1}^6\right)\x{2}^2\right)\x{3}\\ 
& = \sum_{\ind{3} = 0}^{1} \sum_{\ind{2} = 0}^{2} \sum_{\ind{1} =
    0}^{7} \sym{\ind{1}, \ind{2}, \ind{3}}{1}{1} \x{1}^{\ind{1}}
  \x{2}^{\ind{2}} \x{3}^{\ind{3}},
\end{align*}
where all coefficients $\sym{\ind{1},\ind{2},\ind{3}}{1}{1}$ are zero except
for $\sym{2,0,0}{1}{1}=\sym{1,2,0}{1}{1}=\sym{7,2,0}{1}{1} =
\sym{7,1,1}{1}{1} = \sym{0,2,1}{1}{1} = \sym{6,2,1}{1}{1} = 1$.  \hfill \QED
\vskip .1cm

Rather than expanding $\parChkMat(\x{1}, \ldots, \x{K})$ into a full
parity check matrix as in~(\ref{eq.constPolyParCheck}), one often
wants to generate the form given in
equation~(\ref{eq.polyParChkGeneral}) of the polynomial parity check
matrix $\parChkMat(\x{1})$ of a QC LDPC code {\em in one variable}.
To do this we use the construction of~(\ref{eq.constPolyParCheck}) for
all but the first level.  We replace each $h_{j,l}(\x{1},\ldots,
\x{K})$ with the polynomial matrix in $\x{1}$
\begin{align}
\sum_{\ind{K} = 0}^{\p{K}-1} \!\!\!\! \ldots \! \sum_{\ind{1} =
  0}^{\p{1}-1} \!  \! \sym{\ind{1}, \ldots, \ind{K}}{j}{l} \!\!
\left( \shiftMat{1,\p{K}}^{\ind{K}} \!\! \kronecker \! \cdots \!
\kronecker \shiftMat{1,\p{2}}^{\ind{2}} \right)
\x{1}^{\ind{1}}. \label{eq.constPolyParCheckB}
\end{align}
The matrix $\parChkMat(\x{1})$ is of size $\Jrow{1} \times \Lcol{1}$.
We return once more to our example to illustrate this idea.

{\em Example~\ref{ex.hierQCLDPC} (continued):} Consider the
final term of $h_{1,1}(\x{1}, \x{2}, \x{3})$, namely
$(1+\x{1}^6)\x{2}^2\x{3}$, corresponding to the non-zero coefficients
$\sym{0,2,1}{1}{1}$ and $\sym{6,2,1}{1}{1}$.  According to
equation~(\ref{eq.constPolyParCheckB}), The contribution of this term
to $\parChkMat(\x{1})$ is
\begin{equation*}
\sym{0,2,1}{1}{1} \left( \shiftMat{1,2} \kronecker
\shiftMat{1,3}^2\right)\x{1}^0 + \sym{6,2,1}{1}{1} \left( \shiftMat{1,2}
\kronecker \shiftMat{1,3}^2 \right) \x{1}^6,
\end{equation*}
where $\x{1}^0 = 1$, $\sym{1,2,1}{1}{1} = \sym{6,2,1}{1}{1} = 1$ and
\begin{equation}
\shiftMat{1,2} \kronecker \shiftMat{1,3}^2 = 
%
%\left[ \begin{array}{cc} \circThree{0}{0}{0} &
%    \circThree{0}{0}{1}\\ \circThree{0}{0}{1} &
%    \circThree{0}{0}{0} \end{array} \right].
% 
\left[ \begin{array}{cccccc}
    0&0&0&0&0&1\\ 0&0&0&1&0&0\\ 0&0&0&0&1&0\\ 0&0&1&0&0&0\\ 
1&0&0&0&0&0\\ 0&1&0&0&0&0 \end{array}\right]. \label{eq.6x6matrix}
\end{equation}
Referring back to the left-hand six-by-six sub-matrix of
$\parChkMat(x,y,z)$ in~(\ref{eq.polyParChkB}) we can confirm the
correctness of this pattern, as a $1+x^6$ term appears in each
of the non-zero entries in the matrix of equation (\ref{eq.6x6matrix}). 

Having worked this example, we can now see how the form of
equation~(\ref{eq.constPolyParCheckB}) nicely reveals the structure of
HQC LDPC codes.  Each row and each column of the matrix
$\shiftMat{1,\p{K}}^{\ind{K}} \kronecker \cdots \kronecker
\shiftMat{1,\p{2}}^{\ind{2}}$ has exactly one non-zero element.  If
the coefficient $\sym{\ind{1}, \ind{2}, \ldots, \ind{K}}{j}{l}$ is
non-zero, the permutation matrix $\shiftMat{1,\p{1}}^{\ind{1}}$
(equivalent to the term $x_{[1]}^{i_1}$) is added at the location of
each of these non-zero elements.  \hfill \QED

\vskip .1cm

Finally, we note that the polynomial parity check matrix of a
$K$-level HQC LDPC code can more generally be expanded
into a parity check polynomial $\parChkMat(\x{1}, \ldots,
\x{\tilde{K}})$ in $\tilde{K}$ variables where $\tilde{K} < K$.  We
call this the ``level-$\tilde{K}$'' polynomial parity check matrix of
the code.  We derive this matrix by expanding out all but the last
$\tilde{K}$ levels.  Replace each $h_{j,l}(\x{1},\ldots, \x{K})$ with
the polynomial matrix in $\x{1}, \ldots, \x{\tilde{K}}$
\begin{align*}
\sum_{\ind{K} = 0}^{\p{K}-1} \!\!\!\! \ldots \! \sum_{\ind{1} =
  0}^{\p{1}-1} \!  \! \sym{\ind{1}, \ldots, \ind{K}}{j}{l} \!\!
\left( \shiftMat{1,\p{K}}^{\ind{K}} \!\! \kronecker \! \cdots \!
\kronecker \shiftMat{1,\p{\tilde{K}+1}}^{\ind{\tilde{K}+1}} \right)
\prod_{k=1}^{\tilde{K}} \x{k}^{\ind{k}}.
\end{align*}
The matrix $\parChkMat(\x{1}, \ldots, \x{\tilde{K}})$ has dimension
$\Jrow{\tilde{K}} \times \Lcol{\tilde{K}}$.

%%%%%%%%%%%%%%%%%%%%%%%%%%%%%%%%%%%%%%%
\subsection{Tree structure of HQC LDPC codes}
\label{sec.connectMat}

We now show that we can alternately describe an HQC LDPC code by
specifying the code's {\em tree structure}.  The tree structure of any
HQC LDPC code is defined by a matrix of {\em labeled trees}, defined
in Definition~\ref{def.connTree}.  These labeled trees quite naturally
reveal the hierarchical structure of the code.  We will show that
there is a complete equivalence between
Definition~\ref{def.hierQC_LDPC} of the last section and the
definitions of this section.  We can start with
Definition~\ref{def.hierQC_LDPC} and easily find the unique set of
labeled trees that specify the code or, starting from a tree
structure, find the unique HQC LDPC code that has that structure.

The reasons to consider this alternate description are two-fold.  First, the
representations of this section help reveal the
hierarchical structure within the algebraic description of
Definition~\ref{def.hierQC_LDPC}.  Second, we will 
use {\em unlabeled} trees to define a family of HQC LDPC codes, and then
will want to search for a labeling within that family
to optimize girth.

The basic observation that motivates the following definitions is that
the non-zero terms of the polynomials that define any HQC
LDPC code have a {\em hierarchical clustering} that can be represented by a
labeled tree.  We
formally define such a {\em labeled tree} as follows.

\begin{defn} \label{def.connTree}
A {\em labeled tree} $\conMat$, corresponding to an entry in the
$\Jrow{K} \times \Lcol{K}$ multi-variate polynomial parity check
matrix $\parChkMat(\cdot)$ in $K$ variables defining a $K$-level HQC
LDPC code, is a depth-$K$ tree.  The root node of the tree is the
single node at the top ($K$th) level.  Each node at level $k$, $1 \leq
k \leq K$, has a number of edges connecting it to nodes in the next
level down.  The number of edges must be an integer in the set $\{1,
\ldots, \p{k}-1\}$.

Each edge below a node at level $k$ is labeled by an integer in the
set $\{0, 1, \ldots, \p{k}-1\}$.  Edges are termed ``siblings'' if
they share the same parent (i.e., are connected to the same node at
the higher level).  The edge labels of sibling nodes are constrained
to be distinct. We refer to the edges below the lowest nodes as
``leaves.'' We will have need to index the edges at each level of the
tree, so use $|\conMat[k]|$ to denote the number of edges in $\conMat$
at level $k$, i.e., the set of edges that have a parent node at level
$k$.
\QED
\end{defn}

The code discussed in Example~\ref{ex.hierQCLDPC} is characterized by
the matrix of two labeled trees shown in Figure~\ref{fig.trees}. The
left-hand tree characterizes the polynomial $h_{1,1}(x,y,z)$ and the
right-hand tree characterizes $h_{1,2}(x,y,z)$, both specified
in~(\ref{eq.polyParChkThreeVarB}).  Before understanding how these
labeled trees relate to the structure of the code we note that for
this code $\p{1} = 8$, $\p{2} = 3$ and $\p{3} = 2$, and node and edge
labels are within the ranges specified by
Definition~\ref{def.connTree}.

The next definition relates these trees to
the structure of the code.

\begin{figure}%[htbp]
      \centering
      \includegraphics[scale=0.4,type=eps,ext=.eps,read=.eps]{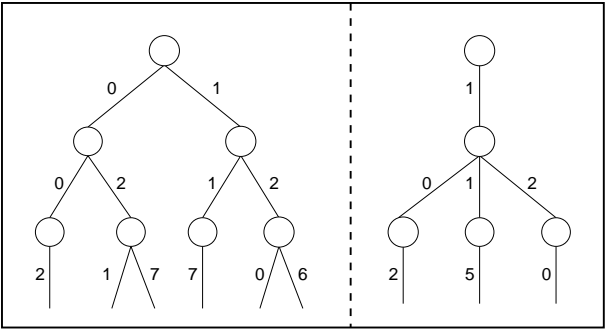}
      \caption{Example of the tree structure of a family of
        three-level hierarchical QC LDPC codes.  The left-hand tree is
        $\conMat_{1,1}$, the right-hand tree is $\conMat_{1,2}$.}
      \label{fig.trees}
\end{figure}

\begin{defn} \label{def.connStruct}
The {\em tree structure} of a $K$-level HQC LDPC code is specified by
a matrix of labeled trees $\setTree = \{\conMat_{j,l}\}$, $1 \leq j
\leq \Jrow{K}$, $1 \leq l \leq \Lcol{K}$.  To each leaf of
$\conMat_{j,l}$
%
%indexed by $i$, $1 \leq i \leq |\conMat_{j,l}|$, 
% 
we associate a single non-zero coefficient $\sym{\ind{1}, \ldots,
  \ind{K}}{j}{l}$ in a one-to-one manner.  If the edge labels on the
unique path from the leaf to the root node are $e_1, \ldots, e_K$ then
the non-zero coefficient associated with the leaf is $\sym{e_1,
  \ldots, e_K}{j}{l} = 1$.

In certain cases (corresponding to all-zero polynomials) we want to
define a ``null'' tree.  This is a non-existent tree 
(and therefore no edges exist so all coefficients are zero).  We use
the special symbol $\nullTree$ to denote the null tree.  E.g.,
$\conMat_{2,1} = \nullTree$ for the code specified
in~(\ref{eq.polyChkMatExA}).
\hfill \QED
\end{defn}

The number of edges below level $K$ of tree $\conMat_{j,l}$ indicates
the number of {\em distinct} powers of $\x{K}$ that appear in
$h_{j,l}(\x{1}, \ldots, \x{K})$.  Each node at level $K-1$ corresponds
to one of these terms.  The number of edges below each of the nodes at
level $K-1$ indicates the number of distinct powers of $\x{K-1}$
associated with that term, and so on down the tree.  The number of
leaves in the tree equals the number of terms in the polynomial
$h_{j,l}(\x{1}, \ldots, \x{K})$.  The maximum number of leaves below
any of the lowest level nodes (across all $(j,l)$ pairs) tells us the
weight of the code (weight-I, weight-II, etc.).  The edge labels
indicate the exponents that define the non-zero polynomials.

We can also define a more fine-grained ``weight at level $k$'' of a
hierarchical code by the maximum number of edges below any of the
nodes at level $k$. A hierarchical code can have different weights at
different levels; for example, the code from
Example~\ref{ex.hierQCLDPC} with tree structure shown in
Figure~\ref{fig.trees} is weight-II at level 1 (the lowest level),
weight-III at level 2, and weight-II at level 3.

The following lemma shows that the two ways of conceptualizing
HQC LDPC codes (Definition~\ref{def.hierQC_LDPC} or
Definition~\ref{def.connStruct}) are equivalent.

\begin{lemma}
There is a one-to-one mapping between any HQC LDPC codes
as defined in Definition~\ref{def.hierQC_LDPC} and a tree structure,
as defined in Definition~\ref{def.connStruct}. \hfill \QED
\end{lemma}

\begin{proof} 
We first show that any HQC LDPC code has a tree structure that can be
read off from the form of the polynomials that make up its polynomial
parity-check matrix. To see this, start with
Definition~\ref{def.hierQC_LDPC}.  The $\Jrow{K} \Lcol{K}$ polynomials
each define one labeled tree.  Using the distributive law, we cluster
the terms of each polynomial as much as possible (i.e., into the
least-factored form of the polynomial).  The resulting (hierarchical)
clustering of terms specifies a labeled tree.

Conversely, we now show that any set of labeled trees can be uniquely
mapped to an HQC LDPC code.  Starting with the set of labeled trees, we
first solve for the non-zero coefficients by concatenating edge labels
on all paths from distinct leaves to the root.  Using the resulting
set of non-zero coefficients in Definition~\ref{def.hierQC_LDPC}
specifies the code.  
\end{proof}

\begin{example}
To understand the structure on the code imposed by the tree topology,
consider again the two trees shown in Fig.~\ref{fig.trees}. 
By ``tree topology,'' 
we simply mean the unlabeled versions of these trees.  Each 
unlabeled tree has three levels and
there are two of them.  From this we infer that these unlabeled trees specify a
family of three-level HQC LDPC codes where $\Jrow{3} = 1$
and $\Lcol{3} = 2$.  Since the maximum number of leaves below a node at the first level
is two, these trees specifies a family of weight-II QC LDPC codes.

Now focus on the left-hand tree.  To simplify notation, let us again
use $x$ for $\x{1}$, $y$ for $\x{2}$, and $z$ for $\x{3}$.  Since the
number of leaves is six, we deduce that $h_{1,1}(x,y,z)$ has six
terms, i.e.,
\begin{equation*}
h_{1,1}(x, y, z) = \sum_{i=1}^6 x^{a_i} y^{b_i} z^{c_i},
\end{equation*}
where, using $\p{1} = 8$, $\p{2} = 3$ and $\p{3} = 2$, $0 \leq a_i
\leq 7$, $0 \leq b_i \leq 2$, and $0 \leq c_i \leq 1$.  Since the root
node has two edges, we deduce that these six terms are clustered into
two sets of polynomials defined by $c_1 = c_2 = c_3$ and $c_4 = c_5 =
c_6$, thus
\begin{equation*}
(x^{a_1} y^{b_1} + x^{a_2} y^{b_2} + x^{a_3}
  y^{b_3})z^{c_1} + (x^{a_4} y^{b_4} + x^{a_5}
  y^{b_5} + x^{a_6} y^{b_6}) z^{c_4}.
\end{equation*}
where $c_1 \neq c_4$.  (Since $c_1$ and $c_4$ are both binary, without
loss of generality we could set $c_1 = 0$ and $c_4 = 1$ at this
point.)  Now from the second level in the tree we deduce that the
terms in $z^{c_1}$ group into two sets, one with two terms so $b_2 =
b_3$.  The same happens with the terms in $z^{c_4}$ where $b_5 = b_6$.
This tells us that the polynomials compatible with this tree have the
form
\begin{equation}
(x^{a_1} y^{b_1} + (x^{a_2} + x^{a_3})
  y^{b_2})z^{c_1} + (x^{a_4} y^{b_4} + (x^{a_5} +
  x^{a_6}) y^{b_5}) z^{c_4},\label{eq.choosePowers}
\end{equation}
where $c_1 \neq c_4$, $b_1 \neq b_2$, $b_4 \neq b_5$,
$a_2 \neq a_3$ and $a_5 \neq a_6$ (but, e.g., $b_1 = b_4$ is allowed).
\hfill \QED
\end{example}

One can now see that the topology of the unlabeled version of
the trees of Fig.~\ref{fig.trees}
specifies a family of HQC LDPC codes, of which the code
considered in Example~\ref{ex.hierQCLDPC}, and specified
in~(\ref{eq.polyParChkThreeVarB}), is one member.  As the last
example illustrates, many degrees of freedom remain within the
specified family.  In particular these are the choice of the $a_i$,
$b_i$ and $c_i$ in~(\ref{eq.choosePowers}), subject to the constraints
$c_1 \neq c_4, b_1 \neq b_2, \ldots, a_5 \neq a_6$.  In the algorithms
of Section~\ref{sec.girthMax}, were we maximize the girth of our codes,
we search among these degrees of freedom, keeping the code's
unlabeled tree structure fixed.

Finally, we note that in a non-hierarchical weight-I QC LDPC code, the
trees in $\setTree$ are quite simple.  Each is either the null tree or
a tree that consists of a single root node with
a single leaf below it.
No leaf has a sibling so no constraints are placed on the choice of
edge labels.

%%%%%%%%%%%%%%%%%%%%%%%%%%%%%%%%%%%%%%%%%%%%%
\section{Cycles in Hierarchical QC LDPC codes}
\label{sec.condForCycles}

We now state the necessary and sufficient conditions on the polynomial
parity check matrix of an HQC LDPC code for that code to have a cycle
of a particular length.  These conditions generalize those specified
by Fossorier in~\cite{MarcQC} for weight-I QC LDPC codes. They are
also formalizations and generalizations of the examples we gave for
higher-weight QC LDPC codes in Section \ref{sec.findCyclesNonHier};
the main important new requirement compared to those examples is that
our cycles now need to be cycles at all levels of the hierarchy
simultaneously.

\subsection{Finding cycles in HQC LDPC codes}

We start by defining a path (or ``candidate cycle'') through a
$K$-variate polynomial parity check matrix.

\begin{defn} \label{def.path}
A length-$2\len$ path $\path$ through a $K$-variate $\Jrow{K} \times
\Lcol{K}$ polynomial parity check matrix matrix $\parChkMat(\cdot)$ of
an HQC LDPC code is specified by two sets, $\ordSet$ and $\pathCoeff$,
i.e., $\path = \{\ordSet, \pathCoeff\}$.

The first set $\ordSet$ is an ordered series
\begin{equation} \label{cycleseries}
\ordSet = \{ (j_1, l_1), \! (j_2, l_2), \! (j_3,l_3), \cdots, \!
(j_{2 \len}, l_{2 \len})\}
\end{equation}
such that 
\begin{enumerate}
\item[(i)] $1 \leq j_t \leq \Jrow{K}$ and $1 \leq l_t \leq \Lcol{K}$ for all $t$, $1 \leq t \leq 2 \Lambda$,
\item[(ii)] $j_{2\len} = j_{1}$ ,
\item[(iii)] $j_t = j_{t+1}$ for $t \in \mathbb{Z}_{\rm even}$ (even integers),
\item[(iv)] $l_t = l_{t+1}$ for $t \in \mathbb{Z}_{\rm odd}$ (odd integers),
\item[(v)] $|\coeffSet[j, l]| > 0$ for all $(j,l) \in \ordSet$, 
where the set $\coeffSet[j,l]$ is defined to be the set
of coefficients in the polynomial in the $j$th row and $l$th column of
$\parChkMat(\cdot)$ that are non-zero:
\begin{equation}
\coeffSet[j, l] = \{\sym{\ind{1}, \ldots, \ind{K}}{j}{l}:
\sym{\ind{1}, \ldots, \ind{K}}{j}{l} \neq 0\}. 
\end{equation}
\end{enumerate}

The second set $\pathCoeff$ is a set of length-$K$ vectors of coefficient indices
\begin{equation}
\pathCoeff = \{{\bf s}[j_1,l_1], {\bf s}[j_2,l_2], \ldots,
{\bf s}[j_{2 \Lambda},l_{2 \Lambda}]\} \label{eq.coeffSet}
\end{equation}
where, as implied by the notation, $(j_t,l_t) \in \ordSet$ for all $t,
1 \leq t \leq 2 \Lambda$, and $|\pathCoeff| = |\ordSet|$.
Furthermore,
\begin{enumerate}
\item[(vi)] the $k$th coordinate $s_k[j,l]$ of ${\bf s}[j,l]$ satisfies 
$0 \leq s_k[j,l] \leq \p{k}-1$ for all $(j,l) \in \ordSet$,
\item[(vii)] $\sym{{\bf s}[j,l]}{j}{l} \in \coeffSet[j, l]$ for all $(j,l)
  \in \ordSet$, where $\sym{{\bf s}[j,l]}{j}{l}$ is a compact notation for
$\sym{\ind{1}, \ldots, \ind{K}}{j}{l}$.
\item[(viii)] if consecutive elements of $\ordSet$ are identical, i.e., $(j_t,
  l_t) = (j_{t+1}, l_{t+1})$ for some $t$, $1 \leq t \leq 2 \Lambda$,
  then ${\bf s}[j_t,l_t] \neq {\bf s}[j_{t+1}, l_{t+1}]$.
\end{enumerate}
\end{defn}

The above definition generalizes those definitions made and used in
Sections~\ref{sec.findCyclesIntro} and~\ref{sec.inevitable} for
finding cycles in higher-weight QC LDPC codes.  In those sections the
ordered set $\ordSet$ and coefficient indices $\pathCoeff$ were first
introduced and their characteristics were described.  For examples of
$\ordSet$ see~(\ref{eq.ordSerEx}), (\ref{eq.ordSerExB}),
and~(\ref{eq.ordSerExC}), and for those of $\pathCoeff$
see~(\ref{eq.pathCoeffEx}), (\ref{eq.pathCoeffExB}),
and~(\ref{eq.pathCoeffExC}).  These examples illustrate the reasoning
behind criteria (1)--(8) in the definition above.

We now state the conditions for a length-$2 \len$ path $\path =
\{\ordSet, \pathCoeff\}$ actually to correspond to length-$2 \len$
cycles in the Tanner graph.  Consider the following alternating sums,
one for each $k$, $1 \le k \le K$:
\begin{equation}
\Sigma[k] = \sum_{t=1}^{2 \Lambda} (-1)^t s_k[j_t,l_t]. \label{def.pathSum}
\end{equation}
As reflected in the following theorem, these sums are the generalization
of the sum in~(\ref{p4cycle}) to HQC LDPC codes.

\begin{thm} \label{thm.condPathIsCycle}
A path length-$2\len$ path $\path = \{\ordSet, \pathCoeff\}$ through
the $K$-variate $\Jrow{K} \times \Lcol{K}$ polynomial parity check
matrix matrix $\parChkMat(\cdot)$ correspond to length-$2 \len$ cycles
in the Tanner graph if and only if {\em for every} $k$, $1 \leq k \leq
K$,
\begin{eqnarray}\label{GirthCon}
 \Sigma[k] \; \textmd{mod} \; \; \p{k}  = 0.
\end{eqnarray}
\hfill \QED
\end{thm}

\begin{proof}
First consider the case where $K=1$, i.e., a {\em non-hierarchical} QC
LDPC code for which~(\ref{GirthCon}) corresponds to Fossorier's
condition.  Recall the logic of Section~\ref{sec.findCycles}.  In this
setting if condition~(\ref{GirthCon}) is {\em not} satisfied, then the
column of the polynomial parity check matrix from which the path
originates is distinct from the one on which the path terminates.
Since distinct columns of the polynomial parity check matrix
correspond to distinct sets of variable nodes in the Tanner graph,
this means that if~(\ref{GirthCon}) is not satisfied the path does not
correspond to a set of cycles.

In general, what condition~(\ref{GirthCon}) is helping us to
understand is whether, in the expanded parity check matrix at the {\em
  next lower level}, the path through the polynomial parity check
matrix corresponds to a set of path through the parity check matrix
that all correspond to cycles in the Tanner graph.  In the case of a
non-hierarchical QC LDPC code there is only one level of expansion,
from the polynomial parity check matrix to the parity check matrix.
However, in an HQC LDPC code there are multiple levels of expansion.

Now consider HQC LDPC codes where $K > 1$.  Given any path consider
whether condition~(\ref{GirthCon}) holds for $k = K$.  If the
condition does not hold then, similar to Fossorier's logic, the path
through the parity-check matrix at the next lower level, i.e., through
the level-$(K-1)$ polynomial parity check matrix, will not start and
end in the same column.  In the hierarchical setting each column at
level $K-1$ corresponds to a set of variable nodes.  However, due to
the way we expand out the parity-check matrix using Kronecker products
in Definition~\ref{def.hierQC_LDPC}, the sets of variable nodes
corresponding to distinct columns of the level-$k$ polynomial parity
check matrix for any given $k$ are non-intersecting.  A
path that originates and terminates in distinct subsets of the
variable nodes cannot correspond to a set of cycles.  Thus,
if~(\ref{GirthCon}) does not hold for $k=K$, the path cannot
correspond to a set of cycles.

On the other hand, if~(\ref{GirthCon}) is satisfied for $k = K$ then
cycles {\em may} exist, depending on what happens at the lower levels.
Using the same argument we recurse down the levels from $k = K$ to $k
= 1$.  If there is any $k$ for which~(\ref{GirthCon}) is not satisfied
then the path originates from and terminate at distinct variable nodes
and therefore does not correspond to a set of cycles.  However,
if~(\ref{GirthCon}) is satisfied for all $k$, $1 \leq k \leq K$, then
the path originates and terminates on the same variable node and
cycles exists. 
\end{proof}

We immediately have the following theorem.
\begin{thm} \label{thm.girthCond}
A necessary and sufficient condition for a $K$-level hierarchical QC
LDPC code to have girth at least $2(\len + 1)$ is the following.  For
all paths through the polynomial parity check matrix of length at most
$2 \len$ (path length at least four and at most $2 \len$),
condition~(\ref{GirthCon}) does not hold for at least one $k$, $1 \leq k \leq
K$. \hfill \QED
\end{thm}

\subsection{Examples}

We now give examples of two paths through the polynomial parity check
matrix of the code of Example~\ref{ex.hierQCLDPC}.  In the first we
describe a path that corresponds to cycles through the Tanner graph.
We first consider the code as a QC LDPC
code (ignoring its hierarchical structure) and use Fossorier's
condition to verify the existence of cycles.  We then look at the same
code from a hierarchical perspective to illustrate
Theorem~\ref{thm.girthCond}.  In the second example we consider a path through
the same code that does not correspond to a cycle through the Tanner
graph.

%%%%%%%%%%%%%%%%%%%%%%%%%%
\example {\em (Cycle in an HQC LDPC code)}
\label{sec.exampleOfCycles}
Consider again the polynomial parity check matrices $\parChkMat(x)$
and $\parChkMat(x,y)$, respectively specified
in~(\ref{eq.polyParChkB}) and~(\ref{eq.polyParChkTwoVarB}).  First
consider the non-hierarchical description of the code specified by
$\parChkMat(x)$.  A cycle of length-four exists traversing the path
$\path = \{\ordSet, \pathCoeff\}$ where
\begin{equation*}
\ordSet = \{(0,0), (1,0), (1,5), (0,5)\}.
\end{equation*}
This corresponds to, in order, the four polynomials
\begin{align*}
\begin{array}{lll}
x^2 & = & \sym{2}{0}{0} \, x^2,\\
x + x^7 & = & \sym{1}{1}{0} \, x + \sym{7}{1}{0} \, x^7,\\
x^7 & = & \sym{7}{1}{5} \, x^7, \\
1 + x^6 & = & \sym{0}{0}{5} \, x^0 + \sym{6}{0}{5} \, x^6.
\end{array}
\end{align*}
Selecting out $\sym{2}{0}{0}$, $\sym{1}{1}{0}$, $\sym{7}{1}{5}$ and
$\sym{0}{0}{5}$ means we choose
\begin{equation*}
\pathCoeff = \{2, 1, 7, 0\}.
\end{equation*}
We calculate the sum in~(\ref{def.pathSum}) to be
\begin{equation}
\Sigma[1] \; {\rm mod} \; 8 = (-2 + 1 -7 + 0) \; {\rm mod} \; 8 = 0,
\end{equation}
where $\p{1} = 8$ for this code.  This example confirms, in the
general notation, the test for cycles in non-hierarchical QC LDPC
codes already discussed in Sec.~\ref{sec.findCyclesIntro}.

Now, consider the same cycle from the hierarchical perspective.  With
respect to the two-level representation $\parChkMat(x,y)$
of~(\ref{eq.polyParChkTwoVarB}) the same cycle through the Tanner
graph corresponds to the ordered series
\begin{equation*}
\ordSet = \{(0,0), (0,0), (0,1), (0,1)\}.
\end{equation*}
Now we have polynomials $x^2 + (x + x^7) y^2$ and $x^7 y + (1 + x^6)
y^2$ which, respectively, are
\begin{equation*}
\begin{array}{lll}
\sym{2,0}{0}{0} \, x^2 + \sym{1,2}{0}{0} \,
x y^2 + \sym{7,2}{0}{0} \, x^7 y^2,
\end{array}
\end{equation*}
and
\begin{equation*}
\begin{array}{lll}
\sym{7,1}{0}{1} \, x^2 + \sym{0,2}{0}{1} \,
x y^2 + \sym{6,2}{0}{1} \, x^7 y^2.\\
\end{array}
\end{equation*}
The same cycles correspond to the coefficient indices
\begin{equation*}
\pathCoeff = \left\{ \left[\begin{array}{c}2\\0\end{array} \right],
\left[\begin{array}{c}1\\2\end{array} \right],
\left[\begin{array}{c}7\\1\end{array} \right],
\left[\begin{array}{c}0\\2\end{array} \right]
\right\}.
\end{equation*}
Note that the first sub-index of each coefficient corresponds to the
sub-index of the coefficients selected at the one-level view.  The
alternating sums along the path are
\begin{align*}
& \Sigma[1] \; {\rm mod} \; 8 = (-2 + 1 - 7 + 0) \; \; {\rm mod} \; 8 = 0 \\
& \Sigma[2] \; {\rm mod} \; 3 = (-0 + 2 - 1 + 2) \; \; {\rm mod} \; 3 = 0
\end{align*}
where $\p{1} = 8$ and $\p{2} = 3$ for this code.  While we do not work
out the example for the three-level representation $\parChkMat(x,y,z)$
of~(\ref{eq.polyParChkThreeVarB}), we note that the ordered traversed
by this cycle would be $\ordSet = \{(0,0), (0,0), (0,0),
(0,0)\}$. \hfill \QED

\example {\em (Non-cycle in an HQC LDPC code)} We now provide an
example of a path through $\parChkMat(x,y)$ for which $\Sigma[1] = 0
\; \textmd{mod} \; \p{1}$ but $\Sigma[2] \neq 0 \; \textmd{mod} \;
\p{2}$.  Let the ordered set be $\ordSet = \{(0,0), (1,0), (1,1),
(0,1)\}$.  of~(\ref{eq.polyParChkTwoVarB}) be $(0,0)$, $(1,0)$,
$(1,1)$, $(0,1)$.  This corresponds to polynomials $x^2 + (x+x^7)
y^2$, $x^7y + (1+x^6)y^2$, $x^2 + (x+x^7)y^2$, and $x^7y +
(1+x^6)y^2$.  We select the set of set coefficient indices to be
\begin{equation*}
\pathCoeff = \left\{ \left[\begin{array}{c}2\\0\end{array} \right],
\left[\begin{array}{c}6\\2\end{array} \right],
\left[\begin{array}{c}2\\0\end{array} \right],
\left[\begin{array}{c}6\\2\end{array} \right]
\right\},
\end{equation*}
from which we can verify that $\Sigma[1] = 0 \; \textmd{mod} \; 8$ but
$\Sigma[2] \neq 0 \; \textmd{mod} \; 3$.  Hence while
condition~(\ref{GirthCon}) holds at level one, it does not hold at
level two.  Referring to the expanded $\parChkMat(x)$
in~(\ref{eq.polyParChkB}) one can confirm this conclusion using the
logic of Sec.~\ref{sec.findCyclesIntro}.  In particular, $x^6$ is
located in the sixth column of the first row of $\parChkMat(x)$, while
the polynomial $x^2$ traversed by the path is located in the fifth
column of the fifth row of $\parChkMat(x)$. \hfill \QED

\subsection{Inevitable cycles in HQC LDPC codes}
\label{sec.HQCLDPCinevitable}

Since HQC LDPC codes are a generalization of QC LDPC codes, they also
have inevitable cycles.  In this section we describe how the logic and
results of Sec.~\ref{sec.inevitable} regarding inevitable cycles
extend to HQC LDPC codes.  We illustrate the logic for specific
examples of HQC LDPC codes that we will use in our design pipeline
presented in Sec.~\ref{sec.pipeline}.

Recall that in Sec.~\ref{sec.inevitable} we discussed two classes of
inevitable cycles.  We first saw that there will inevitably be cycles
of length six in any weight-III QC LDPC code.  We also saw that the
code will have eight-cycles if the polynomial parity check matrix
$\parChkMat(\cdot)$ of a weight-II QC LDPC code contains two
weight-two polynomials in any row or in any column.  

We analogously find that there will inevitably be cycles of length six
for an HQC LDPC code if any labeled tree $\conMat_{j,l}$ in the
tree matrix defining the code has
{\em three leaves}.  There will inevitably 
be eight-cycles if, in any row or in any column
of the matrix of labeled trees defining the HQC LDPC code, there is a {\em pair} of
labeled trees {\em both having two leaves}.

The logic behind these statements is almost identical to the earlier
case.  We describe it completely for the first situation.  We pick a
length-six ordered series $\ordSet$ equal to~(\ref{eq.ordSerExB}),
i.e., $\ordSet = \{(j,l),(j,l),(j,l),(j,l),(j,l),(j,l)\}$, where
$(j,l)$ is the index of the labeled tree $\conMat_{j,l}$ that has
three leaves.  Let the three length-$K$ coefficient vectors correspond
to the three leaves be ${\bf s}_a, {\bf s}_b, {\bf s}_c$ and select
the coefficient set $\pathCoeff = \{{\bf s}_a, {\bf s}_b, {\bf s}_c,
{\bf s}_a, {\bf s}_b, {\bf s}_c\}$.  Then, because each element is
both an even and an odd element of the set, (\ref{GirthCon}) is
satisfied for every $k$, just as it was in the QC LDPC example
of~(\ref{eq.automaticCycleEx}).  The logic for automatic eight-cycles
follows from the analogous extensions of the choices made
in~(\ref{eq.ordSerExC}) and~(\ref{eq.pathCoeffExC}).

We now illustrate these points about inevitable cycles for a subclass
of two-level HQC LDPC codes that are described solely by labeled trees
with weight-one at the bottom level.  That is, none of the leaves of the
trees have siblings.  In Sec.~\ref{sec.restricted} we name such codes
{\em restricted two-level HQC LDPC codes}.  An example of such a tree
is given in Fig.~\ref{fig.restricted}
\begin{figure}
\centering
\includegraphics[width=3.0in]{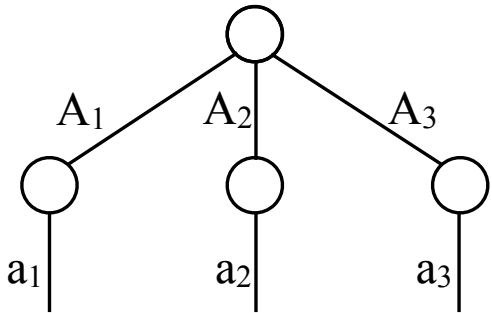}
\caption{The labeled trees in a restricted two-level HQC LDPC code
  will all have two levels, with each node at the bottom level having
  exactly one leaf below it.}
\label{fig.restricted}
\end{figure}

\example {\em (Inevitable length-six cycle in HQC LDPC
  codes)} \label{ex.inevitableSixCycle} First consider any code
containing a tree of the type illustrated in
Fig.~\ref{fig.restricted}.  This code has three leaves and so,
according to our discussion, the code must contain six cycles.
Without loss of generality, let such a tree be located in row $j$ and
column $l$ of the parity check matrix $\parChkMat(x,y)$.  The
polynomial has the
form
\begin{equation*}
x^{a_1} y^{A_1} + x^{a_2} y^{A_2} + x^{a_3} y^{A_3}.
\end{equation*}
As discussed above, choose the ordered series $\ordSet$ to be
\begin{equation}
\ordSet = \{(j,l),(j,l),(j,l),(j,l),(j,l),(j,l)\} 
\end{equation}
and the ordered set of coefficient vectors
to be
\begin{equation}
\pathCoeff = \left\{ \tightArray{a_1}{A_1}\!,
\tightArray{a_2}{A_2}\!,
\tightArray{a_3}{A_3}\!,
\tightArray{a_1}{A_1}\!,
\tightArray{a_2}{A_2}\!,
\tightArray{a_3}{A_3}\right\}. \label{eq.sixCycleInevPath}
\end{equation}
Cycles inevitably exist because
\begin{align*}
(-a_1 + a_2 - a_3 + a_1 - a_2 +a_3) \; \;  \textmd{mod} \; \; p_{[1]}  & = 0,\\ 
(-A_1 + A_2 - A_3 + A_1 - A_2 +A_3) \; \;  \textmd{mod} \; \; p_{[2]}  & = 0,
\end{align*}
regardless of the values of the coefficients or of $p_{[1]}$ or
$p_{[2]}$. 
\hfill \QED

\example {\em (Inevitable length-eight cycle in HQC LDPC
  codes)} \label{ex.inevitableEightCycle} Now suppose that the parity
check matrix of a restricted two-level HQC LDPC code contains two
labeled trees in the same row or column where both trees are similar to
the one depicted in Fig.~\ref{fig.restricted}, but with only two
leaves each.  

Suppose that the two weight-two polynomials are in
the same row $j$ but in two different columns $l_1$ and $l_2$.  Let
the polynomial at $(j,l_1)$ be $x^{a_1} y^{A_1} + x^{a_2} y^{A_2}$ and
let the polynomial at $(j,l_2)$ be $x^{b_1} y^{B_2} + x^{b_2}
y^{B_2}$. Consider the same ordered series as
in~(\ref{eq.ordSerExC}), i.e., $\ordSet =
\{(j,l_1),(j,l_l),(j,l_2),(j,l_2),(j,l_1),(j,l_1),(j,l_2),(j,l_2)\}$,
and choose the ordered set of coefficient indices to be
\begin{equation}
\pathCoeff = \left\{ \tightArray{a_1}{A_1}\!, \tightArray{a_2}{A_2}\!,
\tightArray{b_1}{B_1}\!, \tightArray{b_2}{B_2}\!,
\tightArray{a_2}{A_2}\!, \tightArray{a_1}{A_1}\!,
\tightArray{b_2}{B_2}\!, \tightArray{b_1}{B_1}
\right\}. \label{eq.eightCycleInevPath}
\end{equation}
Eight cycles are inevitable because
\begin{align*}
(- \! a_1 \! + \! a_2 \! - \! b_1 \! + \! b_2 \! - \! a_2 \! + \! a_1
  \! - \! b_2 \! + \! b_1) \; \textmd{mod} \; p_{[1]} & = 0,\\ (- \!
  A_1 \! + \! A_2 \! - \! B_1 \! + \! B_2 \! - \!  A_2 \! + \! A_1 \!
  - \! B_2 \! + \! B_1) \; \textmd{mod} \; p_{[2]} & = 0,
\end{align*}
regardless of the values of the coefficients or of $p_{[1]}$ or
$p_{[2]}$. 
\hfill \QED
\section{Maximizing the girth of QC LDPC codes}
\label{sec.girthMax}

In this section we present the ideas behind our girth-maximizing
algorithms for QC LDPC and for HQC LDPC codes.  The latter is a
generalization of the former.  These algorithms can rid the codes of
all non-inevitable cycles.  In Sec.~\ref{sec.pipeline} we will describe a
secondary procedure for ridding the codes of their inevitable cycles.
As the details of the algorithms are somewhat involved, we choose only
to describe the basic ideas in the main text, and defer to the
appendices the details.  The overall algorithms
are described in Appendix~\ref{sec.girthMaxAlgs} while in
Appendices~\ref{app.maxGirth}--\ref{app.costHQCLDPC} we describe the
subroutines that contain much of the computational complexity (and
descriptive intricacies).

In Section~\ref{sec.genHillIdea} we describe the general idea of the
algorithms, which applies both to QC and to HQC LDPC codes.  Then, in
Section~\ref{sec.hillNonHier} we give more detail for the case of
weight-I QC LDPC codes.  The discussion of the generalization to HQC
LDPC codes (which includes higher-weight QC LDPC codes as a special
case) is deferred to the appendices.

\subsection{Girth maximization using hill climbing}
\label{sec.genHillIdea}

The general idea of our algorithm (for both QC and HQC LDPC codes) is
as follows.  We start by specify the desired tree topology of the code
by specifying a set of $|\setTree|$ {\em unlabeled} trees.  We
initialize our algorithm with a code chosen randomly from the ensemble
of codes that have the desired tree topology.  This means that we
randomly assign labels to the trees $\setTree$ subject to the
constraints that sibling edges must have distinct labels.

Our algorithm iteratively updates a sequence of edge labels.  At each
iteration it changes the single edge label to the value that effects
the greatest reduction in a cost function.  The cost function we use
depends on the number of cycles in the current code that have length
less than the desired girth.  Shorter cycles are weighted to be more
costly than longer cycles. The algorithm terminates when either (a)
the current values of all coefficients give zero cost (and thus the
code has the desired girth), or (b) when we can no longer change any
single coefficient to a value that further reduces the cost (and thus
the number of undesired cycles).  When the tree topology of the code
implies the existence of inevitable cycles (b) will always be the
stopping criterion.  Updates are performed subject to the sibling
constraint on edge labels.  This preserves the tree topology of the
code and thus, e.g., the protograph structure of the code is an
invariant under the updates.  We note that a change in a single edge
label will, in general, have a trickle-down effect on a number of code
coefficients (equal to the number of leaves in the tree that are a
descendent of that edge).

The main challenge in implementing the algorithm lies in book-keeping:
tracking how many cycles of each length the current code contains, and
what the resulting number of cycles will be if each edge label is
changed to each of its other possible value. The calculation becomes
particularly involved when one searches for codes of girth $10$ (which
is the largest girth for which we have so far implemented our
algorithm) because of the many possible ways that eight-cycles can
form.

%%%%%%%%%%%%%%%%%%%%%%%%%%%%%%%%%%%
\subsection{Girth maximizing algorithm for QC LDPC codes}
\label{sec.hillNonHier}

In this section we present the main algorithmic ideas in the
simplified setting of weight-I QC LDPC codes.  This simplification also reduces
notation.  For the duration of this section, we set $\p{1} = p$, $\Lcol{1}
= L$, $\Jrow{1} = J$.  Further, path elements are scalars so
$\bfInd{j}{l} = \indSym[j,l]$.  In a weight-I QC LDPC each tree
$\conMat_{j,l}$ has a single edge and $\sym{\indSym}{j}{l} \neq 0$ for
at most one value of $\indSym$ (if $\conMat_{j,l} = \nullTree$ then
$\sym{\indSym}{j}{l} = 0$ for all $\indSym$).  The set of {\em other}
possible edge labels are the set of $z$, $0 \leq z \leq p-1$, such
that $z \neq \indSym$ (there are no sibling edges so there are no
further constraints on the choice of $z$).

We now define a set of cost vectors, each of which tracks the cost (in
terms of the weighted sum of the number of cycles) of changing any
edge label to each of its other possible values.  In particular, for
each edge in each $\conMat_{j,l} \neq \nullTree$ we
define
\begin{equation}
\costB{j,l}=[\cost{0}, \cost{1}, \cdots,
  \cost{p-1}], \label{def.costVec}
\end{equation}
where $\costB{j, l}(z)$ is the cost we pay for assigning
$\sym{z}{j}{l} = 1$ for each value of $z$ for $0 \leq z \leq p-1$.  If
the desired code girth is $\girth$ then the cost $\costB{j,l}$ is a
linear function of the number of cycles of each length that results
from each possible choice for $z$.  The weight vector $\weight=[w_2,
  w_3, \cdots, w_{\girth/2-1}]$ defines the cost function, where
$w_{\len}$ is the cost assigned to each length-$2\len$ cycles.

It is useful to visualize the set of cost vectors as a matrix of
vectors.  For example, a regular $(3,6)$ LDPC code can be represented
as
\begin{eqnarray}
\left[
  \begin{array}{llllllll} \costB{1,1} & \costB{1,2} & \costB{1,3} & \costB{1,4} & \costB{1,5} &  
\costB{1,6} \\ \costB{2,1} & \costB{2,2} & \costB{2,3} & \costB{2,4} &
\costB{2,5} & \costB{2,6} \\ \costB{3,1} & \costB{3,2} & \costB{3,3} &
\costB{3,4} & \costB{3,5} & \costB{3,6}\end{array} \right]. \label{eq.costMatrix}
\end{eqnarray}

Given a parity check matrix $\parChkMat$ and desired girth $g$, the
cost vectors are calculated via the following argument. We consider
the set of all possible and distinct length-$2\len$ paths per
Def.~\ref{def.path}, i.e.,
\begin{equation*}
\path_{\len} = \{ \path\} = \{\ordSet, \pathCoeff\}\; \mbox{s.t.} \;
|\ordSet| = |\pathCoeff| = 2 \len \; \mbox{for all} \; \path \in
\path_{\len},
\end{equation*}
for $\len = 1, \ldots, \girth/2$.  For each path $\path \in
\path_{\len}$ and each $(j_t, l_t) \in \ordSet$ we consider the
corresponding coefficient $\indSym[j_t,l_t] \in \pathCoeff$.  Assuming
all other {\em distinct} coefficients $\indSym[j_{t'}, l_{t'}]$ for
$t' \neq t$ are kept fixed we note the ``guilty'' value(s) of
$\indSym[j_t,l_t]$ to be the value $z$, $0 \leq z \leq p-1$, such that
if $\indSym[j_t,l_t]$ were changed to $z$, then
condition~(\ref{GirthCon}) would be satisfied.  In other words, a
cycle would result.

For example, for a potential six-cycle, we know that a cycle will
exist if and only if $- \indSym[j_1,l_1] + \indSym[j_2,l_2] -
\indSym[j_3,l_3] + \indSym[j_4,l_4] - \indSym[j_5,l_5] +
\indSym[j_6,l_6] \;\textmd{mod} \; p = 0$.
Suppose, for example, that the current
summed value of $- \indSym[j_1,l_1] + \indSym[j_2,l_2] -
\indSym[j_3,l_3] + \indSym[j_4,l_4] - \indSym[j_5,l_5] +
\indSym[j_6,l_6] \; \textmd{mod} \; p$ is equal to one.  Then, the
guilty values for $\indSym[j_1,l_1]$, $\indSym[j_3,l_3]$, and
$\indSym[j_5,l_5]$ would be one less than their respective current
values, and the guilty values for $\indSym[j_2,l_2]$,
$\indSym[j_4,l_4]$, and $\indSym[j_6,l_6]$ would be one greater than
their respective current values.

Computing ``guilty'' values is relatively uncomplicated for paths
consisting of $2\len$ {\em distinct} elements.  It becomes more
complicated if some elements of the path appear more than once.  This
can occur in potential eight-cycles and occurs, e.g., in the second
example of Fig.~\ref{fig.findCycle}.  In such cases, we must keep in
mind that when such coefficients are changed, the contribution to
alternating sum can double, triple (e.g., in the length-12 cycle of
Fig.~\ref{fig.findCycle} because the path passes through each
sub-matrix three times), or contribute even more times.  Alternately,
repeated elements can also cancel (if they enter modulated by both
$+1$ and $-1$), not contributing at all to the sum.  We deal with this
complexity in Appendix~\ref{app.maxGirth} by defining the
``multiplicity'' $\kappa$ of a path element; used in the cost
calculating algorithms subsequently specified in
Appendices~\ref{app.costQCLDPC} and~\ref{app.costHQCLDPC}.

\section{Design pipeline for high-girth QC LDPC codes}
\label{sec.pipeline}

In this section we describe our design procedure for high-girth QC
LDPC codes.  We want to be able to map any interesting protograph into
a high-girth QC code.  As mentioned earlier, the protographs that
motivate us have multiple edges between variable types and check nodes
(such as those shown in Fig.~\ref{fig.onesided}) and thus QC LDPC codes
created from them using a direct transformation would suffer from
inevitable cycles.  

In this section we show how to map such a
protograph into an ``inflated'' HQC LDPC code structure, on which
we can use the girth maximizing
algorithm of Sec.~\ref{sec.girthMax} to remove all non-inevitable
cycles.  We then show how the resulting HQC LDPC codes can be
``squashed'' down to yield a non-hierarchical QC LDPC code which no
longer contains the inevitable cycles and which is a member
of the family of codes described by our protograph.  The subclass of HQC LDPC codes
with with we work are the {\em restricted two-level HQC LDPC codes},
already mentioned briefly in the examples of
Sec.~\ref{sec.HQCLDPCinevitable}.

The outline of the section is as follows. In Sec.~\ref{sec.restricted}
we fully define the class of restricted two-level HQC LDPC codes.
In Sec.~\ref{sec.protorestricted} we show how to directly transform
any protograph into such a code. In Sec.~\ref{sec.squashingProcedure}, we describe
the squashing procedure, and finally in Sec.~\ref{sec.squashing}, we explain
the full design pipeline, including ``inflating'' the connectivity matrix corresponding
to the protograph, directly transforming the inflated connectivity matrix
into a family of restricted two-level HQC LDPC codes, maximizing the girth over
that family, and squashing the resulting HQC LDPC code.

\subsection{Restricted two-level HQC LDPC codes}
\label{sec.restricted}

As ``restricted two-level'' implies, the hierarchy in restricted
two-level HQC LDPC codes has only two levels.  The additional
``restriction'' is that the weight of the first (lowest) level must be
one. In terms of the tree structure description of these codes, the
labeled trees will all have a form like that shown in
Fig.~\ref{fig.restricted}, with the nodes at the bottom level each
having exactly one leaf, i.e., leafs have no siblings. In comparison,
there are leaves in left-hand tree of Fig.~\ref{fig.trees} that do
have siblings. Nodes at the second level can have an arbitrary number
of edges.

The fact that these codes have two levels means that they are
described by a polynomial parity check matrix in two dummy variables
$\parChkMat(x,y)$. The restriction to the lowest level having weight
one means that any weight-$w$ polynomial in the matrix
$\parChkMat(x,y)$ must have the form
\begin{equation}
\label{eq.restrictedform}
x^{a_1} y^{A_1} + x^{a_2} y^{A_2} + ... + x^{a_w} y^{A_w}
\end{equation}
where all the $A_i$ exponents must be distinct. As usual, the
exponents are integers which range between $0$ and $p_{[1]}-1$ for the
$x$ exponents and $0$ and $p_{[2]}-1$ for the $y$ exponents.

Because the weight at the lowest level is restricted to be one, these
codes, when described as standard QC LDPC codes, look like weight-I QC
LDPC codes, whose base matrix is composed of circulant sub-matrices of
size $p_{[2]}$ by $p_{[2]}$. In~\cite{vontobelRefs} Smarandache and
Vontobel briefly introduce a further restricted class of such QC LDPC
codes (they also required that the codes be weight-II at the second
level and that $p_{[2]}=2$) which they term ``type-I QC codes based on
double-covers of type-II QC codes.''

\subsection{Transforming protographs into Restricted Two-Level HQC LDPC Codes}
\label{sec.protorestricted}

Recall that in Sec.~\ref{sec.protographs}, we introduce a ``direct
transformation'' to convert a protograph into an ordinary QC LDPC
code. The direct transformation replaces the connectivity matrix
equivalent to the protograph with a polynomial parity check matrix
$\parChkMat(x)$ whose polynomial entries had weight equal to the
entries in the connectivity matrix. A completely analogous direct
transformation exists for converting protographs into restricted
two-level HQC LDPC codes. One replaces the connectivity matrix with a
bi-variate polynomial parity check matrix $\parChkMat(x,y)$ whose
polynomial entries each have the restricted form
of~(\ref{eq.restrictedform}) and have weight equal to the
entries in the connectivity matrix.

For example, the connectivity matrix corresponding to the protograph
depicted in Fig.~\ref{fig:proto1} is
\begin{equation}
C = \left[ \begin{array}{cccc} 
1 & 1 & 1 \\
0 & 1 & 2 \end{array}\right]. \label{eq.mapToConnectMat}
\end{equation}
This matrix is directly transformed into a two-level restricted HQC
LDPC code with  polynomial parity check matrix
\begin{equation}
\parChkMat(x,y) = \left[ \begin{array}{cccc} 
x^a y^A & x^b y^B & x^c y^C \\
0 & x^d y^D & x^e y^E + x^f y^F \end{array}\right],
\end{equation}
where $a$, $b$, $c$, $d$, $e$, and $f$ are integer exponents between
$0$ and $p_{[1]}-1$, and $A$, $B$, $C$, $D$, $E$, and $F$ are integer
exponents between $0$ and $p_{[2]}-1$ that satisfy $E \ne F$.

\subsection{Squashing sets of trees to eliminate inevitable cycles}
\label{sec.squashingProcedure}

Because restricted two-level HQC LDPC codes are weight-I at the lowest level,
they can also be considered weight-I QC LDPC codes, and can therefore be described
in terms of their base matrix. \footnote{Recall
  from Sec.~\ref{sec.defQC_LDPC} that the base matrix is the matrix of
  powers of the polynomial parity check matrix expressed in a single dummy
  variable.}
In this section we develop a technique that selectively removes rows or columns
from the base matrix describing a restricted two-level HQC LDPC code 
in a way that eliminates all inevitable six- and
eight-cycles from the corresponding Tanner graphs of the code.  There are
two underlying assumptions in this section.  First, that via a
girth-maximization procedure the base matrix entries involved have already been
optimized to eliminate all non-inevitable cycles.  Second, we
concentrate on restricted two-level HQC LDPC codes where $\p{2} = 4$, which implies
that the base matrix is composed of circulant sub-matrices of size four by four.

There are two situations we will want to consider.  Respectively they will
correspond to Ex.~\ref{ex.inevitableSixCycle}
and~\ref{ex.inevitableEightCycle} of Sec.~\ref{sec.HQCLDPCinevitable}.
The full connection to these examples will only become clear in the next section, when
we explain our ``inflation'' procedure, which has the effect of placing pairs of similarly
structured
four by four sub-matrices on top of each other (or besides each other).

The first situation involves a polynomial of weight 3 in the polynomial parity
check matrix \parChkMat(x,y), which after inflation will be converted
into two polynomials of weight 3, e.g., $h_{1,1}(x,y)$ and $h_{2,1}(x,y)$,  
with identical $y$ exponents, in the
same column  of the polynomial parity check matrix.
Assuming a restricted two-level code
with $\p{2} = 4$, the corresponding sub-matrices in the base matrix
would respectively
look something like
\begin{equation}
\begin{array}{c} \left[ \begin{array}{cccc} 
a &  b  &  c  & -1 \\
-1 &  a  &  b  &  c \\
c  &  -1 &  a  &  b \\
b  &  c  & -1  &  a\end{array}\right] \vspace{2ex}\\
\left[ \begin{array}{cccc} 
d  &  e  &  f  & -1 \\
-1 &  d  &  e  &  f \\
f  &  -1 &  d  &  e \\
e  &  f  & -1  &  d\end{array}\right]
\end{array}, \label{eq.squashSixCycle}
\end{equation}
where we recall that $-1$ represents the $\p{1} \times \p{1}$
all-zeros matrix.  

The second situation involves four polynomials of weight 2 arranged
rectilinearly, e.g., $h_{1,1}(x,y)$, $h_{2,1}(x,y)$, $h_{1,2}(x,y)$
and $h_{2,2}(x,y)$,.  Furthermore, after inflation, 
the $y$ exponents of the polynomials in the
same column will have the same exponents, so that 
the corresponding sub-matrices would look
something like
\begin{equation}
\begin{array}{cc}
\left[ \begin{array}{cccc} 
a &  b  &  -1  & -1 \\
-1 &  a  &  b  &  -1 \\
-1  &  -1 &  a  &  b \\
b  &  -1  & -1  &  a\end{array}\right] & 
\left[ \begin{array}{cccc} 
-1 & c  &  d  &  -1  \\
-1 & -1 &  c  &  d  \\
d & -1  &  -1 &  c  \\
c & d  &  -1  & -1  \end{array}\right] \vspace{2ex} \\
\left[ \begin{array}{cccc} 
e &  f  &  -1  & -1\\
-1 &  e  &  f  & -1 \\
-1  &  -1 &  e  & f \\
f  &  -1  & -1  & e \end{array}\right] & 
\left[ \begin{array}{cccc} 
-1 & g &  h  &  -1  \\
-1 & -1 &  g  &  h  \\
h & -1  &  -1 &  g  \\
g & h  &  -1  & -1 \end{array}\right] 
\end{array}. \label{eq.squashEightCycle}
\end{equation}

By the results of Sec.~\ref{sec.HQCLDPCinevitable} the first situation
contains six-cycles within each sub-matrix and the second situation
contains inevitable eight-cycles between the pair of sub-matrices in each
row and in each column.  We argue that if we ``squash'' the two
matrices in the first example---by stacking the first two rows of the
upper matrix on the last two rows of the lower matrix---then the matrix produced
\begin{equation}
\left[ \begin{array}{cccc} 
a &  b  &  c  & -1 \\
-1 &  a  &  b  &  c \\
f  &  -1 &  d  &  e \\
e  &  f  & -1  &  d\end{array}\right] \label{eq.postSquashSixCycle}
\end{equation}
contains no six-cycles.  Similar if we squash the matrices in the second
example then the resulting pair of matrices
\begin{equation}
\left[ \begin{array}{cccc} 
a &  b  &  -1  & -1 \\
-1 &  a  &  b  &  -1 \\
-1  &  -1 &  e  &  f \\
f  &  -1  & -1  &  e\end{array}\right] \hspace{1em}
\left[ \begin{array}{cccc} 
-1 & c  &  d  & -1 \\
-1 & -1 &  c  &  d  \\
h & -1  &  -1 &  g  \\
g & h  &  -1  & -1  \end{array}\right] \label{eq.postSquashEightCycle}
\end{equation}
contains no eight-cycles.

Since by assumption there were no non-inevitable six- or eight-cycles
between the original matrices, to show our assertion we need solely to
demonstrate that the squashing procedure removes all inevitable
cycles.  We argue this based on the following lemma, proved in
Appendix~\ref{app.squash}.
\begin{lemma} \label{lemma.squash} \hspace{1em}\\
\vspace{-3ex}
\begin{itemize}
\item[(i)] Any inevitable six-cycle within a polynomial of the form
  $x^{a_1} y^{A_1} + x^{a_2} y^{A_2} + x^{a_3} y^{A_3}$ traverses
  three distinct rows and three distinct columns of the corresponding
  base matrix.

\item[(ii)] Any inevitable eight-cycles between a pair of polynomials
  of the form $x^{a_1} y^{A_1} + x^{a_2} y^{A_2}$ and $x^{b_1} y^{B_1}
  + x^{b_2} y^{B_2}$ located in the same row (column) of the
  polynomial parity check matrix traverses three distinct rows
  (columns) of the corresponding base matrix.  \hfill \QED
\end{itemize}
\end{lemma}

Now, consider the squashing of the matrices
in~(\ref{eq.squashSixCycle}) into the matrix
in~(\ref{eq.postSquashSixCycle}).  Note that the latter matrix has
only {\em two} rows from each of the matrices
in~(\ref{eq.squashSixCycle}).  However, by
Lemma~\ref{lemma.squash}-(i) all inevitable cycles pass through three
rows.  Therefore, the matrix in~(\ref{eq.postSquashSixCycle}) does not
contain any inevitable six-cycles.

Next, consider the squashing of the matrices
in~(\ref{eq.squashEightCycle}) into the matrices
in~(\ref{eq.postSquashEightCycle}).  First we a show that the
squashing procedure eliminate the automatic cycles between pairs of
matrices arising from pairs of weight-2 polynomials on the same row of the
polynomial parity check matrix.  This follows from
Lemma~\ref{lemma.squash}-(ii) which tells us that these eight-cycles
traverse three distinct rows, because only two rows of each of the
matrices is retained.  Next consider the inevitable cycles between
pairs of matrices arising from pairs of weight-2 polynomials 
in the same column of the polynomial parity check matrix.  Since we squash vertically,
parts of all columns of the base matrix are retained.  However, if one
examines~(\ref{eq.postSquashEightCycle}) one sees that the second and
fourth columns of the left-hand matrix only includes contributions
from the upper left-hand and bottom-left-hand matrices
of~(\ref{eq.squashEightCycle}), respectively.  The remaining
inevitable cycles from~(\ref{eq.squashEightCycle}) therefore cannot
include these columns.  But, that leaves only two columns in the
left-hand matrix and by Lemma~\ref{lemma.squash}-(ii) we know that
these inevitable cycles require three columns.  Therefore the
inevitable cycles have been eliminated.  The same logic holds for the
right-hand side of~(\ref{eq.postSquashEightCycle}).

Note that for the above logic regarding eight-cycles to hold it is
important that the $y$-exponents of the two matrices to be squashed
together (those in the same column) are the same.  Thus, e.g.,
$\conMat_{1,1}$ $\conMat_{2,1}$ correspond to any two polynomials of
the form $x^ay^A + x^b y^B$ and $x^e y^A + x^f y^B$.  Note also that
the same squashing procedure would work in the horizontal direction as
long as the matrices on the same row have the same $y$-exponents.  The
logic is the same with the argument for rows and columns reversed.

\subsection{Design procedure for high-girth codes}
\label{sec.squashing}

We now turn to demonstrating how to construct a weight-I QC LDPC code
that does not have any six-cycles or eight-cycles.  We first sketch
the procedure, depicted in Fig.~\ref{fig.designPipeline} and then
illustrate the details with a worked design example.

\begin{figure}
\psfrag{&inputs}{inputs:  $\p{1}$,}
\psfrag{&params}{connectivity matrix $C$}
\psfrag{&inflate}{1) Inflate}
\psfrag{&directMap}{2) Direct transform}
\psfrag{&maxGirth}{3) Maximize girth}
\psfrag{&squash}{4) Squash}
\psfrag{&outputs}{code}
\centering
\includegraphics[width=8.5cm]{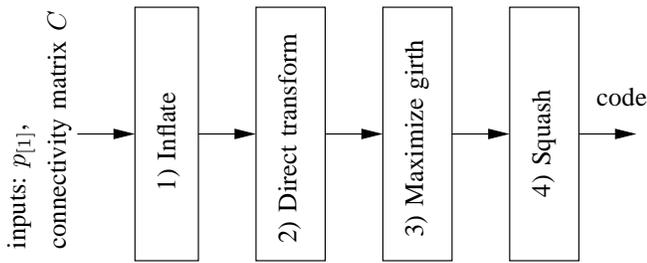}
\caption{The design procedure to produce high-girth codes.  The inputs
  are a protograph and the dimension of the first-level circulant
  matrices. The output is a weight-I QC LDPC code.}
\label{fig.designPipeline}
\end{figure}

Roughly speaking the procedure will start with a desired protograph
and code parameter $\p{1}$ (our procedure assumes $\p{2} = 4$).  We
first map the protograph into a connectivity matrix,
cf.~(\ref{eq.mapToConnectMat}).  Depending on the weight and relative
locations of the entries in the connectivity matrix, we ``inflate''
the connectivity matrix.  We then use the direct mapping of
Sec.~\ref{sec.protorestricted} to produce a polynomial parity check
matrix for a restricted two-level HQC LDPC code.  Next, using our max
girth algorithm we eliminate all non-inevitable six- and eight-cycles.
Finally, we use the squashing procedure of
Sec.~\ref{sec.squashingProcedure} to eliminate inevitable cycles.  Of
course, the way in which we inflate the code must be compatible with
the way we squash the code to produce a valid parity check matrix that
meets the parameters of interest.  It should be emphasized that the
LDPC code resulting from this procedure will be a QC LDPC and not a
{\em hierarchical} QC LDPC code, although the final structure will be
quite similar to that of an HQC LDPC code.

{\bf 1) Inflate connectivity matrix:} As indicated, the procedure first
produces the connectivity matrix $C$ of the protograph, which we
assume has no entries greater than 3. (We make no effort here to deal
with inevitable cycles caused by weights greater than 3). The
``inflation'' procedure works as follows.  We fist mark for
duplication each row of the matrix with two or more elements of value
2 or greater or a single element of value 3.  We also mark for duplication each column
that has two or more elements of value 2 or greater. Then we inflate
$C$ to produce a new connectivity matrix $C'$ in which each of the
rows in $C$ marked for duplication are duplicated.  We then inflate
again to produce $C''$ from $C'$ by duplicating each of the marked
columns. As will be evident when we get to squashing, we must track in
the matrices $C'$ and $C''$ which rows and columns are duplicated
versions of each other.  The following example illustrates the
inflating procedure.

\begin{example}
\label{ex.squashing}
Suppose we start with a protograph that has the connectivity matrix
\begin{equation}
C = \left[ \begin{array}{cccc} 
3 & 2 & 1 \\
0 & 2 & 1 \\\end{array}\right].
\label{eq:Cmatrix2}
\end{equation}
The first row in this connectivity matrix contains an element with
value 3 (and also two elements of value 2 or greater), so we mark it,
and we also mark the second column because it has two elements with
value 2 or greater. Duplicating the first row, we obtain
\begin{equation}
C' = \left[ \begin{array}{cccc} 
3 & 2 & 1\\
3 & 2 & 1\\
0 & 2 & 1\end{array}\right].
\label{eq:Cmatrix2fat}
\end{equation}
Now duplicating the second column, we obtain the inflated connectivity
matrix
\begin{equation}
C'' = \left[ \begin{array}{cccc} 
3 & 2 & 2 & 1\\
3 & 2 & 2 & 1\\
0 & 2 & 2 & 1\end{array}\right].
\label{eq:Cmatrix2inflated}
\end{equation}
In $C''$, the first and second rows, and also the second and third
columns, are tracked as duplicated versions of each other. \hfill \QED
\end{example}
\vskip .1cm

{\bf 2) Directly transformation $C''$ into $\parChkMat''(x,y)$:} Next we
make a direct transformation of the inflated connectivity matrix $C''$
into the polynomial parity check matrix $\parChkMat''(x,y)$ for a
two-level restricted HQC LDPC code with $p_{[2]}=4$.  We perform this
transformation under one additional restriction.  The restriction is
that the $y$ exponents in pairs of duplicated rows or pairs of
duplicated columns must be identical to each other. The value of
$p_{[1]}$ is left as a design parameter.

{\em Example~\ref{ex.squashing} (continued):} The inflated
connectivity matrix $C''$ is directly transformed into a polynomial
parity check matrix $\parChkMat''(x,y)$, yielding the form
\begin{equation*}
\parChkMat''(x,y)=
\left[ \begin{smallmatrix} 
x^a y^A +  x^b y^B + x^c y^C & x^d y^D + x^e y^E & x^f y^D + x^g y^E & x^h y^H\\
x^i y^A +  x^j y^B + x^k y^C & x^l y^D + x^m y^E & x^n y^D + x^o y^E & x^p y^H\\
0 &  x^q y^Q + x^r y^R & x^s y^Q + x^t y^R & x^u y^U \end{smallmatrix}\right].
%\label{eq.Hxyinflatedexponents}
\end{equation*}
Notice that the $y$ exponents in the first and second row and in the
second and third columns of this matrix have been restricted to be
identical to each other. Otherwise, all the exponents are free
parameters that satisfy $0 \le a_i \le p_{[1]}-1$ for any $x$ exponent
$a_i$ and $0 \le A_i \le p_{[2]}-1 = 3 $ for any $y$ exponent
$A_i$. \hfill \QED
\vskip .1cm

{\bf 3) Maximize the code's girth:} In the next step we apply the
girth-maximization algorithm of Sec.~\ref{sec.girthMax} to produce a
set of $x$-exponents $a_i$ and $y$-exponents $A_i$ such that no short
cycles exist except those that are inevitable.  Of course, the
hill-climbing algorithm of Sec.~\ref{sec.girthMax} is just one
possible approach.  Other algorithms could be used in its place. The
polynomial parity-check matrix $\parChkMat''(x,y)$ obtained in this
manner can be converted into an equivalent base matrix $\baseMat''$
for a weight-I QC LDPC code.

{\em Example~\ref{ex.squashing} (continued):} Using our
girth-maximizing algorithm, we find that with $p_{[1]}=200$ the
following choices for the $x$ and $y$ exponents in $\parChkMat''(x,y)$
will create no six-cycles or eight-cycles except for inevitable short
cycles:
\begin{equation*}
\left[ \begin{smallmatrix} 
x^{128} y^1 +  x^{69} y^2 + x^{118} y^3 & x^{11} y^2 + x^{121} y^3 & x^{170} y^2 + x^{109} y^3 & x^{38} y^3\\
x^{63} y^1 +  x^{156} y^2 + x^{38} y^3 & x^{186} y^2 + x^{183} y^3 & x^{52} y^2 + x^{146} y^3 & x^{43} y^3\\
0 &  x^{100} y^0 + x^{104} y^1 & x^{187} y^0 + x^{50} y^1 & x^{59} y^1 \end{smallmatrix}\right].
%\label{eq.inflatedH''xy}
\end{equation*}
The code with the above polynomial parity check matrix is equivalent
to a standard weight-I QC LDPC code with base matrix $\baseMat''$
given by
\begin{equation}
\left[ \begin{smallmatrix}
-1 & 128 & 69 & 118 & -1 & -1 & 11 & 121 & -1 & -1 & 170 & 109 & -1 & -1 & -1 & 38 \\
118 & -1 & 128 & 69 & 121 & -1 & -1 & 11 & 109 & -1 & -1 & 170 & 38 & -1 & -1 & -1 \\
69 & 118 & -1 & 128 & 11 & 121 & -1 & -1 & 170 & 109 & -1 & -1 & -1 & 38 & -1 & -1 \\
128 & 69 & 118 & -1 & -1 & 11 & 121 & -1 & -1 & 170 & 109 & -1 & -1 & -1 & 38 & -1 \\
-1 & 63 & 156 & 38 & -1 & -1 & 186 & 183 & -1 & -1 & 52 & 146 & -1 & -1 & -1 & 43 \\
38 & -1 & 63 & 156 & 183 & -1 & -1 & 186 & 146 & -1 & -1 & 52 & 43 & -1 & -1 & -1 \\
156 & 38 & -1 & 63 & 186 & 183 & -1 & -1 & 52 & 146 & -1 & -1 & -1 & 43 & -1 & -1 \\
63 & 156 & 38 & -1 & -1 & 186 & 183 & -1 & -1 & 52 & 146 & -1 & -1 & -1 & 43 & -1 \\
-1 & -1 & -1 & -1 & 100 & 104 & -1 & -1 & 187 & 50 & -1 & -1 & -1 & 59 & -1 & -1 \\
-1 & -1 & -1 & -1 & -1 & 100 & 104 & -1 & -1 & 187 & 50 & -1 & -1 & -1 & 59 & -1 \\
-1 & -1 & -1 & -1 & -1 & -1 & 100 & 104 & -1 & -1 & 187 & 50 & -1 & -1 & -1 & 59 \\
-1 & -1 & -1 & -1 & 104 & -1 & -1 & 100 & 50 & -1 & -1 & 187 & 59 & -1 & -1 & -1 
\end{smallmatrix}\right].
\label{eq.Binflated}
\end{equation}
Notice that the base matrix $\baseMat''$ is composed of $4$ by $4$
circulant sub-matrices.  \hfill \QED
\vskip .1cm

{\bf 4) Squash the base matrix to remove inevitable cycles:} We now have
a base matrix $\baseMat''$ corresponding to the inflated connectivity
matrix $C''$. The next steps in our procedure will remove columns and
rows from $\baseMat''$ to obtain a base matrix corresponding to our
original connectivity matrix $C$.

First, we note that each column of the connectivity matrix $C''$
corresponds to four columns in the base matrix $\baseMat''$. 
In the next step of our procedure,
we focus on the columns that have been marked as duplicates in $C''$.
 We
retain the left two columns and remove the right two columns
in $\baseMat''$ from the four that
 correspond to the left column of a duplicated pair in 
$C''$, and also remove the left two columns but 
retain the right two columns in $\baseMat''$ from
the four that correspond
to the right column of a duplicated pair in $C''$. We call the 
thinned-out base matrix that is obtained from this procedure $\baseMat'$.

{\em Example~\ref{ex.squashing} (continued):} 
Recall that the second and third columns of $C''$ given
in equation (\ref{eq:Cmatrix2inflated})
have been marked as
duplicates of each other. So to obtain $\baseMat'$ from the base matrix $\baseMat''$
given in equation (\ref{eq.Binflated}), we retain the left two columns from the second
four in $\baseMat'$, and the right two columns from the third four in $\baseMat'$,
so that $\baseMat'$ is given by

\begin{equation*}
\baseMat' = 
\left[ \begin{smallmatrix}
-1 & 128 & 69 & 118 & -1 & -1 & 170 & 109 & -1 & -1 & -1 & 38 \\
118 & -1 & 128 & 69 & 121 & -1 & -1 & 170 & 38 & -1 & -1 & -1 \\
69 & 118 & -1 & 128 & 11 & 121 & -1 & -1 & -1 & 38 & -1 & -1 \\
128 & 69 & 118 & -1 & -1 & 11 & 109 & -1 & -1 & -1 & 38 & -1 \\
-1 & 63 & 156 & 38 & -1 & -1 & 52 & 146 & -1 & -1 & -1 & 43 \\
38 & -1 & 63 & 156 & 183 & -1 & -1 & 52 & 43 & -1 & -1 & -1 \\
156 & 38 & -1 & 63 & 186 & 183 & -1 & -1 & -1 & 43 & -1 & -1 \\
63 & 156 & 38 & -1 & -1 & 186 & 146 & -1 & -1 & -1 & 43 & -1 \\
-1 & -1 & -1 & -1 & 100 & 104 & -1 & -1 & -1 & 59 & -1 & -1 \\
-1 & -1 & -1 & -1 & -1 & 100 & 50 & -1 & -1 & -1 & 59 & -1 \\
-1 & -1 & -1 & -1 & -1 & -1 & 187 & 50 & -1 & -1 & -1 & 59 \\
-1 & -1 & -1 & -1 & 104 & -1 & -1 & 187 & 59 & -1 & -1 & -1 
\end{smallmatrix}\right].
%\label{eq.B'}
\end{equation*}
\hfill \QED
\vskip .1cm

Now note that each row in the connectivity matrix $C'$ corresponds to four rows
in the base matrix $\baseMat'$.
In the final step of our procedure, we focus on the rows that have
been marked as duplicates in $C'$. We retain the top two rows
in $\baseMat'$ from the four that correspond to the top row in a duplicated pair in $C'$,
and we retain the bottom two rows in $\baseMat'$ from the four that correspond
to the bottom row of a duplicated pair in $C'$. We call the base matrix obtained by this
further thinning-out procedure $\baseMat$; this is the base matrix that will
correspond to our original connectivity matrix $C$.

{\em Example~\ref{ex.squashing} (continued):} 
The first and second rows of $C'$ given in equation (\ref{eq:Cmatrix2fat}) have
been marked as duplicates. That means that we should retain the top two rows
of the first group of four rows from $\baseMat'$, and the bottom two rows
from the second group of four rows. Thus, we obtain
\begin{equation*}
\baseMat = 
\left[ \begin{smallmatrix}
-1 & 128 & 69 & 118 & -1 & -1 & 170 & 109 & -1 & -1 & -1 & 38 \\
118 & -1 & 128 & 69 & 121 & -1 & -1 & 170 & 38 & -1 & -1 & -1 \\
156 & 38 & -1 & 63 & 186 & 183 & -1 & -1 & -1 & 43 & -1 & -1 \\
63 & 156 & 38 & -1 & -1 & 186 & 146 & -1 & -1 & -1 & 43 & -1 \\
-1 & -1 & -1 & -1 & 100 & 104 & -1 & -1 & -1 & 59 & -1 & -1 \\
-1 & -1 & -1 & -1 & -1 & 100 & 50 & -1 & -1 & -1 & 59 & -1 \\
-1 & -1 & -1 & -1 & -1 & -1 & 187 & 50 & -1 & -1 & -1 & 59 \\
-1 & -1 & -1 & -1 & 104 & -1 & -1 & 187 & 59 & -1 & -1 & -1 
\end{smallmatrix}\right].
%\label{eq.B}
\end{equation*}

Notice that the code defined by the final base matrix $\baseMat$ is not a hierarchical
QC LDPC code, because that base matrix is constructed from 4 by 4 sub-matrices that
are not circulant. Still, the code is a member of the 
class defined by the original protograph. In our example, each group of four
rows and four columns in the base matrix defines a type of check or bit.
So in our example, from the structure of
$\baseMat$, each check of the first type will be connected to three bits
of the first type, and two bits of the second type, and so on, 
just as required by the protograph.

In fact, any code
defined by a base matrix of a form similar to our $\baseMat$, for example of the
form
\begin{equation}
\baseMat = 
\left[ \begin{smallmatrix}
-1 & a_1 & a_2 & a_3 & -1 & -1 & a_4 & a_5 & -1 & -1 & -1 & a_6 \\
a_7 & -1 & a_8 & a_9 & a_{10} & -1 & -1 & a_{11} & a_{12} & -1 & -1 & -1 \\
a_{13} & a_{14} & -1 & a_{15} & a_{16} & a_{17} & -1 & -1 & -1 & a_{18} & -1 & -1 \\
a_{19} & a_{20} & a_{21} & -1 & -1 & a_{22} & a_{23} & -1 & -1 & -1 & a_{24} & -1 \\
-1 & -1 & -1 & -1 & a_{25} & a_{26} & -1 & -1 & -1 & a_{27} & -1 & -1 \\
-1 & -1 & -1 & -1 & -1 & a_{28} & a_{29} & -1 & -1 & -1 & a_{30} & -1 \\
-1 & -1 & -1 & -1 & -1 & -1 & a_{31} & a_{32} & -1 & -1 & -1 & a_{33} \\
-1 & -1 & -1 & -1 & a_{34} & -1 & -1 & a_{35} & a_{36} & -1 & -1 & -1 
\end{smallmatrix}\right],
\label{eq.Barbitrary}
\end{equation}
where the $a_i$ parameters are arbitrary, would also be a member of
the class defined by our protograph.  \hfill \QED

So the question might be raised, why not simply try to find suitable
parameters for a weight-I QC LDPC defined by a base matrix like that
in equation (\ref{eq.Barbitrary}) directly, instead of using the squashing procedure?
This question will be answered in more detail in Section \ref{sec.results}, but the
short answer is that the squashing procedure is more practical because it enforces 
useful additional structure in the base matrix, and thus normally involves far
fewer parameters for the hill-climbing algorithm to optimize. When one tries
to optimize over more parameters, there is a greater chance that the hill-climbing
algorithms will get stuck in an unfortunate local optimum.

\section{Numerical results}
%Results, analysis, and comparison of algorithms}
\label{sec.results}

In this section we present a set of numerical results illustrating
our design methodology and associated performance results.  In
Sec.~\ref{sec.werSims} we present performance results for a pair of
girth-10 one-sided spatially-coupled codes and compare them to those
of girth-6 codes.  In Sec.~\ref{sec.girthMaxVsGuessAndTest} we give a
sense of the effectiveness of the hill-climbing approach to girth
maximization.  We do this by comparing the time required to find a
code of a certain girth by hill-climbing and by the guess-and-test
algorithm~\cite{MarcQC}.  Finally, in Sec.~\ref{sec.effectSquashing}
we demonstrate the effectiveness of the squashing procedure by
comparing to other candidate approaches.

%%%%%%%%%%%%%%%%%%%%%%%%%%%%%
\subsection{Performance of girth-10 QC LDPC codes}
\label{sec.werSims}

In this section we present word-error-rate (WER) and bit-error-rate
(BER) results for a pair of girth-10 one-sided spatially-coupled
codes.  We plot analogous results for girth-6 codes for comparison.
The first code is a rate-$0.45$ length-$8000$ QC LDPC code.  The
protograph structure of the code is a lengthened version of the one
depicted in Fig.~\ref{fig.onesided}(a).  As in that protograph, each
variable has degree three and check nodes have degree six, four or
two.  The protograph of the code we present has 20 variable nodes and
11 check nodes (in contrast the protograph in
Fig.~\ref{fig.onesided}(a) has 14 variable and 8 check nodes).  In
other words $L_{[2]} = 20$ and $J_{[2]} = 11$. The connectivity matrix
of the code is
 \begin{equation*}
C_1=
\left[ \begin{smallmatrix}
1  & 1  & 0  & 0  & 0  & 0  & 0  & 0  & 0  & 0  & 0  & 0  & 0  & 0  & 0  & 0  & 0  & 0  & 0  & 0\\
1  & 1  & 1  & 1  & 0  & 0  & 0  & 0  & 0  & 0  & 0  & 0  & 0  & 0  & 0  & 0  & 0  & 0  & 0  & 0\\
1  & 1  & 1  & 1  & 1  & 1  & 0  & 0  & 0  & 0  & 0  & 0  & 0  & 0  & 0  & 0  & 0  & 0  & 0  & 0\\
0  & 0  & 1  & 1  & 1  & 1  & 1  & 1  & 0  & 0  & 0  & 0  & 0  & 0  & 0  & 0  & 0  & 0  & 0  & 0\\
0  & 0  & 0  & 0  & 1  & 1  & 1  & 1  & 1  & 1  & 0  & 0  & 0  & 0  & 0  & 0  & 0  & 0  & 0  & 0\\
0  & 0  & 0  & 0  & 0  & 0  & 1  & 1  & 1  & 1  & 1  & 1  & 0  & 0  & 0  & 0  & 0  & 0  & 0  & 0\\
0  & 0  & 0  & 0  & 0  & 0  & 0  & 0  & 1  & 1  & 1  & 1  & 1  & 1  & 0  & 0  & 0  & 0  & 0  & 0\\
0  & 0  & 0  & 0  & 0  & 0  & 0  & 0  & 0  & 0  & 1  & 1  & 1  & 1  & 1  & 1  & 0  & 0  & 0  & 0\\
0  & 0  & 0  & 0  & 0  & 0  & 0  & 0  & 0  & 0  & 0  & 0  & 1  & 1  & 1  & 1  & 1  & 1  & 0  & 0\\
0  & 0  & 0  & 0  & 0  & 0  & 0  & 0  & 0  & 0  & 0  & 0  & 0  & 0  & 1  & 1  & 1  & 1  & 1  & 1\\
0  & 0  & 0  & 0  & 0  & 0  & 0  & 0  & 0  & 0  & 0  & 0  & 0  & 0  & 0  & 0  & 1  & 1  & 2  & 2
\end{smallmatrix} \right].
\label{eq.code1Conn}
\end{equation*}
Setting $p_{[2]}=4$ and using our design approach (girth maximization
and squashing) we found a girth 10 QC LDPC code with $p_{[1]} =
100$. The code length is $L_{[2]}\times p_{[2]}\times p_{[1]} = 8000$
and its base matrix, $B_1$, is specified in
App.~\ref{sec.baseMatrixApp}.

The second code is a rate-$1/3$ length-$24000$ QC LDPC code.  The
protograph structure of the second code is a shortened version of the
structure depicted in Fig.~\ref{fig.onesided}(b).  As in that
protograph the variables are all of degree four.  There are six
variable nodes and four check node, i.e., $L_{[2]} = 6$ and $J_{[2]} =
4$ (in contrast the protograph in Fig.~\ref{fig.onesided}(b) has 14
variable and 8 check nodes). The connectivity matrix of the code
 \begin{equation}
C_2=\left[ \begin{matrix}
1 & 1 & 0 & 0 & 0 & 0\\
1 & 1 & 1 & 1 & 0 & 0\\
1 & 1 & 1 & 1 & 1 & 1 \\
1 & 1 & 2 & 2 & 3 & 3
\end{matrix}\right].
\label{eq.code2Conn}
\end{equation}
Again we use $p_{[2]}=4$ and find a girth 10 QC LDPC code with
$p_{[1]} = 1000$. This code's length is $L_{[2]}\times p_{[2]}\times
p_{[1]} = 24000$. The base matrix, $B_2$, of this code is also
specified in App.~\ref{sec.baseMatrixApp}.

In Figs.~\ref{N8000_WER_BER} and~\ref{N24000_WER_BER} we plot the
respective error rate performance for the two codes on the binary
symmetric channel (BSC).  For purposes of comparison we plot analogous
results for three randomly generated girth-6 QC LDPC codes.  These
codes have the same length, same rate, and same non-zero positions in
the base matrix as the girth-10 codes to which they are compared.

\begin{figure}[htbp]
      \centering
      \includegraphics[width=8.75cm,type=eps,ext=.eps,read=.eps]{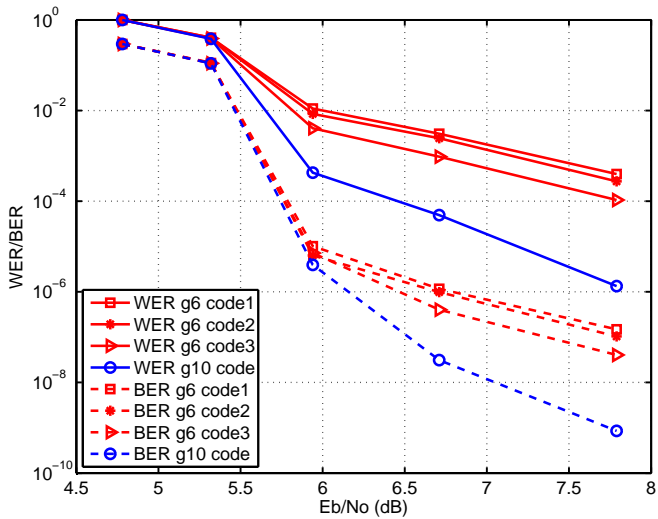}
      \caption{Word- and bit-error rate plots for the rate-$0.45$,
        length-$8000$ girth-6 and girth-10 QC LDPC codes.}
      \label{N8000_WER_BER}
\end{figure}

\begin{figure}[htbp]
      \centering
      \includegraphics[width=8.75cm,type=eps,ext=.eps,read=.eps]{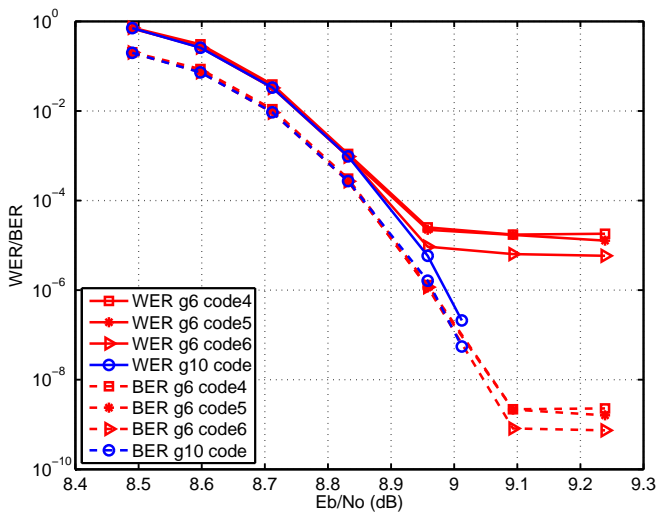}
      \caption{Word- and bit-error rate plots for the rate-$1/3$,
        length-$24000$ girth-6 and girth-10 QC LDPC codes.}
      \label{N24000_WER_BER}
\end{figure}

In all cases, we plot the WER and BER as a function of the
signal-to-noise ratio (SNR), using the Gallager-B decoding
algorithm~\cite{gallager} running for a maximum of $200$ iterations to
guarantee the convergence of decoding.  While there is a significant
difference between the error rates of a standard sum-product decoder
and Gallager-B, the performance trends of Gallager-B and sum-product
are mostly quite similar.  That said, computational complexity is our
main reason to plot results for Gallager-B rather than sum-product.
The error floor of Gallager-B occurs at a higher WER and thus is
easier to attain.  In addition, the Gallager-B algorithm runs very
fast.  This further helps to collect useful statistics about the error
floor regime.

In the plots the SNRs are calculated assuming that the BSC results
from hard-decision demodulation of a binary phase-shift keying (BPSK)
$\pm 1$ sequence transmitted over an additive white Gaussian noise
(AWGN) channel.  The resulting relation between the crossover
probability $p$ of the equivalent BSC-$p$ and the SNR of the AWGN
channel is $ p = \mbox{Q}\left(\sqrt{2 R \cdot 10^{SNR/10}}\right),
\nonumber $ where $R$ is the rate of the code and $\mbox{Q}(\cdot)$ is
the Q-function.

Figure~\ref{N8000_WER_BER} plots the results of the rate-$0.45$
length-$8000$ codes and illustrates the general improvement to error
floor behavior provided by larger girth.  At the highest SNR (around
7.8 dB) the WERs and BERs of the girth-10 code is about two orders of
magnitude larger than those of the girth-6 codes.  Furthermore, the
three girth-6 codes plotted do show some variability in their error
rates.  This illustrates that the error floor is not only a function
of girth, though higher girth certainly helps.

Figure~\ref{N24000_WER_BER} plots the results of the rate-$1/3$
length-$24000$ codes and illustrates some of the same points as were
made for the shorter code, as well as some new ones.  First, we note
that at these lengths the error floor effect is very abrupt,
initiating just below $9$ dB.  Again, higher girth yields a marked
improvement, most clearly seen in the WER plots.  And again, as also
noted in Fig.~\ref{N8000_WER_BER}, we see some variability in the
girth-6 codes.

A new observation comes from observing that the SNRs at which the
error floor of the girth-6 codes becomes noticeable is different for
the WER and BER plots.  It occurs at a higher SNR for BER.  First
consider the highest SNR at which we have results for the girth-6
codes, roughly $9.25$ dB.  Here the difference between WER and BER is
about four orders of magnitude.  Recalling that the codes are of
length $24000$ this means that the post-decoding error patterns in
this regime consist of only a few erroneous bits, consistent with
failures caused by small trapping sets.  In contrast, if we consider
the last data point prior to the error floor, at around $8.95$ dB, the
ratio between word- and bit-error rate is only about one order of
magnitude.  This implies that the BER is still dominated by much
heavier weight error patterns, consistent with the decoder being in
the waterfall regime.  Now, consider the final data point for the
girth-10 code that we were able to obtain at just over $9$ dB.  While
by considering the WER plots of the girth-6 codes we confirm that
those codes are already in their error floor regime, the same is not
true of the girth-10 code, for which the difference between its WER
and BER is less than an order of magnitude.

%%%%%%%%%%%%%%%%%%%%%%%%%%%%%%%%%%
\subsection{Effectiveness of the girth maximization algorithm}
\label{sec.girthMaxVsGuessAndTest}

In this section we develop a sense of how much the hill-climbing type
of girth maximization algorithm presented in Sec.~\ref{sec.girthMax}
helps in finding high girth codes.  We compare our algorithm to the
baseline guess-and-test algorithm~\cite{MarcQC}.  To understand
guess-and-test, consider a regular weight-I QC LDPC code specified by
a $J\times L$ base matrix and a desired girth.  Guess-and-test fixes
all entries in the first row and the first column of the base matrix
to be zero.  The rest it chooses independently and uniformly between
$0$ and $p-1$.  This process is continued until a set is found such
that the condition for the existence of a cycle, e.g.,
(\ref{p4cycle}), is verified not to hold for all cycles shorter than
the desired girth.  The problem with guess-and-test is that it is
time-consuming and doesn't exploit the structure of the cycles in its
search, in contrast to our hill-climbing algorithm.

To make an informative comparison with guess-and-test, we define
``success rate'' to be the fraction of times that a run of either
algorithm (guess-and-test or hill-climbing) results in a base matrix
that has the desired girth and circulant matrices size $\p{1}$.
Figure~\ref{CompareJ3L9} depicts the success rate of the
guess-and-test and hill climbing in generating girth-$8$ weight-I
regular QC LDPC codes with base matrices of size $3 \times 9$. We
observe that for the guess-and-test to find a parity check matrix with
girth-$8$ at a circulant size $\p{1}=50$ we need, on average, to test
$10^6$ random matrices.  In contrast, hill climbing has near certain
success.

\begin{figure}[htbp]
      \centering
      \includegraphics[width=8.75cm,type=eps,ext=.eps,read=.eps]{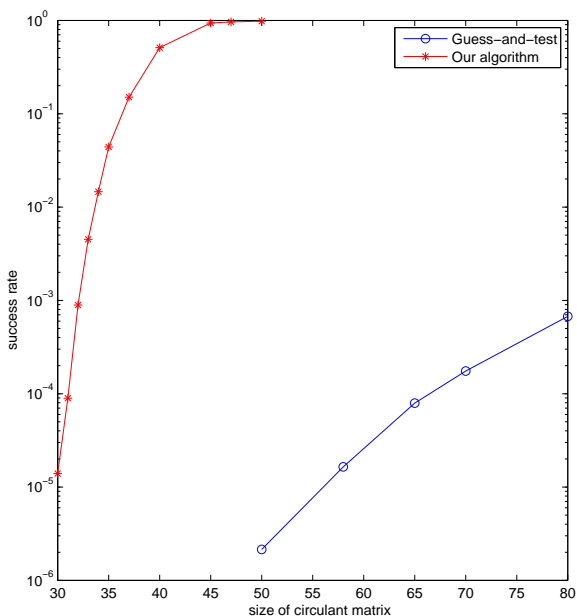}
      \caption{Comparison of the success rate of guess-and-test and
        hill climbing in finding a weight-I girth-8 regular QC LDPC
        code when base matrix dimensions are $3 \times 9$.}
      \label{CompareJ3L9}
\end{figure}

%%%%%%%%%%%%%%%%%%%%%%%%%
\subsection{Effectiveness of the squashing procedure}
\label{sec.effectSquashing}

We conclude our discussion of numerical examples by discussing the
computational motivations for the squashing procedure.  Recall that in
Sec.~\ref{sec.squashing} we raised the following question.  Why do we
not simply try directly to find suitable parameters for a weight-I QC
LDPC code, rather than constructing an HQC code and using the
squashing procedure? We now show that it is much harder to find a
suitable code using the direct method.

To show this we present results on the following experiment.  First we
construct several protographs with structures similar to
Fig.~\ref{fig.onesided}(a) with the number of check nodes ranging from
three to nine. We set $p_{[2]} = 4$ which means that the number of
rows in the corresponding base matrices ranges from 12 to 36. For each
protograph, we construct girth-10 QC LDPC codes with $p_{[1]} = 100$
using the girth maximization algorithm and the squashing procedure.
We also try to construct girth-10 weight-I QC LDPC codes with base
matrices having the same size and same non-zero positions as those
obtained from the squashing procedure using the direct method.  The
same hill climbing algorithm is applied to this design problem as is
used in conjunction with the HQC LDPC approach.  We record the time of
designing ten codes for each configuration. Figure~\ref{timing}
depicts the average time required to construct one girth-10 QC LDPC
code using each of these two schemes.

\begin{figure}[htbp]
      \centering
      \includegraphics[width=8.75cm,type=eps,ext=.eps,read=.eps]{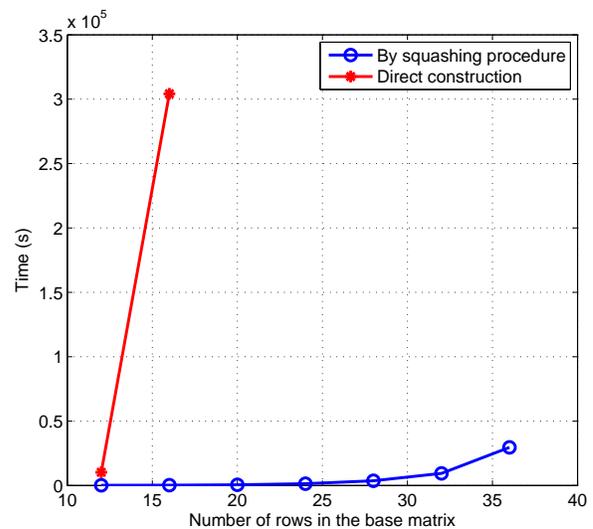}
      \caption{Average time of constructing one girth-10 QC LDPC code
        with the direct method and the squashing procedure.}
      \label{timing}
\end{figure}

For both schemes, the time required to find a girth-10 code increases
with the number of rows in the base matrix.  When the squashing
procedure is used, we can find a suitable base matrix in reasonable
time even for large base matrices (large number of rows).  In
contrast, when using the direct method, we have to spend an extremely
long time searching even for a small base matrices. From this
comparison, we conclude that the squashing method is quite a bit more
efficient.
\section{Conclusion}
\label{sec.conclusion}

In this paper we present a methodology for designing high-girth QC
LDPC codes that match a given protograph structure.  In developing our
methodology, we introduce a new class of {\em hierarchical} QC LDPC
codes and explain how to determine the girth of such codes. The
hierarchical QC LDPC codes can be represented using parity check
matrices over multi-variate polynomials, or in terms of a tree
structure. We show that higher-weight versions of hierarchical codes
suffer from inevitable cycles in analogous ways to non-hierarchical QC
LDPC codes, but that a straightforward {\em squashing} procedure can
remove these cycles.  We introduce a hill-climbing procedure to
eliminate the non-inevitable cycles from the code, and subsequently
remove the inevitable cycles by squashing.  Thus the main use of the
hierarchical codes in this paper is to reduce the number of free
parameters in the codes in an effort to make the girth maximization
procedure computationally tractable and fast, while knowing that the
inevitable cycles can be removed by squashing.  In our numerical
results we illustrate the computational advantage of the hill-climbing
and squashing procedures in comparison with other standard approaches.

We demonstrate our concepts and design procedure for the case of
one-sided spatially-coupled QC LDPC codes.  We present designs for two
such codes, of different rates and block lengths, both of girth-10.
We compare their performance to girth-6 codes and observe a
significant decrease in the error floor.  We note that the second
code, whose variable nodes are of weight four, does not demonstrate
any error floor tendencies down to a WER of about $10^{-7}$, i.e., the
slope of the WER as a function of SNR is still
steepening. Computational effort limited us from simulating lower
WERs.  But we note that the Gallager-B algorithm we chose to simulate
displays much higher error floors than the standard sum-product or
min-sum algorithms.  (In fact, this is why we choose to simulate this
algorithm.)  Given that the class of one-sided spatially coupled codes
has already been theoretically shown to have excellent waterfall
performance, we believe the evidence presented strongly indicates that
the techniques introduced herein can produce practical codes with very
good performance in both the waterfall and error floor regimes.

\subsection{Girth maximizing algorithms}
\label{sec.girthMaxAlgs}

In this appendix we present our girth maximizing algorithms.  As
discussed in the text the objective of these algorithms is to remove
all {\em non-inevitable} cycles from the quasi-cyclic codes.  We first
present our algorithm for weight-1 QC LDPC codes, and then for general
heavier-weight or HQC LDPC codes.  We do this for simplicity of
explanation as the latter algorithm is a generalization of the former.

\vskip 0.1 in

\noindent{\bf Algorithm 1:} Weight-I QC LDPC code construction

{\bf (i) Set-up and code initialization:} Specify the desired girth
$\girth$, matrix dimension $p$, and $\setTree$.

For each pair $(j,l)$ such that ${\bf T}_{j,l} \neq \nullTree$, pick a
value $z$ independently and uniformly from $\{0,\ldots,p-1\}$.
Initialize the code with $\sym{z}{j}{l} = 1$ (and $\sym{z'}{j}{l} = 0$
for all $z' \neq z$).

{\bf (ii) Calculate cost vector of current code:} Use Subroutine~1,
described in Appendix~\ref{app.costQCLDPC}, to calculate the cost
vectors of the current code, i.e., $\Gamma = \{\costB{j,l}\}$.  Then
for each element of $\Gamma$ we calculate the change in edge label
that most reduces cost, and the resulting cost, respectively:
\begin{align*}
\tilde{z}_{j,l} = \mathop{\arg \min}_{z:\, 0 \leq z \leq p-1}
\costB{j,l}(z),\\ \costTil{j,l} = \min_{z:\, 0 \leq z \leq p-1}
\costB{j,l}(z).
\end{align*}
Recalling that $\indSym[j,l]$ is the value of the of the coefficient
of the current code $\sym{\indSym}{j}{l}$, let
\begin{eqnarray}
\costBneg{j,l} = \costB{j,l}(\indSym[j,l]) \nonumber
\end{eqnarray}
be the cost of the coefficient if it remains unchanged.  

{\bf(iii) Identify best coefficient to change:} Identify the coefficient
to change that would most greatly reduces the cost, i.e.,
\begin{equation*}
(j_{\max}, l_{\max}) = \mathop{\arg \max}_{(j,l):\, 1 \leq j \leq J,
    \,1 \leq l \leq L, \, \conMat_{j,l} \neq \nullTree}
  \costBneg{j,l}-\costTil{j,l},
\end{equation*}
where we break ties randomly.  There are two possible outcomes.

\begin{itemize}
\item[(a)] If
  $\costBneg{j_{\max},l_{\max}}-\costTil{j_{\max},l_{\max}} > 0$,
  we update the code by setting
\begin{equation*}
\sym{\tilde{z}_{j_{max},l_{max}}}{j_{\max}}{l_{\max}} = 1,
\end{equation*}
and
\begin{equation*}
\sym{\indSym}{j_{\max}}{l_{\max}} = 0.
\end{equation*}
We iterate by now returning to Step (ii).

\item[(b)] If
  $\costBneg{j_{\max},l_{\max}}-\costTil{j_{\max},l_{\max}} = 0$,
  the algorithm terminates.\\
\end{itemize}

{\bf (iii) Terminate algorithm:} There are two possible termination
conditions.
\begin{itemize}
\item[(a)] If $\costB{j,l}(\indSym[j,l])=0$ for all $(j,l)$ such that
  $\conMat_{j,l} \neq \nullTree$, then we have found a code that
  satisfies the desired parameters.
\item[(b)] Else if there is a $(j,l)$ such that
  $\costB{j,l}(\indSym[j,l])\neq 0$ the algorithm has converged to a
  local minimum.
\end{itemize}

We now present the generalized algorithm for heavier-weight QC LDPC
and HQC LDPC codes.  In contrast to the first algorithm, the trees
$\conMat_{j,l} \in \setTree$ that define these codes have more than
one edge.  Therefore, for each edge of each tree we define a cost
vector.  We index the cost vectors both by their level in the tree and
by their position within each level, as well as by $j$ and $l$, thus
\begin{equation*}
\costBHier{j,l,i,k}=[\cost{0}, \cost{1}, \cdots,
  \cost{\p{k}-1}]
\end{equation*}
for $1 \leq i \leq |\conMat_{j,l}[k]|$ and $1 \leq k \leq K$ where we
recall that $|\conMat_{j,l}[k]|$ is the number of edges at level $k$
in $\conMat_{j,l}$.

\vskip 0.1 in
\noindent{\bf Algorithm 2:} Hierarchical QC LDPC code construction

{\bf (i) Set-up and code initialization:} Specify the desired girth
$\girth$, matrix dimension $p$, and $\setTree$.

For each pair $(j,l)$ such that ${\bf T}_{j,l} \neq \nullTree$,
randomly initialize the values for each edge label (while obeying the
requirement that sibling edges must have distinct labels).  Probably
the most straightforward way to do this is to work down the tree from
level $K$ to the first level, picking the edge labels for each set of
sibling edges at level $k$ uniformly without replacement from $\{0,
\ldots, \p{k}-1\}$.  Given the initial edge labels, compute all
non-zero code coefficients, i.e., those associated with each leaf.

{\bf (ii) Calculate cost vector of current code:} Use Subroutine~2,
described in Appendix~\ref{app.costHQCLDPC}, to calculate the cost
vectors of the current code, i.e., $\Gamma = \{\costBHier{j,l,i,k}\}$.
Then for each element of $\Gamma$ we calculate the change in edge
label that most reduces cost, and the resulting cost, respectively:
\begin{align*}
\tilde{z}_{j,l,i,k} & = \mathop{\arg \min}_{z: 0 \leq z \leq \p{k}-1}
\costB{j,l,i,k}(z), \\
\costTil{j,l,i,k} &= \min_{z: 0 \leq z \leq \p{k}-1}
\costB{j,l,i,k}(z).
\end{align*}

Recalling that $\ind{k}[j,l]$ is the value of the of the $k$th
coordinate of the current code coefficient $\sym{\indSymBf}{j}{l}$, let
\begin{eqnarray}
\costBneg{j,l,i,k} = \costB{j,l,i,k}(\ind{k}[j,l]) \nonumber
\end{eqnarray}
be the cost if the coefficient value at the $k$th level remains
unchanged.

{\bf(iii) Identify best edge label to change:} Identify the edge label
to change that would most greatly reduces the cost, i.e.,
\begin{equation*}
(j_{\max}, l_{\max},  i_{\max}, k_{\max}) = \hspace{-3em} \mathop{\arg
    \max}_{\begin{footnotesize} \begin{array}{ll} (j,l,i,k) : & \hspace{-1em}
        1 \leq j \leq \Jrow{K},  1 \leq l \leq \Lcol{K}\\ 
        & \hspace{-1em} \conMat_{j,l} \neq \nullTree,  1 \leq i \leq
        |\conMat_{j,l}[k]|\\ & \hspace{-1em} 1 \leq k \leq K
      \end{array} \end{footnotesize}} 
 \hspace{-3em} \costBneg{j,l,i,k}-\costTil{j,l,i,k},
\end{equation*}
where we break ties randomly.  There are two possible outcomes.

\begin{itemize}
\item[(a)] If $\costBneg{j_{\max},l_{\max}, i_{\max},
  k_{\max}}-\costTil{j_{\max},l_{\max}, i_{\max}, k_{\max}} > 0$, we
  update the code by setting the value of the $i_{\max}$th edge at the
  $k_{\max}$th level of $\conMat_{j_{\max}, l_{\max}}$ equal to
  $\tilde{z}_{j_{\max},l_{\max},i_{\max},k_{\max}}$.

We iterate by now returning to Step (ii).

\item[(b)] If $\costBneg{j_{\max},l_{\max}, i_{\max},
  k_{\max}}-\costTil{j_{\max},l_{\max}, i_{\max}, k_{\max}} = 0$, the algorithm
  terminates.\\
\end{itemize}

{\bf (iv) Terminate algorithm:} There are two possible termination
conditions.
\begin{itemize}
\item[(a)] If for all $(j,l,i,k)$ we have $\costB{j,l,i,k}(z)=0$ when
  $z$ is set to equal the current label of the $i$th edge at level $k$
  in tree $\conMat_{j,l}$, then we have found a code that satisfies
  the desired parameters.
\item[(b)] Else there is a $(j,l,i,k)$ such that
  $\costB{j,l,i,k}(z)\neq 0$ and the algorithm has converged to a
  local minimum.
\end{itemize}

\subsection{The multiplicity of a path element}
\label{app.maxGirth}

Recall from the discussion of Section~\ref{sec.hillNonHier} that the
determination of guilty values becomes complicated when there are
repeated elements in a path.  To aid in dealing with these repeated
elements, in this appendix, we define the ``multiplicity'' of each
path element. This definition is needed for for the cost vector
calculation subroutines of both QC and HQC LDPC codes, described in
Appendices~\ref{app.costQCLDPC} and~\ref{app.costHQCLDPC},
respectively.

\begin{defn} \label{def.multiplicity}
Given a path $\path = \{\ordSet, \pathCoeff\}$, any coefficient in
$\pathCoeff$ is said to be {\em repeated} $r$ times if there are $r$
elements of $\path$, indexed by $i_1, \ldots, i_{r}$, for which
$(j_{i_1}, l_{i_1}) = (j_{i_2}, l_{i_2}) = \ldots = (j_{i_r},
l_{i_r})$ and for which $\bfInd{j_{i_1}}{l_{i_1}} = \ldots =
\bfInd{j_{i_r}}{l_{i_r}}$.  The {\em multiplicity} $\kappa$ of the
  element is computed as
\begin{equation}
\kappa = \sum_{t = 1}^r (-1)^{i_t}. \label{def.kappa}
\end{equation}
For path elements where $|\kappa| > 1$, $i_1$ is termed the {\em first
  occurrence} of the element. \hfill \QED
\end{defn}
The multiplicity can be a positive integer, a negative integer, or
zero.  When a path element has multiplicity zero the value of the
coefficient has no effect on whether (that particular) path
corresponds to a cycle.

\subsection{Cost calculation subroutine for weight-I QC LDPC codes}
\label{app.costQCLDPC}

In this appendix we present the subroutine for the calculation of the
cost vectors of a weight-I QC LDPC code.  In other words, given a set
of labeled trees we calculate the matrix specified
in~(\ref{eq.costMatrix}).

\noindent{\bf Subroutine 1:} 

The subroutine takes as inputs the current tree structure $\setTree$
(i.e., the set of labeled trees or, equivalently, the current parity
check matrix $\parChkMat$), the desired girth $g$, and a vector of
costs $\weight$.

{\bf (i) Define helper variables:} Define $x_{j,l,z}^{(\len)}$ to be
the number of cycles of length-$2\len$ that would result if edge label
$\indSym[j,l]$ were set to equal value $z$.  In other words, the code
was modified to be one in which $\sym{z}{j}{l} = 1$ and
$\sym{z'}{j}{l} = 0$ for all $z' \neq z$.  Initialize all
$x_{j,l,z}^{(\len)} = 0$.

{\bf (ii) Iterate through path lengths, paths, and path elements:}
Consider in turn: (a) each path length $\len$ where $2 \leq \len \leq
\girth/2-1$ and $\girth$ is the desired girth; (b) each path of length
$\len$, $\path \in \path_{\len}$ where $\path = \{\ordSet,
\pathCoeff\}$ and $|\ordSet| = |\pathCoeff| = 2 \Lambda$; and (c) the
first occurrence of each path element (indexed by $t, 1 \leq t \leq 2
\len$) in $\path$ that has non-zero multiplicity.  

{\bf (iii) Calculate guilty values and adjust helper variables:} Let
$s[j_\tau,l_\tau]$ be the first occurrence of a path element of
multiplicity $\kappa \neq 0$.  We want to compute the set of possible
values for $s[j_\tau, l_\tau]$ that would satisfy the condition for
the existence of a cycle.  Recall from~(\ref{GirthCon}) that a cycle
exists for the current path values if
\begin{equation}
\sum_{t=1}^{2 \Lambda} (-1)^{t} s[j_t,l_t] \;
\textmd{mod} \; p = 0. \label{eq.doesCycleExist}
\end{equation}
To check if a value $\beta \in \{0, \ldots, p-1\}$ to which $s[j_\tau,
  l_\tau]$ could be changed would satisfy~(\ref{eq.doesCycleExist}),
we subtract the contribution of the current value of $s[j_\tau,
  l_\tau]$, add in the contribution of the candidate value $\beta$,
and see if the result is equal to zero.  That is, we check whether or
not the relation
\begin{equation*}
\left\{\sum_{t=1}^{2 \Lambda} (-1)^{t} s[j_t,l_t] - \kappa s[j_\tau,
  l_\tau] + \kappa \beta \right\}\; \textmd{mod} \; p = 0
\end{equation*}
holds.  Equivalently, we ask is
\begin{equation}
\kappa \beta \equiv \kappa s[j_\tau, l_\tau] - \sum_{t=1}^{2 \Lambda}
(-1)^{t} s[j_t,l_t], \label{eq.condOnBeta}
\end{equation}
where the congruence is modulo-$p$?

For each value of $\beta$, $0 \leq \beta \leq p - 1$
satisfying~(\ref{eq.condOnBeta}) we increment
$x_{j_\tau,l_\tau,\beta}^{(\len)}$ as
\begin{equation*}
x_{j_\tau,l_\tau,\beta}^{(\len)} =
x_{j_\tau,l_\tau,\beta}^{(\len)}+1.
\end{equation*}

{\em Remarks:} By only computing the $\beta$ for the first occurrence
of each path element, we avoid double-counting the contribution to
cycles of elements with $|\kappa| > 1$.  Allowing $\kappa$, defined
in~(\ref{def.kappa}), to take on either positive or negative values
lets the multiplicity of the element indicate its ``aggregate
polarity'', i.e., whether it enters the sum~(\ref{eq.doesCycleExist})
as a positive or a negative contribution.  Since the calculations of
$\beta$ in~(\ref{eq.condOnBeta}) are over a ring, multiple values of
$\beta$ can satisfy the condition.\footnote{If, however, you restrict
  $p$ to be prime, which we do not, then the calculations would be
  over a field and there would be a unique solution $\beta$.  We do
  not choose to do this due to the greater limitation on the possible
  resulting block lengths of the code.}  However, at most there are
  $|\kappa|$ such values of $\beta$.  This is because the set of
  satisfying values of $\beta$ forms a coset of $Z_p$ with respect to
  the subgroup $\{\beta \; \mbox{s.t.} \; \kappa \beta \equiv 0\}$,
  the cardinality of which is upper bounded by $\kappa$.  Finally, we
  note that if $|\kappa| = 1$, a $\beta$
  satisfying~(\ref{eq.condOnBeta}) exists and it is the unique such
  $\beta$.

{\bf (iv) Compute cost vectors:} After considering all paths lengths
$\len$, $2 \leq \len \leq \girth/2-1$, all $\path \in \path_{\len}$, and
all elements of each path $\path$, calculate the cost vectors
element-by-element as
\begin{eqnarray*}
\costB{j,l}(z) = \sum_{\len=2}^{\girth/2-1} x_{j, l, z}^{(\len)} \cdot
w_{\len}. 
 \end{eqnarray*}

\subsection{Cost calculation subroutine for HQC LDPC codes}
\label{app.costHQCLDPC}

We now present the subroutine used to calculate the cost vectors of a
general HQC LDPC code.

\vskip 0.1 in
\noindent{\bf Subroutine~2:} 

The subroutine takes as inputs the current tree structure $\setTree$
(i.e., set of labeled trees or, equivalently, the current parity check
matrix $\parChkMat$), the desired girth $g$, and a vector of costs
$\weight$.

{\bf (i) Define helper variables:} Define $x_{j,l,i,z}^{(\len)}[k]$ to
be the number of cycles of length-$2\len$ that would result if the
$i$th edge at level $k$ in $\conMat_{j,l}$ were set to equal value
$z$, $0 \leq z \leq \p{k}-1$.

{\em Remark:} Modification of a single edge has in a hierarchical code
will, in general, change a number of code coefficients.  In
particular, all coefficients associated with leaves that are
descendents of that edge will change in their $k$th coordinate.  These
coefficients will change from ones in which
\begin{align*}
\sym{\ind{1}, \ldots,
  \ind{k-1}, \ind{k}, \ind{k+1}, \ldots \ind{K}}{j}{l} & = 1 \hspace{1em} \mbox{and}\\
\sym{\ind{1}, \ldots, \ind{k-1}, z, \ind{k+1}, \ldots \ind{K}}{j}{l}
& = 0
\end{align*}
to ones in which
\begin{align*}
\sym{\ind{1}, \ldots, \ind{k-1}, z, \ind{k+1}, \ldots \ind{K}}{j}{l} &
= 1 \hspace{1em} \mbox{ and}\\ \sym{\ind{1}, \ldots, \ind{k-1},
  \ind{k}, \ind{k+1}, \ldots \ind{K}}{j}{l} & = 0.
\end{align*}
 Initialize all $x_{j,l,i,z}^{(\len)}[k] = 0$.

{\bf (ii) Set infinite costs:} For each $\conMat_{j,l} \neq
\nullTree$, each pair $(j,l)$, $1 \leq j \leq \Jrow{K}$, $1 \leq l
\leq \Lcol{K}$, each level $k$, $1 \leq k \leq K$, and each level-$k$
edge index $i$, $1 \leq i \leq |\conMat_{j,l}[k]|$, let ${\cal E}$ be
the set of labels of sibling edges.  For each $z \in {\cal E}$ set
\begin{equation*}
x_{j,l,i,z}^{(\len)}[k] =  \infty.
\end{equation*}

{\em Remark:} Recall from the algorithms described in
Section~\ref{sec.girthMax} that our approach to code optimization is
to identify the change in the single edge label that most reduces a
weighted sum of cycle counts.  In the special case of weight-I QC LDPC
codes there was a one-to-one mapping between code coefficients and
tree edges (since each tree has only a single edge).  In the
generalized setting we are now considering we seek to identify the
change in a single edge of one of the trees that will most reduce the
cost.  By setting certain costs to infinity, certain changes in code
structure will never be made. The changes thus barred are those that
would change the tree topology.  By setting those costs to infinity we
ensure that the unlabeled trees that describe our code remains an
invariant under our algorithm.

{\bf (iii) Iterate through path lengths, paths, and path elements:}
Consider in turn: (a) each path length $\len$ where $2 \leq \len \leq
\girth/2-1$; (b) each path of length $\len$, $\path \in \path_{\len}$
where $\path = \{\ordSet, \pathCoeff\}$ and $|\ordSet| = |\pathCoeff|
= 2 \Lambda$; (c) the first occurrence of each path element (indexed
by $t, 1 \leq t \leq 2 \len$) in $\path$ that has non-zero
multiplicity.

{\bf (iv) Determine whether a particular path element can have
  ``guilty'' vales:} Let $\bfInd{j_\tau}{l_\tau}$ be the first
occurrence of a path element of multiplicity $\kappa \neq 0$.  Recall
that $\ind{k}[j_\tau,l_\tau]$ corresponds to the label of an edge of
tree $\conMat_{j_\tau, l_\tau}$ at level $k$.  Now, for the
coefficient $\bfInd{j_\tau}{l_\tau}$ under consideration, iterate
through each level of the code.  For each level $k$, $1 \leq k \leq K$
compute
\begin{eqnarray*}
\alpha_k = \sum_{t=1}^{2 \Lambda} (-1)^t \ind{k}[j_t,l_t] \;
\textmd{mod} \; \p{k}.
\end{eqnarray*}
Unless $\alpha_{k} = 0$ {\em for all but one} value of $k$, there are
no guilty values.  If there are no guilty values, proceed to the next
path element.  If there is a single level $k'$ such that $\alpha_{k'}
\neq 0$ proceed to step (v).

{\em Remark:} The reason for the all-but-one condition is that we
change at most one edge label per iteration.  Therefore, unless
$\alpha_k = 0$ for all but one value of $k$ there is no single change
in an edge label that would result in a cycle in this iteration.

{\bf (v) Calculate guilty values and adjust helper variables:} Now
consider coordinate $k'$ of the path element $\bfInd{j_\tau}{l_\tau}$
whose multiplicity $\kappa \neq 0$.  The same logic as led
to~(\ref{eq.condOnBeta}) can again be used to identify the guilty
values, now at level $k$.  That is, compute the set of values of
$\beta$, $0 \leq \beta \leq \p{k}-1$ such that
\begin{equation}
\kappa \beta \equiv \kappa \ind{k}[j_\tau,l_\tau] - \sum_{t=1}^{2
  \Lambda} (-1)^{t} \ind{k}[j_t,l_t]. \label{eq.condOnBetaHier}
\end{equation}

For each value of $\beta$, $0 \leq \beta \leq \p{k} - 1$
satisfying~(\ref{eq.condOnBetaHier}) we increment
$x_{j_t,l_t,i,\beta}^{(\len)}[k]$ as
\begin{equation*}
x_{j_t,l_t,i,\beta}^{(\len)}[k] =
x_{j_t,l_t,i,\beta}^{(\len)}[k]+1.  
\end{equation*}
where $i$, $1 \leq i \leq |\conMat_{j_t, l_t}[k]|$, is the index of
the level-$k$ edge in $\conMat_{j_t, l_t}$ whose label is
$\ind{k}[j_\tau,l_\tau]$.

{\em Remarks:} One of the added complications of the generalized
algorithm is that there is not a one-to-one mapping between the code
parameters that we are adjusting (the tree edge values) and the code
coefficients (each of which is associated with one leaf of the tree).
When an edge value is adjusted there is a ripple effect, changing the
coefficients associated with all descendent leaves.  However, each
change in a edge label effects only one of the $K$
sums~(\ref{def.pathSum}), all of which Theorem~\ref{thm.girthCond}
requires to be equal to zero for a cycle to exist.  Thus, although
there is a ripple effect on the code coefficient when adjusting edge
labels, the values of the $\Sigma[k]$ at other levels is not effected.
Thus, considering the tree structure of the code nicely decouples the
question of girth and the search for high-girths from the algebraic
structure of the code.

\subsection{Proof of Lemma~\ref{lemma.squash}}
\label{app.squash}

To prove part (i) of the lemma consider the ordered set of
coefficients~(\ref{eq.sixCycleInevPath}) that describes the inevitable
cycle.  Note that the first and last coefficient must be in the same
row of the base matrix since the path defines a cycle.  The second and
third and the fourth and fifth coefficients must also each be in the
same row.  Since, when viewed at the first level of the code,
successive rows in a path must be distinct, three distinct rows are
traversed.  In Fig.~\ref{fig.illustrativePathThroughBaseMat_six} we
illustrate this logic for a matrix corresponding to the polynomial
$h_{j,l}(x,y) = x^{a_1}y^0 + x^{a_2}y^1 + x^{a_3}y^2$, $\ordSet = \{
(j,l), (j,l), (j,l), (j,l), (j,l), (j,l)\}$,and $\pathCoeff = \{ [a_1
  \; 0]^T, [a_2\;1]^T, [a_3\;2]^T, [a_1\;0]^T, [a_2\;1]^T,
[a_3\;2]^T\}$

\begin{figure}
\psfrag{&a}{\LARGE $a_1$}
\psfrag{&b}{\LARGE $a_2$}
\psfrag{&c}{\LARGE $a_3$}
\psfrag{&o}{\LARGE $\hspace{-0.5em}-1$}
\centering
\includegraphics[width=4.25cm]{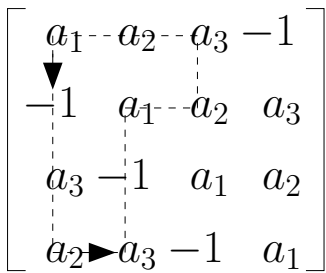}
\caption{Illustrative inevitable six-cycle that traverses three rows
  and three columns.}
\label{fig.illustrativePathThroughBaseMat_six}
\end{figure}

The logic of part (ii) is the same for rows and columns, hence we
provide the proof only for part the former.  Consider the ordered set
of coefficients of~(\ref{eq.eightCycleInevPath}).  We assert that
again the path must traverse at least three rows of the base matrix.
As before the first and last coefficients must be in the same row
since this path defines an inevitable cycle.  Each other sequential
pair of elements -- $([a_2\;A_2]^T, [b_1\;B_1]^T)$, $([b_2\;B_2]^T,
[a_2\;A_2]^T)$, and $([a_1\;A_1]^T, [b_2\;B_2]^T)$ -- must also lie in
the same rows. Consider the pair $([b_2\;B_2]^T, [a_2\;A_2]^T)$.  The
row this pair lies in can either be distinct from the starting row or
it can be the same.  If this row is distinct from the starting row
then, since successive rows are distinct, the row in which
$([a_2\;A_2]^T, [b_1\;B_1]^T)$ lies must be distinct both from this
row and from the starting row and the lemma is proved for this case.
On the other hand, say $([b_2\;B_2]^T, [a_2\;A_2]^T)$ lies in the
starting row.  We assert that in this case $([a_2\;A_2]^T,
[b_1\;B_1]^T)$ and $([a_1\;A_1]^T, [b_2\;B_2]^T)$ must lie in distinct
rows and so the total number of rows again is at least three. To see
this last assertion note first that the first $[a_1 \; A_1]^T$ and the
fifth coefficient $[a_2\;A_2]^T$ are, by assumption, in the same row.
Next observe that the second and sixth coefficients are $[a_2\;A_2]^T$
and $[a_1\;A_1]^T$, respectively, both in distinct rows from the
first.  As long as $\p{2} > 2$ these latter two coefficients (the
second and sixth) must be in distinct rows of the base matrix.  This
follows from the cyclic nature of the code.  The only way a pair of
coefficients could appear in two distinct rows and two distinct
columns in swapped order would be if $\p{2} = 2$, but we have assumed
that $\p{2} = 4$.  

The logic of the second case is illustrated in
Fig.~\ref{fig.illustrativePathThroughBaseMat_eight} for the pair of
polynomials $h_{j,l_1} = x^a + x^b y$ and $h_{j,l_2} = x^c + x^d y^3$.
The path illustrated corresponds to 
\begin{equation*}
\ordSet = \{(j,l_1), (j,l_1),
(j,l_2), (j, l_2), (j,l_1), (j,l_1), (j,l_2), (j, l_2)\}
\end{equation*}
 and
\begin{align*}
\pathCoeff = \left\{\tightArray{a_1}{0}, \tightArray{a_2}{1},
\tightArray{b_1}{0}, \tightArray{b_2}{3}, \tightArray{a_2}{1},
\tightArray{a_1}{0}, \tightArray{b_2}{3}, \tightArray{b_1}{0}\right\}.
\end{align*}

\begin{figure}
\psfrag{&a}{\LARGE $a_1$} \psfrag{&b}{\LARGE $a_2$} \psfrag{&c}{\LARGE
  $b_1$} \psfrag{&d}{\LARGE $b_2$} \psfrag{&o}{\LARGE $\hspace{-0.5em}
  -1$} \centering
\includegraphics[width=8.5cm]{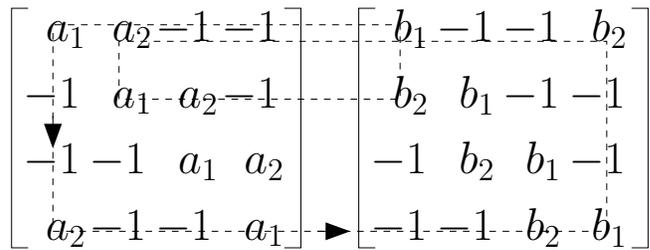}
\caption{Illustrative inevitable eight-cycle that traverses three
  rows.}
\label{fig.illustrativePathThroughBaseMat_eight}
\end{figure}

\subsection{Base matrices}
\label{sec.baseMatrixApp}

In this appendix, the base matrices of the two girth 10 QC LDPC codes
discussed in Sec.~\ref{sec.results} are specified below.  The base
matrix of the first code, $B_1$ is written in the transposed format
due to space.

\newcounter{mytempeqncnt}
\begin{figure*}[!t]
% ensure that we have normalsize text
\normalsize
% Store the current equation number.
\setcounter{mytempeqncnt}{\value{equation}}
% Set the equation number to one less than the one
% desired for the first equation here.
% The value here will have to changed if equations
% are added or removed prior to the place these
% equations are referenced in the main text.
\setcounter{equation}{5} \small $B_1^{T}$ = \\
\begin{equation}
\hspace{-0.5cm} \left[ \begin{smallmatrix}
 &-1 &34 &-1 &-1 &-1 & 1 &-1 &-1 &-1 &-1 &27 &-1 &-1 &-1 &-1 &-1 &-1 &-1 &-1 &-1 &-1 &-1 &-1 &-1 &-1 &-1 &-1 &-1 &-1 &-1 &-1 &-1 &-1 &-1 &-1 &-1 &-1 &-1 &-1 &-1 &-1 &-1 &-1 &-1\\
 &-1 &-1 &34 &-1 &-1 &-1 & 1 &-1 &-1 &-1 &-1 &27 &-1 &-1 &-1 &-1 &-1 &-1 &-1 &-1 &-1 &-1 &-1 &-1 &-1 &-1 &-1 &-1 &-1 &-1 &-1 &-1 &-1 &-1 &-1 &-1 &-1 &-1 &-1 &-1 &-1 &-1 &-1 &-1\\
 &-1 &-1 &-1 &34 &-1 &-1 &-1 & 1 &27 &-1 &-1 &-1 &-1 &-1 &-1 &-1 &-1 &-1 &-1 &-1 &-1 &-1 &-1 &-1 &-1 &-1 &-1 &-1 &-1 &-1 &-1 &-1 &-1 &-1 &-1 &-1 &-1 &-1 &-1 &-1 &-1 &-1 &-1 &-1\\
 &34 &-1 &-1 &-1 & 1 &-1 &-1 &-1 &-1 &27 &-1 &-1 &-1 &-1 &-1 &-1 &-1 &-1 &-1 &-1 &-1 &-1 &-1 &-1 &-1 &-1 &-1 &-1 &-1 &-1 &-1 &-1 &-1 &-1 &-1 &-1 &-1 &-1 &-1 &-1 &-1 &-1 &-1 &-1\\
 &-1 &-1 &-1 & 7 &-1 &54 &-1 &-1 &21 &-1 &-1 &-1 &-1 &-1 &-1 &-1 &-1 &-1 &-1 &-1 &-1 &-1 &-1 &-1 &-1 &-1 &-1 &-1 &-1 &-1 &-1 &-1 &-1 &-1 &-1 &-1 &-1 &-1 &-1 &-1 &-1 &-1 &-1 &-1\\
 & 7 &-1 &-1 &-1 &-1 &-1 &54 &-1 &-1 &21 &-1 &-1 &-1 &-1 &-1 &-1 &-1 &-1 &-1 &-1 &-1 &-1 &-1 &-1 &-1 &-1 &-1 &-1 &-1 &-1 &-1 &-1 &-1 &-1 &-1 &-1 &-1 &-1 &-1 &-1 &-1 &-1 &-1 &-1\\
 &-1 & 7 &-1 &-1 &-1 &-1 &-1 &54 &-1 &-1 &21 &-1 &-1 &-1 &-1 &-1 &-1 &-1 &-1 &-1 &-1 &-1 &-1 &-1 &-1 &-1 &-1 &-1 &-1 &-1 &-1 &-1 &-1 &-1 &-1 &-1 &-1 &-1 &-1 &-1 &-1 &-1 &-1 &-1\\
 &-1 &-1 & 7 &-1 &54 &-1 &-1 &-1 &-1 &-1 &-1 &21 &-1 &-1 &-1 &-1 &-1 &-1 &-1 &-1 &-1 &-1 &-1 &-1 &-1 &-1 &-1 &-1 &-1 &-1 &-1 &-1 &-1 &-1 &-1 &-1 &-1 &-1 &-1 &-1 &-1 &-1 &-1 &-1\\
 &-1 &-1 &-1 &-1 &-1 &-1 &-1 &78 &-1 &-1 &89 &-1 &-1 &79 &-1 &-1 &-1 &-1 &-1 &-1 &-1 &-1 &-1 &-1 &-1 &-1 &-1 &-1 &-1 &-1 &-1 &-1 &-1 &-1 &-1 &-1 &-1 &-1 &-1 &-1 &-1 &-1 &-1 &-1\\
 &-1 &-1 &-1 &-1 &78 &-1 &-1 &-1 &-1 &-1 &-1 &89 &-1 &-1 &79 &-1 &-1 &-1 &-1 &-1 &-1 &-1 &-1 &-1 &-1 &-1 &-1 &-1 &-1 &-1 &-1 &-1 &-1 &-1 &-1 &-1 &-1 &-1 &-1 &-1 &-1 &-1 &-1 &-1\\
 &-1 &-1 &-1 &-1 &-1 &78 &-1 &-1 &89 &-1 &-1 &-1 &-1 &-1 &-1 &79 &-1 &-1 &-1 &-1 &-1 &-1 &-1 &-1 &-1 &-1 &-1 &-1 &-1 &-1 &-1 &-1 &-1 &-1 &-1 &-1 &-1 &-1 &-1 &-1 &-1 &-1 &-1 &-1\\
 &-1 &-1 &-1 &-1 &-1 &-1 &78 &-1 &-1 &89 &-1 &-1 &79 &-1 &-1 &-1 &-1 &-1 &-1 &-1 &-1 &-1 &-1 &-1 &-1 &-1 &-1 &-1 &-1 &-1 &-1 &-1 &-1 &-1 &-1 &-1 &-1 &-1 &-1 &-1 &-1 &-1 &-1 &-1\\
 &-1 &-1 &-1 &-1 &-1 &-1 &42 &-1 &-1 &-1 &79 &-1 &-1 &75 &-1 &-1 &-1 &-1 &-1 &-1 &-1 &-1 &-1 &-1 &-1 &-1 &-1 &-1 &-1 &-1 &-1 &-1 &-1 &-1 &-1 &-1 &-1 &-1 &-1 &-1 &-1 &-1 &-1 &-1\\
 &-1 &-1 &-1 &-1 &-1 &-1 &-1 &42 &-1 &-1 &-1 &79 &-1 &-1 &75 &-1 &-1 &-1 &-1 &-1 &-1 &-1 &-1 &-1 &-1 &-1 &-1 &-1 &-1 &-1 &-1 &-1 &-1 &-1 &-1 &-1 &-1 &-1 &-1 &-1 &-1 &-1 &-1 &-1\\
 &-1 &-1 &-1 &-1 &42 &-1 &-1 &-1 &79 &-1 &-1 &-1 &-1 &-1 &-1 &75 &-1 &-1 &-1 &-1 &-1 &-1 &-1 &-1 &-1 &-1 &-1 &-1 &-1 &-1 &-1 &-1 &-1 &-1 &-1 &-1 &-1 &-1 &-1 &-1 &-1 &-1 &-1 &-1\\
 &-1 &-1 &-1 &-1 &-1 &42 &-1 &-1 &-1 &79 &-1 &-1 &75 &-1 &-1 &-1 &-1 &-1 &-1 &-1 &-1 &-1 &-1 &-1 &-1 &-1 &-1 &-1 &-1 &-1 &-1 &-1 &-1 &-1 &-1 &-1 &-1 &-1 &-1 &-1 &-1 &-1 &-1 &-1\\
 &-1 &-1 &-1 &-1 &-1 &-1 &-1 &-1 &-1 &-1 &38 &-1 &93 &-1 &-1 &-1 &-1 &-1 &92 &-1 &-1 &-1 &-1 &-1 &-1 &-1 &-1 &-1 &-1 &-1 &-1 &-1 &-1 &-1 &-1 &-1 &-1 &-1 &-1 &-1 &-1 &-1 &-1 &-1\\
 &-1 &-1 &-1 &-1 &-1 &-1 &-1 &-1 &-1 &-1 &-1 &38 &-1 &93 &-1 &-1 &-1 &-1 &-1 &92 &-1 &-1 &-1 &-1 &-1 &-1 &-1 &-1 &-1 &-1 &-1 &-1 &-1 &-1 &-1 &-1 &-1 &-1 &-1 &-1 &-1 &-1 &-1 &-1\\
 &-1 &-1 &-1 &-1 &-1 &-1 &-1 &-1 &38 &-1 &-1 &-1 &-1 &-1 &93 &-1 &92 &-1 &-1 &-1 &-1 &-1 &-1 &-1 &-1 &-1 &-1 &-1 &-1 &-1 &-1 &-1 &-1 &-1 &-1 &-1 &-1 &-1 &-1 &-1 &-1 &-1 &-1 &-1\\
 &-1 &-1 &-1 &-1 &-1 &-1 &-1 &-1 &-1 &38 &-1 &-1 &-1 &-1 &-1 &93 &-1 &92 &-1 &-1 &-1 &-1 &-1 &-1 &-1 &-1 &-1 &-1 &-1 &-1 &-1 &-1 &-1 &-1 &-1 &-1 &-1 &-1 &-1 &-1 &-1 &-1 &-1 &-1\\
 &-1 &-1 &-1 &-1 &-1 &-1 &-1 &-1 &-1 &40 &-1 &-1 &-1 &-1 &76 &-1 &15 &-1 &-1 &-1 &-1 &-1 &-1 &-1 &-1 &-1 &-1 &-1 &-1 &-1 &-1 &-1 &-1 &-1 &-1 &-1 &-1 &-1 &-1 &-1 &-1 &-1 &-1 &-1\\
 &-1 &-1 &-1 &-1 &-1 &-1 &-1 &-1 &-1 &-1 &40 &-1 &-1 &-1 &-1 &76 &-1 &15 &-1 &-1 &-1 &-1 &-1 &-1 &-1 &-1 &-1 &-1 &-1 &-1 &-1 &-1 &-1 &-1 &-1 &-1 &-1 &-1 &-1 &-1 &-1 &-1 &-1 &-1\\
 &-1 &-1 &-1 &-1 &-1 &-1 &-1 &-1 &-1 &-1 &-1 &40 &76 &-1 &-1 &-1 &-1 &-1 &15 &-1 &-1 &-1 &-1 &-1 &-1 &-1 &-1 &-1 &-1 &-1 &-1 &-1 &-1 &-1 &-1 &-1 &-1 &-1 &-1 &-1 &-1 &-1 &-1 &-1\\
 &-1 &-1 &-1 &-1 &-1 &-1 &-1 &-1 &40 &-1 &-1 &-1 &-1 &76 &-1 &-1 &-1 &-1 &-1 &15 &-1 &-1 &-1 &-1 &-1 &-1 &-1 &-1 &-1 &-1 &-1 &-1 &-1 &-1 &-1 &-1 &-1 &-1 &-1 &-1 &-1 &-1 &-1 &-1\\
 &-1 &-1 &-1 &-1 &-1 &-1 &-1 &-1 &-1 &-1 &-1 &-1 &-1 &25 &-1 &-1 &-1 &-1 &81 &-1 &-1 &-1 &-1 &60 &-1 &-1 &-1 &-1 &-1 &-1 &-1 &-1 &-1 &-1 &-1 &-1 &-1 &-1 &-1 &-1 &-1 &-1 &-1 &-1\\
 &-1 &-1 &-1 &-1 &-1 &-1 &-1 &-1 &-1 &-1 &-1 &-1 &-1 &-1 &25 &-1 &-1 &-1 &-1 &81 &60 &-1 &-1 &-1 &-1 &-1 &-1 &-1 &-1 &-1 &-1 &-1 &-1 &-1 &-1 &-1 &-1 &-1 &-1 &-1 &-1 &-1 &-1 &-1\\
 &-1 &-1 &-1 &-1 &-1 &-1 &-1 &-1 &-1 &-1 &-1 &-1 &-1 &-1 &-1 &25 &81 &-1 &-1 &-1 &-1 &60 &-1 &-1 &-1 &-1 &-1 &-1 &-1 &-1 &-1 &-1 &-1 &-1 &-1 &-1 &-1 &-1 &-1 &-1 &-1 &-1 &-1 &-1\\
 &-1 &-1 &-1 &-1 &-1 &-1 &-1 &-1 &-1 &-1 &-1 &-1 &25 &-1 &-1 &-1 &-1 &81 &-1 &-1 &-1 &-1 &60 &-1 &-1 &-1 &-1 &-1 &-1 &-1 &-1 &-1 &-1 &-1 &-1 &-1 &-1 &-1 &-1 &-1 &-1 &-1 &-1 &-1\\
 &-1 &-1 &-1 &-1 &-1 &-1 &-1 &-1 &-1 &-1 &-1 &-1 &-1 &-1 &63 &-1 &-1 &-1 &91 &-1 &-1 &-1 & 0 &-1 &-1 &-1 &-1 &-1 &-1 &-1 &-1 &-1 &-1 &-1 &-1 &-1 &-1 &-1 &-1 &-1 &-1 &-1 &-1 &-1\\
 &-1 &-1 &-1 &-1 &-1 &-1 &-1 &-1 &-1 &-1 &-1 &-1 &-1 &-1 &-1 &63 &-1 &-1 &-1 &91 &-1 &-1 &-1 & 0 &-1 &-1 &-1 &-1 &-1 &-1 &-1 &-1 &-1 &-1 &-1 &-1 &-1 &-1 &-1 &-1 &-1 &-1 &-1 &-1\\
 &-1 &-1 &-1 &-1 &-1 &-1 &-1 &-1 &-1 &-1 &-1 &-1 &63 &-1 &-1 &-1 &91 &-1 &-1 &-1 & 0 &-1 &-1 &-1 &-1 &-1 &-1 &-1 &-1 &-1 &-1 &-1 &-1 &-1 &-1 &-1 &-1 &-1 &-1 &-1 &-1 &-1 &-1 &-1\\
 &-1 &-1 &-1 &-1 &-1 &-1 &-1 &-1 &-1 &-1 &-1 &-1 &-1 &63 &-1 &-1 &-1 &91 &-1 &-1 &-1 & 0 &-1 &-1 &-1 &-1 &-1 &-1 &-1 &-1 &-1 &-1 &-1 &-1 &-1 &-1 &-1 &-1 &-1 &-1 &-1 &-1 &-1 &-1\\
 &-1 &-1 &-1 &-1 &-1 &-1 &-1 &-1 &-1 &-1 &-1 &-1 &-1 &-1 &-1 &-1 &-1 &-1 &38 &-1 &90 &-1 &-1 &-1 &-1 &21 &-1 &-1 &-1 &-1 &-1 &-1 &-1 &-1 &-1 &-1 &-1 &-1 &-1 &-1 &-1 &-1 &-1 &-1\\
 &-1 &-1 &-1 &-1 &-1 &-1 &-1 &-1 &-1 &-1 &-1 &-1 &-1 &-1 &-1 &-1 &-1 &-1 &-1 &38 &-1 &90 &-1 &-1 &-1 &-1 &21 &-1 &-1 &-1 &-1 &-1 &-1 &-1 &-1 &-1 &-1 &-1 &-1 &-1 &-1 &-1 &-1 &-1\\
 &-1 &-1 &-1 &-1 &-1 &-1 &-1 &-1 &-1 &-1 &-1 &-1 &-1 &-1 &-1 &-1 &38 &-1 &-1 &-1 &-1 &-1 &90 &-1 &-1 &-1 &-1 &21 &-1 &-1 &-1 &-1 &-1 &-1 &-1 &-1 &-1 &-1 &-1 &-1 &-1 &-1 &-1 &-1\\
 &-1 &-1 &-1 &-1 &-1 &-1 &-1 &-1 &-1 &-1 &-1 &-1 &-1 &-1 &-1 &-1 &-1 &38 &-1 &-1 &-1 &-1 &-1 &90 &21 &-1 &-1 &-1 &-1 &-1 &-1 &-1 &-1 &-1 &-1 &-1 &-1 &-1 &-1 &-1 &-1 &-1 &-1 &-1\\
 &-1 &-1 &-1 &-1 &-1 &-1 &-1 &-1 &-1 &-1 &-1 &-1 &-1 &-1 &-1 &-1 &-1 &-1 &73 &-1 &-1 &-1 &-1 &96 &-1 &-1 &-1 &66 &-1 &-1 &-1 &-1 &-1 &-1 &-1 &-1 &-1 &-1 &-1 &-1 &-1 &-1 &-1 &-1\\
 &-1 &-1 &-1 &-1 &-1 &-1 &-1 &-1 &-1 &-1 &-1 &-1 &-1 &-1 &-1 &-1 &-1 &-1 &-1 &73 &96 &-1 &-1 &-1 &66 &-1 &-1 &-1 &-1 &-1 &-1 &-1 &-1 &-1 &-1 &-1 &-1 &-1 &-1 &-1 &-1 &-1 &-1 &-1\\
 &-1 &-1 &-1 &-1 &-1 &-1 &-1 &-1 &-1 &-1 &-1 &-1 &-1 &-1 &-1 &-1 &73 &-1 &-1 &-1 &-1 &96 &-1 &-1 &-1 &66 &-1 &-1 &-1 &-1 &-1 &-1 &-1 &-1 &-1 &-1 &-1 &-1 &-1 &-1 &-1 &-1 &-1 &-1\\
 &-1 &-1 &-1 &-1 &-1 &-1 &-1 &-1 &-1 &-1 &-1 &-1 &-1 &-1 &-1 &-1 &-1 &73 &-1 &-1 &-1 &-1 &96 &-1 &-1 &-1 &66 &-1 &-1 &-1 &-1 &-1 &-1 &-1 &-1 &-1 &-1 &-1 &-1 &-1 &-1 &-1 &-1 &-1\\
 &-1 &-1 &-1 &-1 &-1 &-1 &-1 &-1 &-1 &-1 &-1 &-1 &-1 &-1 &-1 &-1 &-1 &-1 &-1 &-1 &-1 &85 &-1 &-1 &-1 &66 &-1 &-1 &-1 &-1 & 2 &-1 &-1 &-1 &-1 &-1 &-1 &-1 &-1 &-1 &-1 &-1 &-1 &-1\\
 &-1 &-1 &-1 &-1 &-1 &-1 &-1 &-1 &-1 &-1 &-1 &-1 &-1 &-1 &-1 &-1 &-1 &-1 &-1 &-1 &-1 &-1 &85 &-1 &-1 &-1 &66 &-1 &-1 &-1 &-1 & 2 &-1 &-1 &-1 &-1 &-1 &-1 &-1 &-1 &-1 &-1 &-1 &-1\\
 &-1 &-1 &-1 &-1 &-1 &-1 &-1 &-1 &-1 &-1 &-1 &-1 &-1 &-1 &-1 &-1 &-1 &-1 &-1 &-1 &-1 &-1 &-1 &85 &-1 &-1 &-1 &66 & 2 &-1 &-1 &-1 &-1 &-1 &-1 &-1 &-1 &-1 &-1 &-1 &-1 &-1 &-1 &-1\\
 &-1 &-1 &-1 &-1 &-1 &-1 &-1 &-1 &-1 &-1 &-1 &-1 &-1 &-1 &-1 &-1 &-1 &-1 &-1 &-1 &85 &-1 &-1 &-1 &66 &-1 &-1 &-1 &-1 & 2 &-1 &-1 &-1 &-1 &-1 &-1 &-1 &-1 &-1 &-1 &-1 &-1 &-1 &-1\\
 &-1 &-1 &-1 &-1 &-1 &-1 &-1 &-1 &-1 &-1 &-1 &-1 &-1 &-1 &-1 &-1 &-1 &-1 &-1 &-1 &-1 &90 &-1 &-1 &-1 &-1 &-1 &57 & 2 &-1 &-1 &-1 &-1 &-1 &-1 &-1 &-1 &-1 &-1 &-1 &-1 &-1 &-1 &-1\\
 &-1 &-1 &-1 &-1 &-1 &-1 &-1 &-1 &-1 &-1 &-1 &-1 &-1 &-1 &-1 &-1 &-1 &-1 &-1 &-1 &-1 &-1 &90 &-1 &57 &-1 &-1 &-1 &-1 & 2 &-1 &-1 &-1 &-1 &-1 &-1 &-1 &-1 &-1 &-1 &-1 &-1 &-1 &-1\\
 &-1 &-1 &-1 &-1 &-1 &-1 &-1 &-1 &-1 &-1 &-1 &-1 &-1 &-1 &-1 &-1 &-1 &-1 &-1 &-1 &-1 &-1 &-1 &90 &-1 &57 &-1 &-1 &-1 &-1 & 2 &-1 &-1 &-1 &-1 &-1 &-1 &-1 &-1 &-1 &-1 &-1 &-1 &-1\\
 &-1 &-1 &-1 &-1 &-1 &-1 &-1 &-1 &-1 &-1 &-1 &-1 &-1 &-1 &-1 &-1 &-1 &-1 &-1 &-1 &90 &-1 &-1 &-1 &-1 &-1 &57 &-1 &-1 &-1 &-1 & 2 &-1 &-1 &-1 &-1 &-1 &-1 &-1 &-1 &-1 &-1 &-1 &-1\\
 &-1 &-1 &-1 &-1 &-1 &-1 &-1 &-1 &-1 &-1 &-1 &-1 &-1 &-1 &-1 &-1 &-1 &-1 &-1 &-1 &-1 &-1 &-1 &-1 &-1 &-1 &12 &-1 &83 &-1 &-1 &-1 &73 &-1 &-1 &-1 &-1 &-1 &-1 &-1 &-1 &-1 &-1 &-1\\
 &-1 &-1 &-1 &-1 &-1 &-1 &-1 &-1 &-1 &-1 &-1 &-1 &-1 &-1 &-1 &-1 &-1 &-1 &-1 &-1 &-1 &-1 &-1 &-1 &-1 &-1 &-1 &12 &-1 &83 &-1 &-1 &-1 &73 &-1 &-1 &-1 &-1 &-1 &-1 &-1 &-1 &-1 &-1\\
 &-1 &-1 &-1 &-1 &-1 &-1 &-1 &-1 &-1 &-1 &-1 &-1 &-1 &-1 &-1 &-1 &-1 &-1 &-1 &-1 &-1 &-1 &-1 &-1 &12 &-1 &-1 &-1 &-1 &-1 &83 &-1 &-1 &-1 &73 &-1 &-1 &-1 &-1 &-1 &-1 &-1 &-1 &-1\\
 &-1 &-1 &-1 &-1 &-1 &-1 &-1 &-1 &-1 &-1 &-1 &-1 &-1 &-1 &-1 &-1 &-1 &-1 &-1 &-1 &-1 &-1 &-1 &-1 &-1 &12 &-1 &-1 &-1 &-1 &-1 &83 &-1 &-1 &-1 &73 &-1 &-1 &-1 &-1 &-1 &-1 &-1 &-1\\
 &-1 &-1 &-1 &-1 &-1 &-1 &-1 &-1 &-1 &-1 &-1 &-1 &-1 &-1 &-1 &-1 &-1 &-1 &-1 &-1 &-1 &-1 &-1 &-1 &-1 &-1 &-1 &91 &-1 &17 &-1 &-1 &50 &-1 &-1 &-1 &-1 &-1 &-1 &-1 &-1 &-1 &-1 &-1\\
 &-1 &-1 &-1 &-1 &-1 &-1 &-1 &-1 &-1 &-1 &-1 &-1 &-1 &-1 &-1 &-1 &-1 &-1 &-1 &-1 &-1 &-1 &-1 &-1 &91 &-1 &-1 &-1 &-1 &-1 &17 &-1 &-1 &50 &-1 &-1 &-1 &-1 &-1 &-1 &-1 &-1 &-1 &-1\\
 &-1 &-1 &-1 &-1 &-1 &-1 &-1 &-1 &-1 &-1 &-1 &-1 &-1 &-1 &-1 &-1 &-1 &-1 &-1 &-1 &-1 &-1 &-1 &-1 &-1 &91 &-1 &-1 &-1 &-1 &-1 &17 &-1 &-1 &50 &-1 &-1 &-1 &-1 &-1 &-1 &-1 &-1 &-1\\
 &-1 &-1 &-1 &-1 &-1 &-1 &-1 &-1 &-1 &-1 &-1 &-1 &-1 &-1 &-1 &-1 &-1 &-1 &-1 &-1 &-1 &-1 &-1 &-1 &-1 &-1 &91 &-1 &17 &-1 &-1 &-1 &-1 &-1 &-1 &50 &-1 &-1 &-1 &-1 &-1 &-1 &-1 &-1\\
 &-1 &-1 &-1 &-1 &-1 &-1 &-1 &-1 &-1 &-1 &-1 &-1 &-1 &-1 &-1 &-1 &-1 &-1 &-1 &-1 &-1 &-1 &-1 &-1 &-1 &-1 &-1 &-1 &-1 &-1 &57 &-1 &34 &-1 &-1 &-1 &-1 &-1 &-1 &55 &-1 &-1 &-1 &-1\\
 &-1 &-1 &-1 &-1 &-1 &-1 &-1 &-1 &-1 &-1 &-1 &-1 &-1 &-1 &-1 &-1 &-1 &-1 &-1 &-1 &-1 &-1 &-1 &-1 &-1 &-1 &-1 &-1 &-1 &-1 &-1 &57 &-1 &34 &-1 &-1 &55 &-1 &-1 &-1 &-1 &-1 &-1 &-1\\
 &-1 &-1 &-1 &-1 &-1 &-1 &-1 &-1 &-1 &-1 &-1 &-1 &-1 &-1 &-1 &-1 &-1 &-1 &-1 &-1 &-1 &-1 &-1 &-1 &-1 &-1 &-1 &-1 &57 &-1 &-1 &-1 &-1 &-1 &34 &-1 &-1 &55 &-1 &-1 &-1 &-1 &-1 &-1\\
 &-1 &-1 &-1 &-1 &-1 &-1 &-1 &-1 &-1 &-1 &-1 &-1 &-1 &-1 &-1 &-1 &-1 &-1 &-1 &-1 &-1 &-1 &-1 &-1 &-1 &-1 &-1 &-1 &-1 &57 &-1 &-1 &-1 &-1 &-1 &34 &-1 &-1 &55 &-1 &-1 &-1 &-1 &-1\\
 &-1 &-1 &-1 &-1 &-1 &-1 &-1 &-1 &-1 &-1 &-1 &-1 &-1 &-1 &-1 &-1 &-1 &-1 &-1 &-1 &-1 &-1 &-1 &-1 &-1 &-1 &-1 &-1 &-1 &-1 &-1 &25 &-1 &-1 &-1 &92 &-1 &-1 &78 &-1 &-1 &-1 &-1 &-1\\
 &-1 &-1 &-1 &-1 &-1 &-1 &-1 &-1 &-1 &-1 &-1 &-1 &-1 &-1 &-1 &-1 &-1 &-1 &-1 &-1 &-1 &-1 &-1 &-1 &-1 &-1 &-1 &-1 &25 &-1 &-1 &-1 &92 &-1 &-1 &-1 &-1 &-1 &-1 &78 &-1 &-1 &-1 &-1\\
 &-1 &-1 &-1 &-1 &-1 &-1 &-1 &-1 &-1 &-1 &-1 &-1 &-1 &-1 &-1 &-1 &-1 &-1 &-1 &-1 &-1 &-1 &-1 &-1 &-1 &-1 &-1 &-1 &-1 &25 &-1 &-1 &-1 &92 &-1 &-1 &78 &-1 &-1 &-1 &-1 &-1 &-1 &-1\\
 &-1 &-1 &-1 &-1 &-1 &-1 &-1 &-1 &-1 &-1 &-1 &-1 &-1 &-1 &-1 &-1 &-1 &-1 &-1 &-1 &-1 &-1 &-1 &-1 &-1 &-1 &-1 &-1 &-1 &-1 &25 &-1 &-1 &-1 &92 &-1 &-1 &78 &-1 &-1 &-1 &-1 &-1 &-1\\
 &-1 &-1 &-1 &-1 &-1 &-1 &-1 &-1 &-1 &-1 &-1 &-1 &-1 &-1 &-1 &-1 &-1 &-1 &-1 &-1 &-1 &-1 &-1 &-1 &-1 &-1 &-1 &-1 &-1 &-1 &-1 &-1 &-1 &87 &-1 &-1 &-1 &-1 &-1 &62 &-1 &-1 &-1 &23\\
 &-1 &-1 &-1 &-1 &-1 &-1 &-1 &-1 &-1 &-1 &-1 &-1 &-1 &-1 &-1 &-1 &-1 &-1 &-1 &-1 &-1 &-1 &-1 &-1 &-1 &-1 &-1 &-1 &-1 &-1 &-1 &-1 &-1 &-1 &87 &-1 &62 &-1 &-1 &-1 &19 &-1 &-1 &-1\\
 &-1 &-1 &-1 &-1 &-1 &-1 &-1 &-1 &-1 &-1 &-1 &-1 &-1 &-1 &-1 &-1 &-1 &-1 &-1 &-1 &-1 &-1 &-1 &-1 &-1 &-1 &-1 &-1 &-1 &-1 &-1 &-1 &-1 &-1 &-1 &87 &-1 &62 &-1 &-1 &-1 &19 &-1 &-1\\
 &-1 &-1 &-1 &-1 &-1 &-1 &-1 &-1 &-1 &-1 &-1 &-1 &-1 &-1 &-1 &-1 &-1 &-1 &-1 &-1 &-1 &-1 &-1 &-1 &-1 &-1 &-1 &-1 &-1 &-1 &-1 &-1 &87 &-1 &-1 &-1 &-1 &-1 &62 &-1 &-1 &-1 &23 &-1\\
 &-1 &-1 &-1 &-1 &-1 &-1 &-1 &-1 &-1 &-1 &-1 &-1 &-1 &-1 &-1 &-1 &-1 &-1 &-1 &-1 &-1 &-1 &-1 &-1 &-1 &-1 &-1 &-1 &-1 &-1 &-1 &-1 &-1 &-1 &33 &-1 &50 &-1 &-1 &-1 &-1 &-1 &-1 &71\\
 &-1 &-1 &-1 &-1 &-1 &-1 &-1 &-1 &-1 &-1 &-1 &-1 &-1 &-1 &-1 &-1 &-1 &-1 &-1 &-1 &-1 &-1 &-1 &-1 &-1 &-1 &-1 &-1 &-1 &-1 &-1 &-1 &-1 &-1 &-1 &33 &-1 &50 &-1 &-1 &29 &-1 &-1 &-1\\
 &-1 &-1 &-1 &-1 &-1 &-1 &-1 &-1 &-1 &-1 &-1 &-1 &-1 &-1 &-1 &-1 &-1 &-1 &-1 &-1 &-1 &-1 &-1 &-1 &-1 &-1 &-1 &-1 &-1 &-1 &-1 &-1 &33 &-1 &-1 &-1 &-1 &-1 &50 &-1 &-1 &29 &-1 &-1\\
 &-1 &-1 &-1 &-1 &-1 &-1 &-1 &-1 &-1 &-1 &-1 &-1 &-1 &-1 &-1 &-1 &-1 &-1 &-1 &-1 &-1 &-1 &-1 &-1 &-1 &-1 &-1 &-1 &-1 &-1 &-1 &-1 &-1 &33 &-1 &-1 &-1 &-1 &-1 &50 &-1 &-1 &71 &-1\\
 &-1 &-1 &-1 &-1 &-1 &-1 &-1 &-1 &-1 &-1 &-1 &-1 &-1 &-1 &-1 &-1 &-1 &-1 &-1 &-1 &-1 &-1 &-1 &-1 &-1 &-1 &-1 &-1 &-1 &-1 &-1 &-1 &-1 &-1 &-1 &-1 &-1 &-1 &-1 &31 &-1 &66 &82 &-1\\
 &-1 &-1 &-1 &-1 &-1 &-1 &-1 &-1 &-1 &-1 &-1 &-1 &-1 &-1 &-1 &-1 &-1 &-1 &-1 &-1 &-1 &-1 &-1 &-1 &-1 &-1 &-1 &-1 &-1 &-1 &-1 &-1 &-1 &-1 &-1 &-1 &31 &-1 &-1 &-1 &-1 &-1 &35 &82\\
 &-1 &-1 &-1 &-1 &-1 &-1 &-1 &-1 &-1 &-1 &-1 &-1 &-1 &-1 &-1 &-1 &-1 &-1 &-1 &-1 &-1 &-1 &-1 &-1 &-1 &-1 &-1 &-1 &-1 &-1 &-1 &-1 &-1 &-1 &-1 &-1 &-1 &31 &-1 &-1 &25 &-1 &-1 &35\\
 &-1 &-1 &-1 &-1 &-1 &-1 &-1 &-1 &-1 &-1 &-1 &-1 &-1 &-1 &-1 &-1 &-1 &-1 &-1 &-1 &-1 &-1 &-1 &-1 &-1 &-1 &-1 &-1 &-1 &-1 &-1 &-1 &-1 &-1 &-1 &-1 &-1 &-1 &31 &-1 &66 &25 &-1 &-1\\
 &-1 &-1 &-1 &-1 &-1 &-1 &-1 &-1 &-1 &-1 &-1 &-1 &-1 &-1 &-1 &-1 &-1 &-1 &-1 &-1 &-1 &-1 &-1 &-1 &-1 &-1 &-1 &-1 &-1 &-1 &-1 &-1 &-1 &-1 &-1 &-1 &-1 &-1 &-1 &19 &-1 &95 &47 &-1\\
 &-1 &-1 &-1 &-1 &-1 &-1 &-1 &-1 &-1 &-1 &-1 &-1 &-1 &-1 &-1 &-1 &-1 &-1 &-1 &-1 &-1 &-1 &-1 &-1 &-1 &-1 &-1 &-1 &-1 &-1 &-1 &-1 &-1 &-1 &-1 &-1 &19 &-1 &-1 &-1 &-1 &-1 &17 &47\\
 &-1 &-1 &-1 &-1 &-1 &-1 &-1 &-1 &-1 &-1 &-1 &-1 &-1 &-1 &-1 &-1 &-1 &-1 &-1 &-1 &-1 &-1 &-1 &-1 &-1 &-1 &-1 &-1 &-1 &-1 &-1 &-1 &-1 &-1 &-1 &-1 &-1 &19 &-1 &-1 & 3 &-1 &-1 &17\\
 &-1 &-1 &-1 &-1 &-1 &-1 &-1 &-1 &-1 &-1 &-1 &-1 &-1 &-1 &-1 &-1 &-1 &-1 &-1 &-1 &-1 &-1 &-1 &-1 &-1 &-1 &-1 &-1 &-1 &-1 &-1 &-1 &-1 &-1 &-1 &-1 &-1 &-1 &19 &-1 &95 & 3 &-1 &-1
 \end{smallmatrix} \right]. \nonumber
\end{equation}
% Restore the current equation number.
\setcounter{equation}{\value{mytempeqncnt}}
% IEEE uses as a separator
%\hrulefill
% The spacer can be tweaked to stop underfull vboxes.
\vspace*{4pt}
\end{figure*}

\begin{figure*}[!t]
% ensure that we have normalsize text
\normalsize
% Store the current equation number.
\setcounter{mytempeqncnt}{\value{equation}}
% Set the equation number to one less than the one
% desired for the first equation here.
% The value here will have to changed if equations
% are added or removed prior to the place these
% equations are referenced in the main text.
\setcounter{equation}{5}
\begin{equation}
B_2 = \left[ \begin{smallmatrix}
&-1&-1&748&-1&-1&3&-1&-1&-1&-1&-1&-1&-1&-1&-1&-1&-1&-1&-1&-1&-1&-1&-1&-1\\
&-1&-1&-1&748&-1&-1&3&-1&-1&-1&-1&-1&-1&-1&-1&-1&-1&-1&-1&-1&-1&-1&-1&-1\\
&748&-1&-1&-1&-1&-1&-1&3&-1&-1&-1&-1&-1&-1&-1&-1&-1&-1&-1&-1&-1&-1&-1&-1\\
&-1&748&-1&-1&3&-1&-1&-1&-1&-1&-1&-1&-1&-1&-1&-1&-1&-1&-1&-1&-1&-1&-1&-1\\
&-1&127&-1&-1&-1&941&-1&-1&-1&855&-1&-1&808&-1&-1&-1&-1&-1&-1&-1&-1&-1&-1&-1\\
&-1&-1&127&-1&-1&-1&941&-1&-1&-1&855&-1&-1&808&-1&-1&-1&-1&-1&-1&-1&-1&-1&-1\\
&-1&-1&-1&127&-1&-1&-1&941&-1&-1&-1&855&-1&-1&808&-1&-1&-1&-1&-1&-1&-1&-1&-1\\
&127&-1&-1&-1&941&-1&-1&-1&855&-1&-1&-1&-1&-1&-1&808&-1&-1&-1&-1&-1&-1&-1&-1\\
&382&-1&-1&-1&-1&-1&-1&5&-1&228&-1&-1&366&-1&-1&-1&-1&242&-1&-1&-1&-1&-1&683\\
&-1&382&-1&-1&5&-1&-1&-1&-1&-1&228&-1&-1&366&-1&-1&-1&-1&242&-1&683&-1&-1&-1\\
&-1&-1&382&-1&-1&5&-1&-1&-1&-1&-1&228&-1&-1&366&-1&-1&-1&-1&242&-1&683&-1&-1\\
&-1&-1&-1&382&-1&-1&5&-1&228&-1&-1&-1&-1&-1&-1&366&242&-1&-1&-1&-1&-1&683&-1\\
&-1&-1&391&-1&-1&-1&-1&246&-1&920&728&-1&952&-1&602&-1&295&621&847&-1&506&884&-1&302\\
&-1&-1&-1&391&246&-1&-1&-1&-1&-1&920&728&-1&952&-1&602&-1&295&621&847&302&506&884&-1\\
&22&-1&-1&-1&-1&624&-1&-1&884&-1&-1&473&622&-1&982&-1&644&-1&794&598&-1&614&511&213\\
&-1&22&-1&-1&-1&-1&624&-1&473&884&-1&-1&-1&622&-1&982&598&644&-1&794&213&-1&614&511
 \end{smallmatrix} \right].\nonumber
\label{eq.code2BM}
\end{equation}
% Restore the current equation number.
\setcounter{equation}{\value{mytempeqncnt}}
% IEEE uses as a separator
\hrulefill
% The spacer can be tweaked to stop underfull vboxes.
\vspace*{4pt}
\end{figure*}

\bibliographystyle{IEEE}
%%%%%\bibliography{bib-file}  % commented if *.bbl file included, as
%%%%%see below

%% This is nothing else than the IEEEsample.bbl file that you would
%%
%% obtain with BibTeX: you do not need to send around the *.bbl file
%%
%%---------------------------------------------------------------------------%%
%

%
%%---------------------------------------------------------------------------%%

\end{document}